\documentclass[11pt]{article}
\usepackage{comment}
\usepackage{amsmath}
\usepackage{amsthm}
\usepackage{amssymb}
\usepackage{url}
\usepackage[final]{showkeys}
\usepackage[usenames,dvipsnames]{xcolor}
\usepackage{graphicx}
\usepackage[normalem]{ulem}

\usepackage{titlesec}

\usepackage{accents}

\usepackage{bm}

\titlespacing*{\section}
{0pt}{10pt plus 1pt minus 1pt}{8pt plus 1pt minus 1pt}
\titlespacing*{\subsection}
{0pt}{8pt plus 1pt minus 1pt}{6pt plus 1pt minus 1pt}

\usepackage{fullpage}

\usepackage{stmaryrd}

\usepackage[colorlinks=true, allcolors=blue]{hyperref}

\usepackage{latexsym}
\usepackage{booktabs}
\usepackage{enumitem}
\usepackage{bm}

\usepackage{mathtools}

\usepackage{array} 

\ExplSyntaxOn

\AddToHook{cmd/bfseries/after}{\boldmath}
\AddToHook{cmd/normalfont/after}{\unboldmath}

\NewDocumentCommand{\overunderline}{mmm}
 {
  \group_begin:
  \setlength{\arraycolsep}{0pt}
  \renewcommand{\arraystretch}{0}
  \setlength{\belowrulesep}{\aboverulesep}
  \setlength{\cmidrulekern}{0.1em}
  \ziggy_overunderline:nnn { #1 } { #2 } { #3 }
  \group_end:
 }

\seq_new:N \l__ziggy_overunderline_over_seq
\seq_new:N \l__ziggy_overunderline_under_seq
\seq_new:N \l__ziggy_overunderline_digits_seq
\seq_new:N \l__ziggy_overunderline_phantoms_seq

\cs_new_protected:Nn \ziggy_overunderline:nnn
 {
  \seq_set_split:Nnn \l__ziggy_overunderline_digits_seq { } { #1 }
  \seq_set_map:NNn \l__ziggy_overunderline_phantoms_seq \l__ziggy_overunderline_digits_seq
   {
    \phantom{##1}
   }
  \clist_map_inline:nn { #2 }
   {
    \seq_put_right:Nn \l__ziggy_overunderline_over_seq { \cmidrule(lr){##1} }
   }
  \clist_map_inline:nn { #3 }
   {
    \seq_put_right:Nn \l__ziggy_overunderline_under_seq { \cmidrule(lr){##1} }
   }
  \hbox_to_zero:n
   {
    $\begin{array}[b]{*{\seq_count:N \l__ziggy_overunderline_digits_seq}{c}}
    \seq_use:Nn \l__ziggy_overunderline_over_seq {}
    \seq_use:Nn \l__ziggy_overunderline_digits_seq { & } \\
    \end{array}$\hss
   }
   \begin{array}[t]{*{\seq_count:N \l__ziggy_overunderline_digits_seq}{c}}
   \seq_use:Nn \l__ziggy_overunderline_phantoms_seq { & } \\
   \seq_use:Nn \l__ziggy_overunderline_under_seq {}
   \end{array}
 }

\ExplSyntaxOff

\usepackage[vlined]{algorithm2e}
\DontPrintSemicolon
\SetVlineSkip{1pt}
\NoCaptionOfAlgo
\SetKwIF{If}{ElseIf}{Else}{if}{:}{else if}{else:}{}%
\SetKwFor{While}{while}{:}{}%
\SetKwProg{Fn}{}{:}{}
\SetKwFor{RepTimes}{repeat}{times}{}

\newtheorem{theorem}{Theorem}[section]
\newtheorem{lemma}[theorem]{Lemma}

\newtheorem{corollary}[theorem]{Corollary}
\newtheorem{definition}[theorem]{Definition}

\newtheorem{remark}{Remark}

\newtheorem{invariant}{Invariant}

\newcommand{\ubar}[1]{\underaccent{\bar}{{#1}}}

\newcommand{\cev}[1]{\reflectbox{\ensuremath{\vec{\reflectbox{\ensuremath{#1}}}}}}

\makeatletter
\DeclareRobustCommand{\cev}[1]{%
  {\mathpalette\do@cev{#1}}%
}
\newcommand{\do@cev}[2]{%
  \vbox{\offinterlineskip
    \sbox\z@{$\m@th#1 x$}%
    \ialign{##\cr
      \hidewidth\reflectbox{$\m@th#1\vec{}\mkern4mu$}\hidewidth\cr
      \noalign{\kern-\ht\z@}
      $\m@th#1#2$\cr
    }%
  }%
}
\makeatother

\newcommand{\RR}{\mathbb{R}}

\newcommand{\Concat}{\mbox{{\it Concatenate}}}
\newcommand{\CO}{\mbox{{\it Concatenate}-{\it Opposite}}}
\newcommand{\LS}{\mbox{{\it Long}-{\it Shortcuts}}}
\newcommand{\Simple}{\mbox{{\it Short}-{\it Shortcuts}}}
\newcommand{\Trivial}{\mbox{{\it Trivial}-{\it Shortcuts}}}
\newcommand{\BFS}{\mbox{{\it Breadth}-{\it Search}}}
\newcommand{\ComputeF}{\mbox{{\it Compute}-{\it Funnels}}}
\newcommand{\ComputeS}{\mbox{{\it Compute}-{\it Shortcuts}}}
\newcommand{\UpdateS}{\mbox{{\it Update}-{\it Shortcuts}}}
\newcommand{\Extend}{\mbox{{\it Arc}-{\it Bounded}-{\it To}-{\it Monotone}}}

\newcommand{\argmax}{\operatornamewithlimits{argmax}}

\newcommand{\dani}[1]{{\color{ForestGreen}(dani: #1)}}

\newcommand\numeq[1]{\stackrel{\scriptscriptstyle(\mkern-1.5mu#1\mkern-1.5mu)}{=}}
\newcommand\numge[1]{\stackrel{\scriptscriptstyle(\mkern-1.5mu#1\mkern-1.5mu)}{\ge}}
\newcommand\numle[1]{\stackrel{\scriptscriptstyle(\mkern-1.5mu#1\mkern-1.5mu)}{\le}}

\begin{document}

\title{\null\vspace*{-60pt}Faster All-Pairs Optimal Electric Car Routing}

\author{%
  \begin{tabular}{ c c c c c}
       Dani Dorfman\thanks{Max Planck Institute for Informatics, Saarbr\"ucken
, Germany. Email: {\tt dani.i.dorfman@gmail.com}.}    &    Haim Kaplan\thanks{{\tt \{haimk,zwick\}@tau.ac.il.} Work of Uri Zwick partially supported by grant 2854/20 of the Israeli Science Foundation. Work of Haim Kaplan partially supported by ISF grant 1595/19 and the Blavatnik family foundation.}   &    Robert E.\ Tarjan\thanks{Department of Computer Science, Princeton University. Research partially supported by a gift from Microsoft. Email: {\tt ret@princeton.edu.}}\;  
       & Mikkel Thorup\thanks{BARC, University of Copenhagen, Denmark. Research supported by the VILLUM Foundation grant no.\ 16582. Email: {\tt mikkel2thorup@gmail.com}}\;    &   Uri Zwick${}^{\dag}$        
  \end{tabular}
}
\date{}
\maketitle

\begin{abstract}
\setlength{\parindent}{0pt}
\setlength{\parskip}{3pt plus 2pt}\noindent%

We present a randomized $\tilde{O}(n^{3.5})$-time algorithm for computing \emph{optimal energetic paths} for an electric car between all pairs of vertices in an $n$-vertex directed graph with positive and negative \emph{costs}, or \emph{gains}, which are defined to be the negatives of the costs. The optimal energetic paths are finite and well-defined even if the graph contains negative-cost, or equivalently, positive-gain, cycles. This makes the problem much more challenging than standard shortest paths problems. 

More specifically, for every two vertices $s$ and~$t$ in the graph, the algorithm computes $\alpha_B(s,t)$, the maximum amount of charge the car can reach~$t$ with, if it starts at~$s$ with full battery, i.e., with charge~$B$, where~$B$ is the capacity of the battery. The algorithm also outputs a concise description of the optimal energetic paths that achieve these values. In the presence of positive-gain cycles, optimal paths are not necessarily simple. For dense graphs,
our new $\tilde{O}(n^{3.5})$ time algorithm improves on a previous $\tilde{O}(mn^{2})$-time algorithm of Dorfman et al. [ESA 2023] for the problem. 

The \emph{gain} of an arc is the amount of charge added to the battery of the car when traversing the arc. The charge in the battery can never exceed the capacity~$B$ of the battery and can never be negative. An arc of positive gain may correspond, for example, to a downhill road segment, while an arc with a negative gain may correspond to an uphill segment. A positive-gain cycle, if one exists, can be used in certain cases to charge the battery to its capacity. This makes the problem more interesting and more challenging. As mentioned, optimal energetic paths are well-defined even in the presence of positive-gain cycles. Positive-gain cycles may arise when certain road segments have magnetic charging strips, or when the electric car has solar panels.

Combined with a result of Dorfman et al. [SOSA 2024], this also provides a randomized $\tilde{O}(n^{3.5})$-time algorithm for computing \emph{minimum-cost paths} between all pairs of vertices in an $n$-vertex graph when the battery can be externally recharged, at varying costs, at intermediate vertices.
\end{abstract}

\section{Introduction}\label{S:intro}

Let $G=(V,A,c)$ be a weighted directed graph, where $V$ is the set of vertices,  $A\subseteq V\times V$ is the set of arcs, and where $c:A\to\RR$ is a real-valued \emph{cost} function defined on the arcs. The cost $c(uv)$ of an arc $uv\in A$ \footnote{For brevity we denote an arc from~$u$ to~$v$ by $uv$, rather than $(u,v)$.} is the amount of energy consumed when traversing the arc. Throughout most of this paper, it is more convenient to work with a gain function $g:A\to\RR$ rather than a cost function. The \emph{gain} $g(uv)$ of an arc $uv\in A$ is simply $g(uv)=-c(uv)$, i.e., the amount of energy gained by traversing the arc.
The gain $g(uv)$ is negative if moving from~$u$ to~$v$ requires spending energy, or positive if energy is gained by moving from~$u$ to~$v$.

A weighted directed graph $G=(V,A,g)$, where $g:A\to\RR$ is a gain function, may be viewed as modeling a road network on which an electric car can roam. The electric car is assumed to have a battery of \emph{capacity}~$B$, where $B>0$ is a parameter, i.e., it can store up to~$B$ units of energy. The \emph{charge}, i.e., the amount of energy in the battery, can never be negative, and can never exceed the capacity of the battery. If the car is currently at vertex~$u$ with charge~$b$ in its battery, where $0\le b\le B$, then it can traverse an arc $uv\in A$ if and only if $b+g(uv)\ge 0$. If this condition holds, and the car traverses the arc, then it reaches~$v$ with a charge of $\min\{b+g(uv),B\}$. The car can traverse~$uv$ if $b+g(uv)>B$, but the battery does not charge beyond its capacity of~$B$. The car can traverse a path if and only if it can sequentially traverse its arcs. Throughout most of the paper we assume that no external charging of the battery is allowed. The battery is only charged by traversing arcs with positive gain. We may assume that $g(uv)\in[-B,B]$, for every $uv\in A$, as arcs with $g(uv)<-B$ can never be used, and can thus be removed, and gains $g(uv)>B$ can be changed to $g(uv)=B$ without changing the problem.

We consider the following two related natural questions:

\begin{enumerate}[topsep=3pt,left=15pt,itemsep=3pt,parsep=0pt]
    \item Given two vertices $s,t\in V$, what is the \emph{maximum final charge}, denoted $\alpha_B(s,t)$, with which the car can reach~$t$ if it starts at~$s$ with full battery, i.e., with a charge of~$B$? If the car cannot reach~$t$ even with an initial charge of~$B$ at~$s$, we let $\alpha_B(s,t)=-\infty$. More generally, we let $\alpha_b(s,t)$, where $0\le b\le B$, be the maximum final charge with which the car can reach~$t$ if it starts at~$s$ with a charge of~$b$.
    
    \item Given two vertices $s,t\in V$, what is the \emph{minimum initial charge} at~$s$, denoted $\beta_0(s,t)$, that enables the car to reach~$t$? If the car cannot reach $t$ even with an initial charge of~$B$ at~$s$, we let $\beta_0(s,t)=\infty$. More generally, we let $\beta_b(s,t)$, where $0\le b\le B$, be the minimum initial charge at~$s$ required for reaching~$t$ with a charge of at least~$b$.
\end{enumerate}

It is not difficult to see, as shown in Dorfman et al. \cite[Corollary 5.2]{DorfmanKTZ23}, that $\beta_0(s,t)=B-\cev{\alpha}_B(t,s)$, where $\cev{\alpha}_B(t,s)$ denotes the maximum final charge at~$s$ when starting at~$t$ with full battery in the \emph{reverse} of the graph. Thus, the problems of computing maximal final charges and minimum initial charges are computationally equivalent. (Note, however, that due to the reverse operation used, the single-source version of the maximum final charge problem becomes equivalent to the single-target version of the minimum initial charge problem.) In this paper, we only work with maximal final charges.

If all arc costs are nonnegative, i.e., all gains are nonpositive, then it is easy to see that $\beta_0(s,t)=\delta(s,t)$, and $\alpha_B(s,t)=B-\delta(s,t)$, if $\delta(s,t)\le B$, where $\delta(s,t)$ is the standard distance from~$s$ to~$t$ with respect to the costs of the arcs. Otherwise, $\beta_0(s,t)=\infty$ and $\alpha_B(s,t)=-\infty$. When costs and gains can be both positive and negatives, the problem becomes more complicated. If there are no positive-gain cycles in the graph, the problem can be solved using fairly simple adaptations of standard shortest paths algorithms. Thus, the single-source version of the maximal final charges problem can be solved in $O(mn)$ time using an adaptation of the classical Bellman-Ford algorithm \cite{Bellman58,Ford56}, and the all-pairs version of the problem can be solved in $O(mn+n^2\log n)$ time by an adaptation of the classical algorithm of Johnson \cite{Johnson77}. For these results see, Artmeier, Haselmayr, Leucker and Sachenbacher \cite{artmeier2010shortest}, Eisner, Funke and Storandt \cite{EFS11}, Brim and Chaloupka \cite{BrCh12}, and Dorfman, Kaplan, Tarjan and Zwick \cite{DorfmanKTZ23}. 

The problem becomes much harder when the graph may contain positive-gain cycles. Part of the difficulty is that optimal paths, which are still well-defined, are not necessarily simple and might have to `hop' from one positive-gain cycle to another, until gaining enough charge to head directly to the destination. (See Lemma~\ref{lemma:optimal-structure} below.)
H\'{e}lou\"{e}t et al.~\cite{helouet2019reachability} obtained a polynomial time algorithm for the decision problem of determining whether $\beta_0(s,t) \le B$.
Dorfman et al.~\cite{DorfmanKTZ23} obtained an $O(mn+n^2\log n)$-time algorithm for the single-source version of the problem, which of course implies an $O(mn^2+n^3\log n)$-time algorithm for the all-pairs version. 

Our main result is a randomized $\tilde{O}(n^{3.5})$-time\footnote{The $\tilde{O}(\cdot)$ hides logarithmic factors.} algorithm for solving the all-pairs versions of the maximal final charge problem, and hence also the minimum initial charge problem, improving by a $\Theta(\sqrt{n})$ factor for sufficiently dense graphs on the $O(mn^2+n^3\log n)$ running time of the algorithm of Dorfman et al.~\cite{DorfmanKTZ23}. To appreciate our result, we draw a parallel to standard shortest paths. On a graph with $n$ nodes and $m$ edges (and a suitable potential function), the single source shortest path problem can be solved in $O(m + n\log n)$ time, leading to an $O(mn+n^2 \log n) = O(n^3)$ all pairs algorithm for dense graphs. A breakthrough result by Williams~\cite{williams2014faster} achieved an $O(\frac{n^3}{2^{c\sqrt{\log n}}})$ time all pairs algorithm, shaving a subpolynomial factor for dense graphs.

All the discussion so far assumed that that battery cannot be recharged at intermediate vertices. A natural variant is obtained when we assume that the battery can be charged at some of the vertices of the graph, with a cost per unit of charge that may vary from vertex to vertex. The goal then, is to find \emph{minimum-cost paths} within all pairs of vertices in the graph. This problem was considered by Khuller, Malekian and Mestre \cite{khuller2011fill} in the context of conventional, gas-operated, cars, i.e., when all arc costs are positive, and by Dorfman, Kaplan, Tarjan, Thorup and Zwick \cite{DKTTZ24} in the context of electric cars, i.e., when the costs, or gains, can be both positive and negative, and where there might be positive-gain cycles. The main result of Dorfman et al.~\cite{DKTTZ24} is a reduction from the all-pairs minimum-cost paths problem to the all-pairs maximal final charges and minimum initial charges problems, and to the standard all-pairs shortest paths problem. Combined with the results of Dorfman et al.~\cite{DorfmanKTZ23}, this implies an $O(mn+n^2\log n)$-time algorithm for the all-pairs minimum-cost paths in graphs with no positive-gain cycles. Combined with our result, we obtain a randomized $\tilde{O}(n^{3.5})$-time algorithm for the all-pairs minimum-cost paths in graphs that may contain positive-gain cycles.

To obtain the improved algorithm we need to introduce many new ideas. We next try to give a rough intuitive description of some of them, ignoring some technicalities that will be dealt with later.

Optimal energetic paths can be very long. (Their length cannot be bounded as a function of~$n$ alone. A bound must also take the arc gains and the capacity of the battery into account. For more details, see \cite{DorfmanKTZ23}.) A natural idea to reduce the length of optimal energetic paths is to introduce \emph{shortcuts}, i.e., add new arcs that correspond to possibly long paths in the graph. In the standard shortest paths problem, any path in the graph can be used to generate a shortcut, with the gain (or cost) of the arc equal to the sum of the gains of the arcs on the path. This is far from being the case for energetic paths. Consider, for example, a path $xyz$ with $g(xy)=-1$ and $g(yz)=1$. We cannot add a new arc $xz$ with $g(xz)=0$ to the graph since an electric car with an empty battery would be able to traverse the new arc $xz$, but not the original path $xyz$.

Ignoring some technicalities, we can add a shortcut corresponding to a traversable path in the graph (a path that can be traversed if we start with full battery) if the path is \emph{ascending} or \emph{descending}. (See Figure~\ref{fig:example_funnel_monotone}(a)-(b).) Given a path $u_0u_1\ldots u_{k}$, let $a_i=\sum_{j=0}^{i-1} g(u_ju_{j+1})$, for $0\le j\le k$, be the prefix sums of the gains along the path. We say that a path is \emph{ascending} if $0\le a_i\le a_k$, for every $1\le i\le k$, and descending if $a_k\le a_i\le 0$, for every $1\le i\le k$. If $u_0u_1\ldots u_{k}$ is ascending or descending, then we are allowed to add a shortcut $u_0u_k$ with $g(u_0u_k)=a_k$. For brevity, we refer to ascending or descending paths as \emph{monotone}.\footnote{Note that this does not imply that the prefix sums $a_1,a_2,\ldots,a_k$ form a monotone sequence.}

Unfortunately, most paths are not monotone. Furthermore, subpaths of monotone paths are not necessarily monotone. 
There may also be very long paths that do not contain any monotone subpath. We refer to such paths as \emph{funnels}. Examples of funnels are given in Figure~\ref{fig:example_funnel_monotone}$(c)$-$(f)$.

Our algorithm constructs monotone paths and funnels and combines them to obtain new monotone paths and funnels until enough information is available to find the optimal energetic paths. The exact details, some of which are quite delicate, appear in the rest of the paper.

Another idea used by our new algorithm is \emph{sampling}. It is well known that a random set of vertices of size $(cn\log n)/k$ is likely to hit any given path of length at least~$k$. Taking advantage of this fact in our context is again much more complicated.

The rest of this extended abstract consists of a technical review of our algorithm and main techniques. The full version of the paper is given in the appendix.

\begin{figure}[t]
    \centering
    \includegraphics[width=1.00\textwidth]{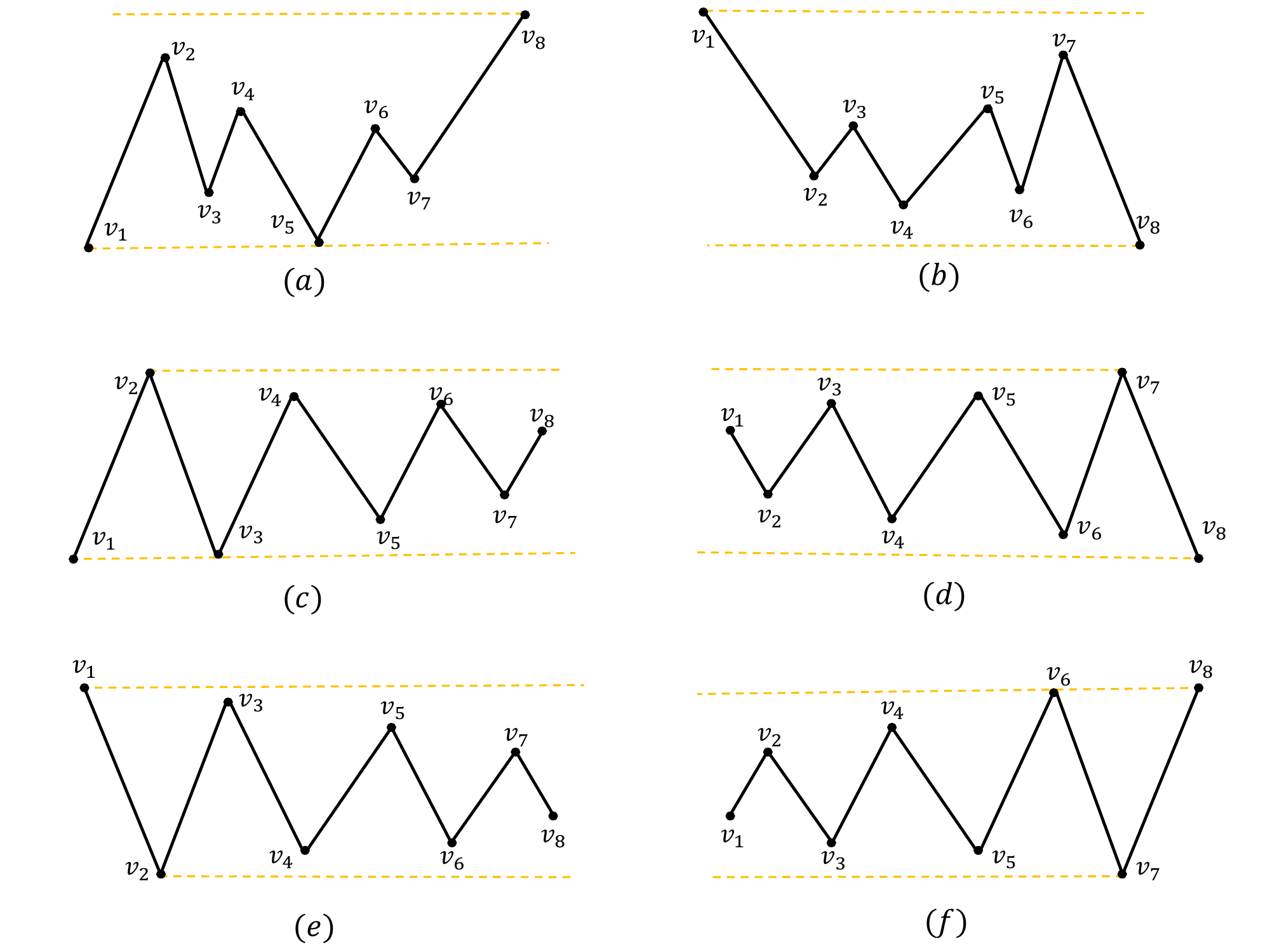}
    \caption{
    The graphs represent directed paths going from left to right. The vertical height of an arc $e$ in the figure is $|g(e)|$. 
    The vertical height of a vertex is its gain on the path (i.e. sum of arc gains). 
    Figures $(a)$ and $(b)$ show an ascending path and a descending path, respectively. 
    Figures~$(c)$-$(f)$ show the four possible cases for funnels. Note that in Figure~$(c)$, $v_3$ (which is the endpoint of the second arc of the funnel) has the same gain as $v_1$, this is valid.}
    \label{fig:example_funnel_monotone}
\end{figure}

\section{Technical Review}\label{sec:technical-review}

A main tool in our algorithm is \emph{shortcutting}. In the setting of standard shortest paths, any path $P=v_1\ldots v_k$ can be shortcutted to a single arc $v_1 v_k$ of gain $g(P)= \sum_{i=1}^{k-1} g(v_i v_{i+1})$ without affecting the lengths of the shortest paths. Unfortunately, because of the upper and lower bound constraints on the battery, this technique breaks down when applied to energetic paths. That is, by shortcutting arbitrary paths, we may change the optimal energetic paths. For example, assume $B = 10$ and let $G$ be a graph that is composed of two paths $P_1=v_1 v_2 v_3$ and $P_2 = u_1 u_2 u_3$, where $g(v_1 v_2) = g(u_2 u_3) = -5$ and $g(v_2 v_3) = g(u_1 u_2)=5$. 
Observe that $\alpha_0(v_1, v_3) = -\infty$ and $\alpha_{10}(u_1, u_3) = 5$. 
On the other hand, by shortcutting the paths $v_1 v_2 v_3$ and $u_1 u_2 u_3$ (to arcs of gain $0$) we will be able to reach $v_3$ from $v_1$ when starting with zero charge. Moreover by using the new $0$ gain arc $u_1 u_3$ the maximum final charge at $u_3$ (when starting with $10$ charge at $u_1$) becomes $10$.

The above discussion encourages us to find \emph{safe} paths that can be shortcutted without affecting the optimal energetic paths (i.e., without affecting the $\alpha$ values). We call these paths \emph{monotone paths}, see Definition~\ref{def:monotone-path} and Figure~\ref{fig:example_funnel_monotone}.  
Monotone paths are either \emph{ascending} or \emph{descending}. An ascending path $P = v_1\ldots v_k$ is a \emph{traversable} path that satisfies that whenever an electric car traverses $P$,  the car has minimum charge at $v_1$ and maximum charge at $v_k$. A traversable path $P = v_1 \ldots v_k$ is a path that does not contain any subpath $v_i \ldots v_j$ of gain smaller than $-B$ (this is equivalent to saying that a car that starts at $v_1$ with full charge can traverse $P$ without the charge level going below zero).
Similarly, a descending path $P = v_1\ldots v_k$ is a traversable path that satisfies that whenever an electric car traverses $P$,  the car has max charge at $v_1$ and minimum charge at $v_k$ (in particular, the gain of a descending path is at least $-B$).
A monotone path avoids the two problems mentioned in the previous example: The charge level of an ascending path never drops below the charge level at $v_1$ and therefore the path $v_1 v_2 v_3$ from the previous example cannot be shortcutted. Moreover, since the charge level remains below the charge level at $v_k$, shortcutting $P$ does not create an alternative path from $v_1$ to $v_k$ that improves the final charge at $v_k$, similarly to what happened with the path $u_1 u_2 u_3$ from the previous example. 

We prove in Theorem~\ref{theorem:shortcut} that in $\tilde{O}(n^{3.5})$ time we can compute a $2$-dimensional table $M[\cdot][\cdot]$ that \emph{dominates} all simple monotone paths. That is, for every simple monotone path $P = v_1\ldots v_k$, it holds that $M[v_1][v_k] \ge g(P)$. Moreover, the table $M$ is \emph{sound}. That is, for every $u,v\in V$, if $M[u][v] \neq -\infty$, then there exists a monotone path $P$ (not necessarily simple) from $u$ to $v$ such that $g(P) \ge M[u][v]$.
Since monotone paths are traversable, it follows that if $M[u][v] \neq -\infty$, then $M[u][v] \ge -B$. Note that it is possible that $M[u][v] > B$.
Once we have computed $M$, solving the all pairs $\alpha_B(\cdot,\cdot)$ problem is rather simple, we explain this derivation at the end of the technical review. 

The following is a high level description of the computation of $M$. For simplicity, in this short review, we only describe how to dominate \emph{ascending} paths. A simple observation is that every monotone path $P$ contains a monotone subpath of edge-length $2$ or $3$. We call such a path a \emph{short} monotone path. Thus, by shortcutting such a short monotone subpath into a single arc, we get an ascending path $P'$ of smaller length than $P$ and larger or equal gain than $g(P)$. This observation leads to a trivial $\tilde{O}(n^4)$ algorithm: Perform $n$ iterations and generate a series of graphs $G_0=G,G_1,\ldots, G_n$. In the $i$'th iteration we find for every $u,v$ the largest gain short monotone path from $u$ to $v$ in $G_i$. Once we have found all such gains, we build $G_i$ by increasing the gains of every arc\footnote{We assume $G_{i-1}$ is a full graph by adding arcs of gain $-\infty$.} 
$(u,v)$
in $G_{i-1}$ if there is a corresponding short monotone path from $u$ to $v$ of a better gain. We can implement each iteration in $\tilde{O}(n^3)$ time using a BST data structure. The table $M$ stores the gains of the arcs of final graph $G_n$.
Given a simple ascending path $P$ in $G_0$, this process implicitly constructs a series of paths $P_i \in G_i$, where $P_i$ is obtained from $P_{i-1}$ by shortcutting as many short monotone paths as possible and $P_n$ is a single arc.\footnote{Note that $P_i$ is not uniquely defined since short monotone paths may overlap. Moreover, we might perform shortcuts in $P_{i-1}$ because of short monotone paths that appear in $G_{i-1}$ and not in $P_{i-1}$.}

An immediate question is whether $\Theta(n)$ iterations are necessary. The answer is yes. The reason for this are \emph{double-funnels} (see Figure~\ref{fig:double-funnel_mono}(a)). A path $P = v_1 \ldots v_k$ is a \emph{double-funnel} if $P$ does not contain a short monotone subpath.
Double-funnels can have $\Theta(n)$ edges and an ascending monotone path which consists mainly of a long double-funnel would require $\Theta(n)$ iterations to be shortcutted into a single arc, see Figure~\ref{fig:double-funnel_mono}(b).

\begin{figure}
    \centering
    \includegraphics[width=1\linewidth]{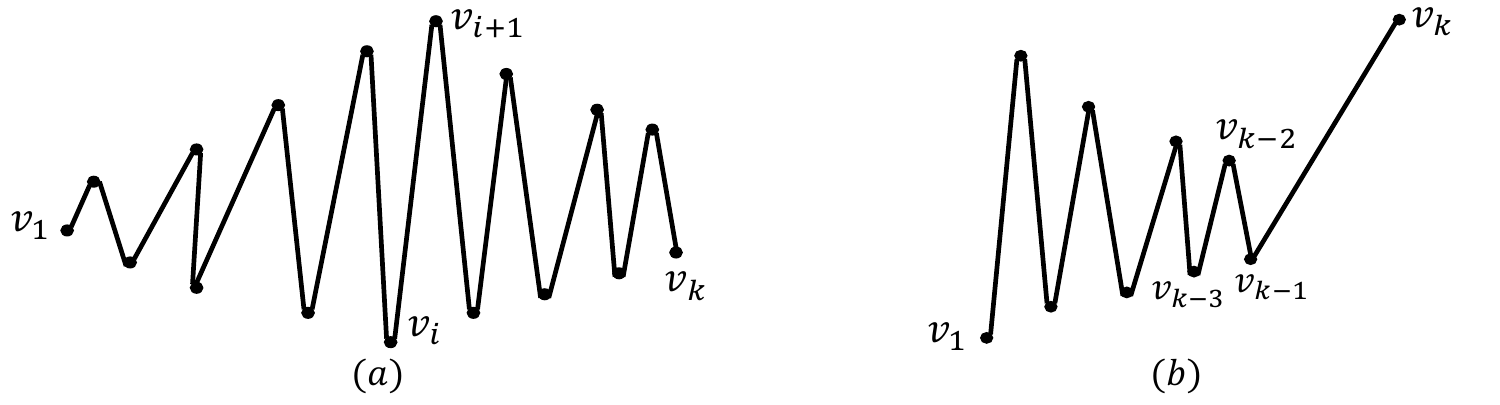}
    \caption{On the left: a double-funnel. On the right: worst case example for the simple algorithm. The depicted (directed) path $P=v_1\ldots v_k$ is monotone. Since $v_1\ldots v_{k-1}$ is a double-funnel, the only short monotone subpath of $P$ is $v_{k-3} v_{k-2} v_{k-1} v_{k}$. Assume $G$ is a path graph that contains only the path $P$. After the first iteration of shortcutting short monotone paths, we are left with the path $P_1 = v_1\ldots v_{v-3} v_k$ that has a similar structure to $P$. Thus, $\lfloor \frac{k}{2} \rfloor$ iterations are necessary in order to shortcut $P$ into a single arc.}
    \label{fig:double-funnel_mono}
\end{figure}

As a consequence of the  discussion above, in order to improve upon the simple algorithm, we need to handle double-funnels and reduce the number of iterations.  A simple observation is that every ascending path can be viewed as an alternation between double-funnels (that are maximal with respect to inclusion) and short monotone paths, see Figure~\ref{fig:alternation-funnel-shortcut}. Indeed, by the definition of a double-funnel, if we extend a double-funnel that is maximal with respect to inclusion by a single arc, the path  ceases to be a double-funnel and therefore contains a short monotone path.

Let $P$ be an ascending path such that $P$ is not shortcutted to a single arc after $T = \sqrt{n}$ iterations of the simple algorithm.
Let $P_0,\ldots, P_T$ (paths in $G_0,\ldots, G_T$, respectively) be the corresponding sequence of ascending paths as we defined before. For every $i=0,\ldots T$, denote by $f_i$ the number of (maximal with respect to inclusion) double-funnels in $P_i$.
By the interleaving property of double-funnels and short monotone paths, for every $i=0,\ldots, T-1$, the number of short-monotone subpaths in $P_i$ is at least $f_i$ and therefore $|P_{i+1}| \le |P_i| - f_i$ (where $|Q|$ denotes the number of arcs in a path $Q$). Since $P=P_0$ is a simple path (and thus of length at most $n-1)$, and since we can uniquely charge a  short monotone path that we shortcut at iteration $i$ to each funnel in $P_i$ it follows that $\sum_{i=1}^T f_i < n$, so the average number of funnels per iteration (of the $T$ iterations that we consider) satisfies $\frac{1}{T}\sum_{i=1}^T f_i = \frac{1}{\sqrt{n}}\sum_{i=1}^{\sqrt{n}} f_i < \sqrt{n}$. 
By Markov's inequality, in at least $\frac{1}{2}T=  \frac{1}{2}\sqrt{n}$ iterations, $f_i \le 2\sqrt{n}$. Thus, in at least half of the $T$ iterations, the paths $P_i$ have $O(\sqrt{n})$ double-funnels. By sampling uniformly at random $\Theta(\log n)$ iterations, we are guaranteed to ``hit" such an iteration w.h.p..  The final component of our algorithm is the procedure $\LS(G_i)$, that, given a path $P_i$ with $O(\sqrt{n})$ double-funnels, finds \emph{long} shortcuts (i.e., shortcuts that correspond to monotone paths that could be of any length) in $P_i$, resulting in a path $P_{i+1}$ that is shorter than $P_i$ by a constant factor.

Based on the above discussion, our algorithm proceeds as follows. We perform $\tilde{\Theta}(\sqrt{n})$ iterations. In each iteration we find all short monotone path and shortcut them (this results in a modified graph with larger arc gains). Moreover, in each iteration, with probability $\tilde{\Theta}(\frac{1}{\sqrt{n}})$ we additionally call $\LS$ which finds long monotone paths in the current graph, shortcuts them, and returns a modified graph.

\begin{figure}
    \centering
    \includegraphics[width=1\linewidth]{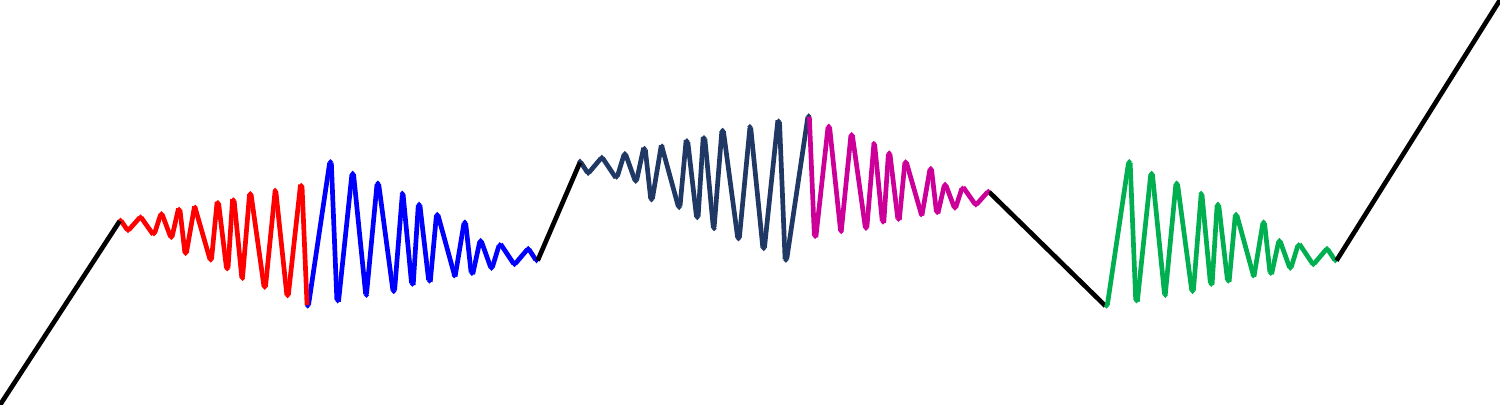}
    \caption{A decomposition of an ascending path to double-funnels that are maximal with respect to inclusion. Observe that ``the gap" between two double-funnels contains a short-monotone path. The double-funnels are split into two funnels. Note that the green double-funnel is not maximal with respect to inclusion (it can be extend backwards by 2 arcs), this was done for aesthetic reasons to show ``the gap" after the purple funnel.}
    \label{fig:alternation-funnel-shortcut}
\end{figure}

We now describe the procedure $\LS(G_i)$. We extensively use two path structures in $\LS$:  \emph{Arc-bounded} paths and \emph{funnels}, see Figure~\ref{fig:example_funnel_monotone}. A path $P = v_1\ldots v_k$ is first arc-bounded  if for every $i = 2\ldots, k$, it holds that $ \sum_{j=1}^{i-1}g(v_j v_{j+1}) \le \max \{0, g(v_1 v_2) \}$ and $ \sum_{j=1}^{i-1}g(v_j v_{j+1}) \ge \min \{0, g(v_1 v_2) \}$. 
A last arc-bounded path is defined analogously. A path is arc-bounded if it is either first or last arc-bounded. A path $P$ is a funnel if it is both arc-bounded and a double-funnel. 
Observe that any double-funnel can be decomposed to two funnels, each starts or ends at the edge of largest gain in absolute value, see Figure~\ref{fig:alternation-funnel-shortcut}.
Given the current graph $G_i$, $\LS(G_i)$ stores a table $D[\cdot][\cdot]$ such that for every $u,v,w\in V$, $D[uv][w]$ stores the largest recorded gain of a first arc bounded path in $G_i$ that starts with the arc $uv$ and ends at $w$ ($D[u][vw]$ is defined similarly for last arc-bounded paths). Algorithm $\LS$ first generates arc-bounded paths (that is, stores values in the table $D$) and finally, finds long monotone paths based on those arc bounded paths. To ease the explanation, we begin by demonstrating the latter. 

\subsection{Generating monotone paths from arc-bounded paths}\label{sec:arc-bounded-to-mono-review}
This part is straightforward: Given a vertex $u\in V$, we consider all arc-bounded paths that start at $u$ and we extend each by a single arc: We scan all triplets $v,w,x\in V$, such that $D[uv][w] \neq -\infty$, and ``concatenate" the arc-bounded path $P^{uv,w}$ that corresponds to $D[uv][w]$ with the arc $wx$, resulting in a path $P^{uv,x}$ to $x$ that starts with $uv$.\footnote{We do not actually store paths. Instead we examine the quantity $D[uv][w]+g(wx)$.} Assume $g(uv)>0$ (other cases are similar). If this concatenated path  remains arc bounded then we did not find a monotone path. Otherwise, either $D[uv][w]+g(wx) > g(uv)$ or $D[uv][w]+g(wx) < 0$. It is easy to see that in the former case, $P^{uv,x}$ is ascending (see Figure~\ref{fig:example-arc-bounded-to-mono}(a)), and in the latter case the subpath from $v$ to $x$ is descending (see Figure~\ref{fig:example-arc-bounded-to-mono}(b)). It is easy to see that the running time of this process is $O(n^3)$.

\begin{figure}
    \centering
    \includegraphics[width=1\linewidth]{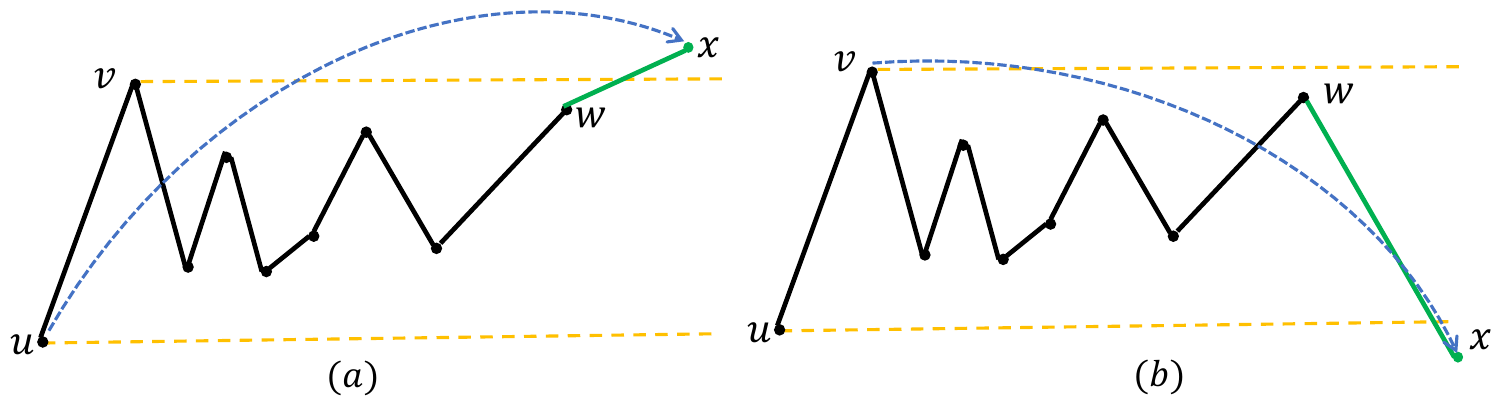}
    \caption{Finding a monotone path by extending an arc-bounded path by a single arc.}
    \label{fig:example-arc-bounded-to-mono}
\end{figure}

\subsection{Finding arc-bounded paths}\label{sec:technical-review-finding-arc-bounded-paths}
As already discusses, any path can be viewed as an alternation between double-funnels (which are just two funnels that are concatenated) and short monotone paths. Thus, handling funnels has a crucial role.

We compute arc bounded paths using two building blocks. 
\begin{enumerate}
    \item A  procedure $\ComputeF(H)$ to compute funnels. Given a graph $H$, $\ComputeF(H)$ returns a table $D[\cdot][\cdot]$ that dominates every funnel (which is a simple path) in $H$. That is, for every funnel $P = v_1 \ldots v_k$ that is first arc-bounded, it holds that $D[v_1 v_2][v_k] \ge g(P)$. Similarly, for every funnel $P = v_1 \ldots v_k$ that is last arc-bounded, it holds that $D[v_1 ][v_{k-1}v_k] \ge g(P)$. Moreover, the table $D$ is \emph{sound}. That is, for every $u,v,w\in V$, if $D[uv][w]\neq -\infty$, then there exists a first arc-bounded path $Q = v_1\ldots v_k$ (not necessarily a funnel) such that $g(Q) \ge D[uv][w]$. For the full details, see Appendix~\ref{sec:funnels}
    \item A concatenation procedure $\Concat(H,D,v)$. Given a graph $H$, a table $D[\cdot][\cdot]$ and a vertex $v\in V$. The procedure, in a brute force manner, scans all 4-tuples $(w,x,y,z)$ of vertices and then tries to concatenate a first arc-bounded path in $D$ that start with the arc $vw$ and end at $x$ with first arc-bounded path that start with the arc $xy$ and end at $z$.\footnote{Formally, we look on the quantity $D[vw][x]+D[xy][z]$ and verify some inequalities to make sure that the concatenated path is indeed first arc-bounded.} Note that this procedure only generates arc-bounded paths that start at $v$.
    For the full details, see Appendices~\ref{section:concatenate} and~\ref{section:CO}.
    A naive implementation of this procedure takes $O(n^4)$ time. 
    Using a balanced binary search tree, we get a running time of $\tilde{O}(n^3)$. Since our claimed running time for the entire algorithm is $\tilde{O}(n^{3.5})$, we can use the $\Concat$ procedure only $\tilde{O}(n^{0.5})$ times.
\end{enumerate}

We now describe $\LS(H)$ and the intuition about it.
The algorithm starts by calling to $\ComputeF(H)$, which in $\tilde{O}(n^{3\frac{1}{3}})$ time computes a table $D[\cdot][\cdot]$ that dominates all simple funnels in $H$. The algorithm then samples uniformly at random sets $S_i\subseteq V$ of size $\tilde{O}\left( \frac{\sqrt{n}}{2^i}\right)$, for $i=1,\ldots, \log (\sqrt{n})$. Then, for every $i=1,\ldots, \log (\sqrt{n})$ and $u \in S_i$ we perform $2^i$ times the procedure $\Concat(H,D,u)$. Finally, we extract monotone paths by applying the procedure from Appendix~\ref{sec:arc-bounded-to-mono-review} on every vertex in $S = \cup_i S_i$.

We now give the intuition behind the algorithm. Recall the discussion about ``hitting" an iteration in which $P_i = v_1\ldots v_k$ (an ascending path in $G_i$, for some $0\le i \le T = \sqrt{n}$, that represents the evolution of $P=P_0$ over the iterations of shortcutting) has at most $2\sqrt{n}$ double-funnels. Assume we run $\LS(G_i)$. Every vertex $v_j \in P_i$ defines a first arc-bounded path $P' = v_j \ldots v_t$, where $j\le t \le k$ is maximal such that $v_j\ldots v_t$ is first arc-bounded, see Figure~\ref{fig:example-different-sampling}.
Note that $P'$ may contain several double-funnels, say $f$. Thus, if we apply $\Concat(G,D,v_j)$, $\Theta(f)$ times, the table $D$ will ``find" $P'$ (that is we will have $D[v_{j} v_{j+1}][v_t] \ge g(P')$).
By the discussion in Appendix~\ref{sec:arc-bounded-to-mono-review}, if we extend $v_j \ldots v_t$ by the arc $v_t v_{t+1}$ we will find a monotone path of length $t-j + O(1)$. For this process to be efficient, we have to balance the work we do (which is proportional to the number of funnels in $P'$ which is the number of calls to concatenate that we need to do to find $P'$) to compute $P'$ with the reward we achieve (which is proportional to the length of $P'$) by shortcutting the monotone path corresponding to $P'$. 

We are shooting for a running time of $O(n^{3.5})$,
therefore as we already said we can call concatenate at most $O(\sqrt{n})$ times (recall that it works for a single particular vertex at each call).
In particular, for every $i=1,\ldots, \log (\sqrt{n})$,
the product of $|S_i|$ and the number of calls of concatenate from each vertex of $S_i$ should be $O(\sqrt{n})$.
To explain why we need the $O(\log(n))$ levels of sampling, we consider the two extreme cases which our sampling interpolates between.
That is, the case of
$i=\log (\sqrt{n})$
where $S_i=O(1)$ and the case of
$i=1$ where $|S_i|=O(\sqrt{n})$. 

These two cases are demonstrated in Figure~\ref{fig:example-different-sampling} for a path $P_i$ of length $\Theta(n)$  and 
$\Theta(\sqrt{n})$ funnels. The first example ($i=\log (\sqrt{n})$), depicted in Figure~\ref{fig:example-different-sampling}$(a)$, considers the case in which all funnels, except for the first one, are of constant length
and the rest is filled with the first funnel which is of linear size.
Moreover, the arc-bounded paths that correspond (in the manner explained in the previous paragraph) to every vertex in a short funnel are of constant length and the arc-bounded paths that correspond to vertices in the long funnel are all reaching the last arc of the path. Thus, in order to achieve sufficient reward (i.e., find long enough monotone paths), we  have to sample a vertex $u$ in the long funnel and then perform $\Theta(\sqrt{n})$ times $\Concat(G_i, D, u)$. Thus, the example shows that there are cases in which we have to perform $\Theta(\sqrt{n})$ concatenations at a single vertex. 

The second extreme case, depicted in Figure~\ref{fig:example-different-sampling}$(b)$, is the case in which all funnels are of length $\Theta(\sqrt{n})$ and for every $v\in P_i$, the arc-bounded path that corresponds to $v$ contains a single funnel. Thus, for every $v\in P_i$ we can apply a single concatenation and find the arc-bounded path that corresponds to $v$ and later extend it to a monotone path of length $O(\sqrt{n})$. In this case to reduce the length of $P_i$ by a constant factor, we  have to sample $\Theta(\sqrt{n})$ vertices (that will hit a constant fraction of the funnels) and perform a constant number of concatenation on each one of them.

\begin{figure}[ht!]
    \centering
    \begin{tabular}{m{0.01\textwidth} m{0.9\textwidth}} 
        $(a)$ & \includegraphics[width=\textwidth]{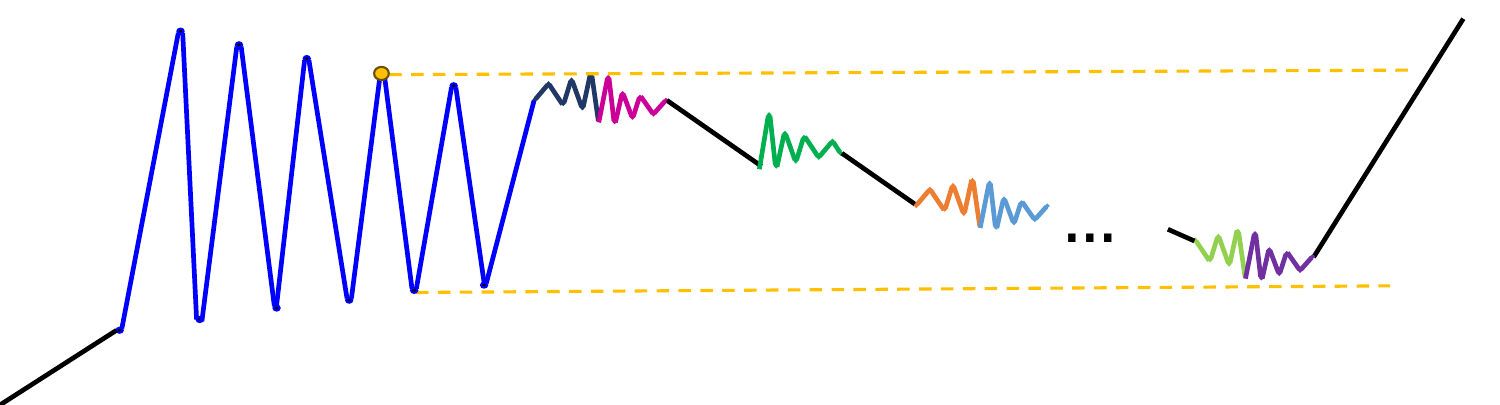} \\[1em]
        $(b)$ & \includegraphics[width=\textwidth]{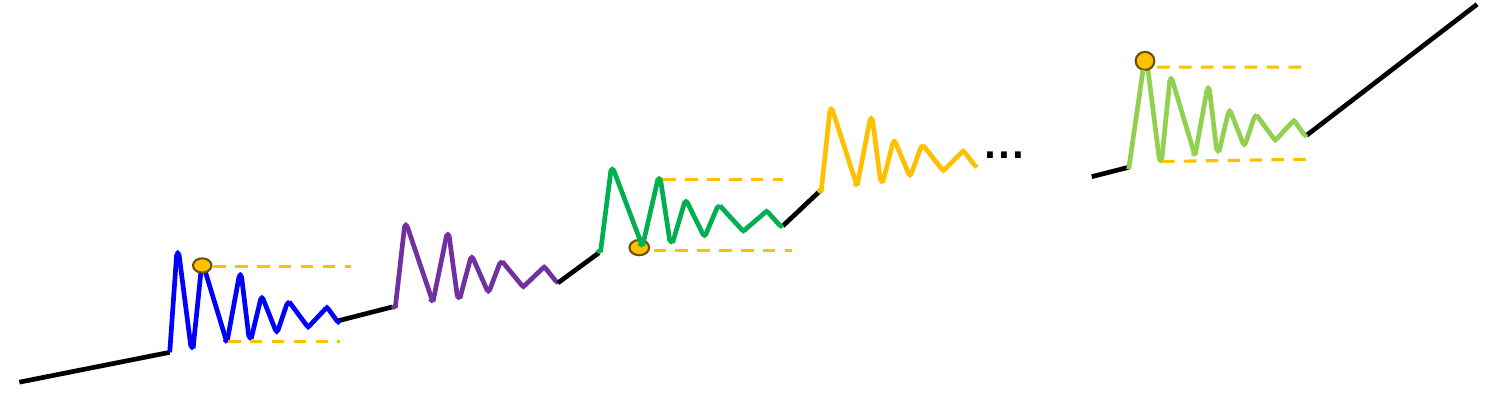} \\[1em]
    \end{tabular}
    \caption{Two extreme cases for algorithm $\LS$. Black lines represent single arcs. Figure $(a)$ shows why we need to sample $O(1)$ vertices but perform $\Theta(\sqrt{n})$ concatenations per vertex. Figure $(b)$ shows why we need to sample $\Theta(\sqrt{n})$ vertices but perform $O(1)$ concatenations per vertex. \label{fig:example-different-sampling}}  
\end{figure}

\subsection{Solving the all-pairs problem}
Finally, we briefly describe the key observations that relate monotone paths to the computation of $\alpha_B(\cdot,\cdot)$. We begin by assuming that the optimal energetic paths are simple and later show how to solve the general case in which the optimal paths use positive cycles.

\subsection{Simple energetic paths}\label{sec:technical-review-alpha-of-simple-paths}
Assume we have computed the table $M[\cdot][\cdot]$ that dominates every simple monotone path in $G$.
Let $s,t\in V$ and let $P = v_1\ldots v_k$ be an optimal energetic path from $v_1=s$ to $v_k=t$ (that is, $\alpha_B(s,t) = \alpha_B(P)$). We consider the special case in which $P$ is simple and for every $1 < i\le k$ it holds that $\alpha_B(v_1 \ldots v_i) < B$. That is, the car starts with full charge at $s$ and its charge level remains below $B$.
We decompose $P$ as follows. Let $v_{i_{1}} = s$ and let $v_{i_{2}}$ be the vertex of lowest gain in $P$. We define $v_{i_{3}}$ to be the vertex of highest gain in the suffix $v_{i_{2}}\ldots v_k$ and so on, see Figure~\ref{fig:mono-to-alpha-derivation}$(a)$. This results in a series of vertices $s=v_{i_{1}}, v_{i_{2}},\ldots, v_{i_{r}}=t$. Clearly, this partitioning divides $P$ into monotone segments that alternate between ascending and descending paths. A key observation is that these monotone paths are optimal in terms of gain. That is, for every $1\le j < r$, there is no monotone path $Q$ from $v_{i_{j}}$ to $v_{i_{j+1}}$ with larger gain than the subpath $v_{i_{j}}\ldots v_{i_{j+1}}$. Otherwise, we can replace the subpath $v_{i_{j}}\ldots v_{i_{j+1}}$ by $Q$ and increase the final charge at $t$,\footnote{We use here the fact that the battery is not full.} a contradiction to the optimality of $P$. Thus, for every $1\le j \le r$, it holds that $M[v_{i_{j}}][v_{i_{j+1}}] = g(v_{i_{j}}\ldots v_{i_{j+1}})$. Let $G'$ be a directed clique whose gains are defined by $M[\cdot][\cdot]$. The final observation is that $v_{i_{1}} v_{i_{2}}\ldots v_{i_{r}}$ is a funnel in $G'$. Thus, by calling $\ComputeF(G')$
we can find this funnel.

\subsection{Handling positive cycles}

A simple observation is that every positive gain cycle $C$ contains a pair of points $x,y\in C$ such that the car can start at $x$ with zero charge, and traverse the cycle until it reaches $y$ with a fully charged battery (i.e., $B$ charge).\footnote{It is possible that the car took the direct path in $C$ from $x$ to $y$, or it cycled through $C$ several times.} We say that $(x,y)$ is an \emph{entry-exit} pair of $C$, where $x$ is the \emph{entry} and $y$ is the exit.

We prove in Lemma~\ref{lemma:special-entry-exit}, that every positive cycle $C$ contains an entry-exit pair $(x,y)$ such that $C^{xy}$, the path from $x$ to $y$ through $C$, is ascending and $C^{yx}$, the path from $y$ to $x$ through $C$, is descending.\footnote{It is possible that $x=y$. For example in a cycle in which all arc gains are positive.} This lemma, leads to a simple algorithm for identifying entry-exit pairs: For every $x,y\in V$, if $M[x][y]>0$ and $M[x][y] + M[y][x] > 0$, then set $\alpha_0(x,y)=B$ (i.e., $(x,y)$ is an entry-exit pair).
The positive shortcut $M[x][y]$ indicates that there is an ascending path $P^{xy}$ from $x$ to $y$.
If $M[x][y]\ge B$ then clearly we can start at $x$ with zero charge and get to $y$ with full charge (by using the shortcut\footnote{Recall that using shortcuts does not change the $\alpha$ values since each shortcut corresponds to a monotone path in $G$ of the same gain.} $xy$ of gain $M[x][y]$). Otherwise, the second inequality $M[x][y] + M[y][x] > 0$ guarantees that we can start at $x$ with zero charge and get back to $x$ with positive charge (by using the shortcuts $xy$ and $yx$).
Therefore, by extending the path  to $y$, we generate an ascending path with larger gain $M[x][y]+M[y][x] + M[x][y] > M[x][y]$, see Figure~\ref{fig:mono-to-alpha-derivation}$(b)$. By repeating this multiple times, we get an ascending path from $x$ to $y$ with gain larger than $B$ justifying setting
$\alpha_0(x,y)=B$.

We perform $3$ additional simple inferences: For every $x,y,z\in V$

\begin{itemize}
    \item If $M[x][y] + M[y][z] \ge 0$ and $M[x][y] \ge 0$, we deduce that the path that consists of the two shortcuts $M[x][y],M[y][z]$ is a witness that $\alpha_0(x,z) \ge 0$. That is, it is possible to start at $x$ with zero charge and reach $z$: Either $M[x][y]\ge B$ and then the claim follows by the traversability of monotone paths ($M[y][z]\ge -B$) or $M[x][y] < B$ and therefore either  $M[y][z] \ge 0$ or $- M[x][y] \le M[y][z] < 0$. The former case is trivial. In the latter case, we can start with zero charge at $x$ and reach $y$ with $M[x][y]$ charge and then continue to $z$ and reach it with $M[x][y] + M[y][z] \ge 0$ charge.  
    \item If $M[x][y] + M[y][z]\ge 0$ and $M[y][z] \ge 0$, we deduce that $\alpha_B(x,z) = B$.
    \item If $M[x][y] \neq -\infty$ (so $M[x][y]\ge -B$), we infer that $\alpha_B(x,y) \ge 0$. That is, it is possible to reach $y$ if we start at $x$ with full charge.
\end{itemize}

\begin{figure}
    \centering
    \includegraphics[width=1\linewidth]{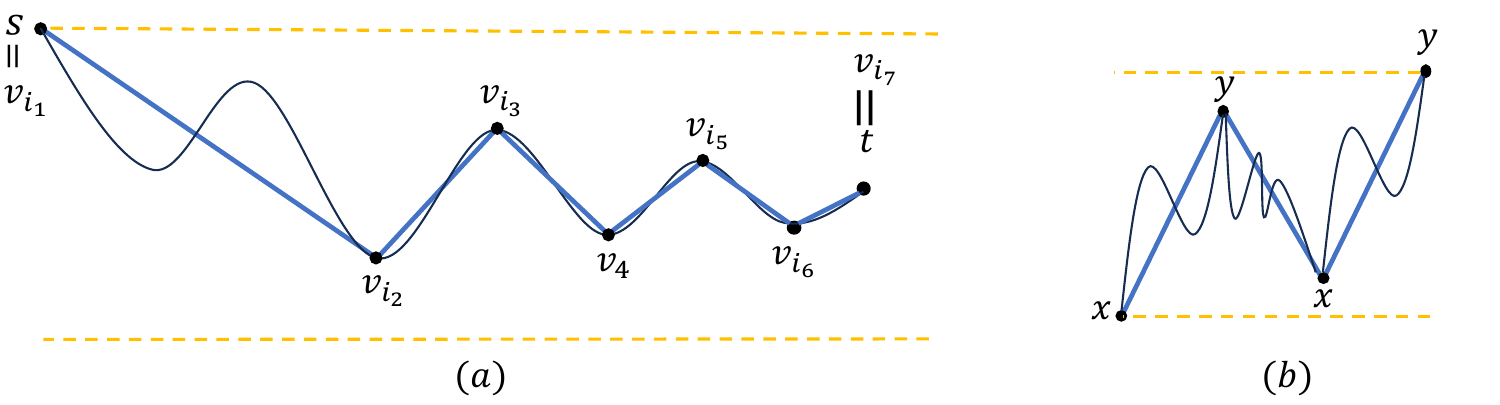}
    \caption{(a) A decomposition of an optimal path from $s$ to $t$ into a sequence of simple monotone paths. After shortcutting these paths, we are left with a funnel. (b) Illustration of why $M[x][y]+M[y][x]>0 \And M[x][y] >0$ leads to $\alpha_0(x,y)=B$. Each blue arc represents a shortcut in $M$. Each such shortcut can be unwrapped into a path in $G$}
    \label{fig:mono-to-alpha-derivation}
\end{figure}

Finally, we combine  these relations into a graph $H$ and compute its transitive closure $H^\star$. 
The graph $H$ is defined as follows. $H = (V^0 \cup V^B, E(H))$, where
$V^0 = \{v^0 \mid v\in V \}$
and 
$V^B = \{v^B \mid v\in V \}$
are two copies of $V$. Each vertex
$v^0\in V^0$ represents being at $v$ with $0$ charge and each vertex $v^B\in V^B$ represents being at $v$ with full charge.
  An arc $u^{b_1} v^{b_2}\in E(H)$ represents that $\alpha_{b_1}(u,v) \ge b_2$.\footnote{Note that the other direction does not necessarily hold: It is possible that $\alpha_{b_1}(u,v) \ge b_2$ but $u^{b_1} v^{b_2}\notin E(H)$.} We create the arcs $E(H) \subseteq \{u^{b_1} v^{b_2} \mid  \alpha_{b_1}(u,v) \ge b_2\}$ according to the $4$ relations shown above (for example, if $M[x][y]\neq -\infty$, we add the arc $x^B y^0$ to $H$).
 We claim in Theorem~\ref{theorem:B-B-strong-full} that, for every $s,t\in V$, $\alpha_B(s,t)=B$ if and only if $s^B t^B \in E(H^\star)$.

Using the graph $H^\star$, our  algorithm reduces the all pairs $\alpha_B(\cdot,\cdot)$ problem to the case in which the energetic paths are simple: For every $s,t \in V$, using $H^\star$, we find all vertices $x\in V$ such that $\alpha_B(s,x)=B$ and then, as in Appendix~\ref{sec:technical-review-alpha-of-simple-paths}, we find the best energetic simple path from 
any such $x$ to $t$. 

The following is a brief review of the correctness of the algorithm. Let $s,t\in V$ and let $P = v_1\ldots v_k$ be an optimal energetic path from $s$ to $t$ (i.e., $\alpha_B(s,t) = \alpha_B(P))$.
We argue that there is a vertex $x$ on $P$
such that $\alpha_B(s,x) = B$ and $\alpha_B(x,t) = \alpha_B(s,t)$.
If $\alpha_B(s,t)=B$, then we are done since this relation is already recorded in $H^\star$ and we can set $x=t$. Otherwise, let $1\le i\le k$ be maximal such that $\alpha_B(v_1\ldots v_i) = B$. It follows that $\alpha_B(s,v_i)=B$ and for every $i< j \le k$ it holds that $\alpha_B(v_1\ldots v_i) < B$. This implies that $v_i \ldots v_k$ must be a simple path.\footnote{Otherwise, $v_i \ldots v_k$ contains a positive cycle, so by repeating the cycle (and using the fact that no vertex on cycle, and the rest of the path, has already reached full charge) we can increase the final charge at $v_k=t$, a contradiction.} So we conclude that  the algorithm  finds the optimal energetic path when inspecting $x = v_i$.

\subsection{A technicality - charge drop schedules}

In this section we describe
\emph{Charge drop schedules} 
and the technical 
challenge that it addresses. Before we delve into the definition, we motivate it by pinpointing several problems with our arguments.
\begin{enumerate}
    \item Throughout this section we explained how to shortcut an ascending path  to single arc via a sequence of 
  short/long shortcut updates. A key invariant that is required for this argument to hold is the fact that given an ascending path $P = v_1\ldots v_k$, if we replace a monotone subpath $v_i \ldots v_j$ of $P$ by a monotone path  $Q$ of larger gain, then the resulting path $P' = v_1\ldots v_i \mid Q \mid v_j \ldots v_k$ (The $\mid$ stands for concatenation) is ascending and $g(P')> g(P)$. Unfortunately, this argument does not hold if $P$ is descending. For example, consider the graph $G$ in Figure~\ref{fig:charge-drop-fixes-issue-1}$(a)$ and the descending path $P = v_1 v_2 v_3 v_4 v_5$. After performing one iteration of the simple algorithm (computing all short monotone paths and updating the gains of the graph), we are left with a graph $G'$ with gain function $g'$ (see Figure~\ref{fig:charge-drop-fixes-issue-1}$(b)$) that does not contain any monotone path from $v_1$ to $v_5$. This is of course unsettling, as finding the best short shortcuts should be a good property of the algorithm and yet it destroyed some other descending paths

    \item Recall the procedure $\Concat(G,D,v)$ that scans all 4-tuples $(w,x,y,z)$ of vertices and then tries to concatenate a first arc-bounded path (stored in $D$) that starts with the arc $vw$ and ends at $x$ with first arc-bounded path that starts with the arc $xy$ and ends at $z$ (which is done by calculating $D[vw][x]+D[xy][z]$ and verifying some inequalities). Consider the following example: Assume $g(vw) = 5, g(xy) = 3$ and $D[vw][x] = 2, D[xy][z] = 3$. Therefore, by running  $\Concat(G,D,v)$, we will concatenate the arc-bounded paths corresponding to $D[vw][x]$ and $D[xy][z]$ and get an arc bounded path that starts at $vw$ and ends at $z$ with gain $D[vw][z] = D[vw][x] + D[xy][z] = 5$. Unfortunately, this concatenation is not guaranteed to happen.  It is possible that earlier in the run of $\Concat(G,D,v)$, the algorithm managed to improve $D[vw][x]$ to $D[vw][x] = 3$ and therefore concatenating $D[vw][x]$ to the arc-bounded path corresponding to $D[xy][z]$ does not result anymore in an arc-bounded path, see Figure~\ref{fig:charge-drop-fixes-issue-2}$(a)$. Again, by performing an update that should be good for us (increasing $D[vw][x]$ from $2$ to $3$), we hurt ourself somewhere else (we did not make the update $D[vw][z] = 5$). 
\end{enumerate}

\begin{figure}[t!]
    \centering
    \includegraphics[width=1\linewidth]{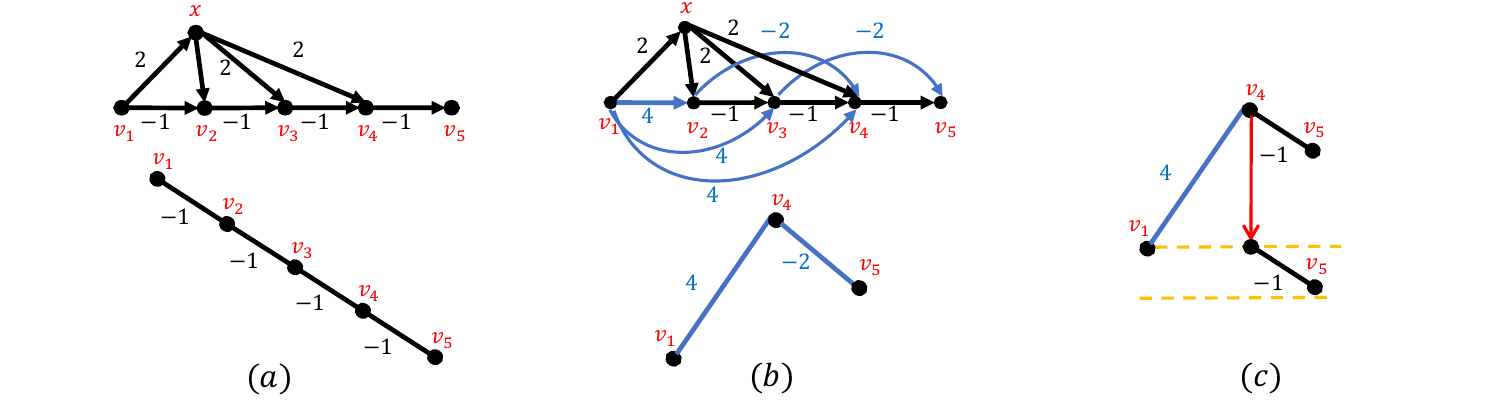}
    \caption{$(a)$ The graph $G$ and the descending path $P=v_1 v_2 v_3 v_4 v_5$. 
    $(b)$ The graph $G'$ that we get after shortcutting all short monotone paths. Blue arcs correspond to either new arcs or arcs with increased gain. Note that there is no monotone path from $v_1$ to $v_5$ in $G'$. $(c)$ By using charge drop schedule, we can transform the path $v_1 v_3 v_5$ into a short descending path of gain $-2$.}
    \label{fig:charge-drop-fixes-issue-1}

    \includegraphics[width=1\linewidth]{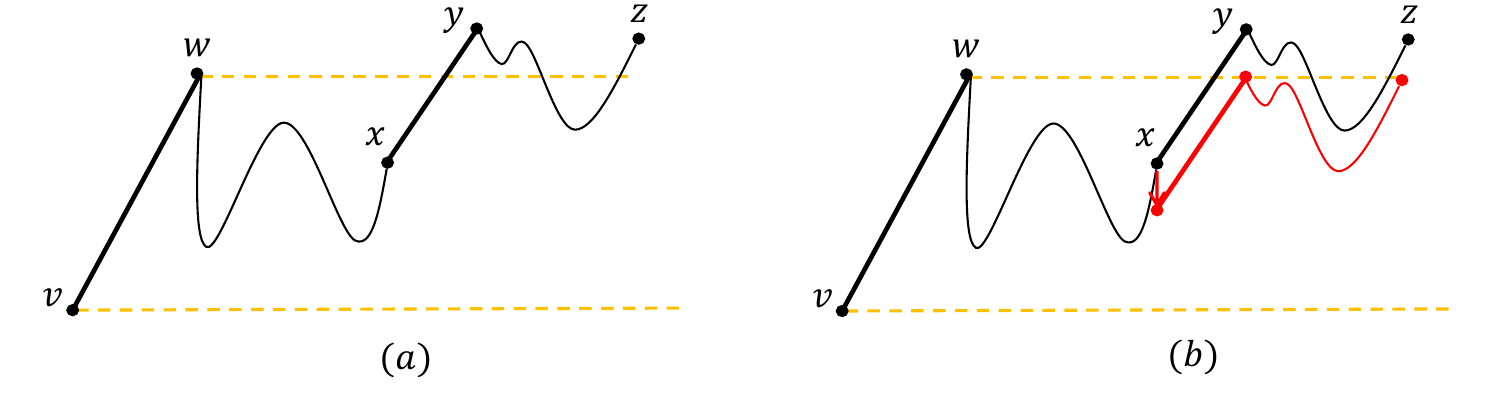}
    \caption{A use case of charge drops. $(a)$ Two arc-bounded paths whose concatenation is not arc-bounded. $(b)$ By applying a simple charge-drop schedule we make the concatenated path arc-bounded.}
    \label{fig:charge-drop-fixes-issue-2}
\end{figure}

In both examples, we suffered from having computed values that are ``too good". The simple concept that solves this problem is \emph{charge drop}. Charge drops allow us, at any vertex along the path, to get rid of some charge, see Figure~\ref{fig:example_charge_drop}. Formally, let $P = v_1\ldots v_k$ be a path in $G$. A charge drop schedule is a vector $C=(d_1,d_2,\ldots, d_k)\in \RR^k_{\ge 0}$, where $d_1 = 0$. The \emph{gain} at $v_i$ with respect to $P$ and $C$, denoted as $g^{P,C}_{v_i}$ is defined as $g_{v_i}^{P,C} =\sum_{t=1}^{i-1} g(v_{t} v_{t+1}) - \sum_{t=2}^{i} d_t$, for $2\le i\le k$ and $g_{v_1} = 0$ otherwise. Monotone paths and arc bounded paths can be defined similarly to before by replacing the gain of an arc $g(v_i v_{i+1})$ by $g(v_i v_{i+1}) - d_{i+1}$. When $P$ is clear from contexts, we abbreviate  $g^{P,0}_{v_i}$ and write $g_{v_i}$.

We now show how to fix the two examples using charge drop schedules.

\begin{enumerate}
    \item In the first example (see Figure~\ref{fig:charge-drop-fixes-issue-1}) $P= v_1 v_2 v_3 v_4 v_5$ is a descending path in $G$, but there is no descending (or ascending) path from $v_1$ to $v_5$ in $G'$.   Instead, $G'$ contains the path $v_1 v_4 v_5$ that has positive gain. By using a simple charge drop schedule that drops $4$ units of charge at $v_4$, 
    we view $v_1 v_4 v_5$ as a short descending path of gain $-2$, see Figure~\ref{fig:charge-drop-fixes-issue-1}$(c)$.

    \item In the second example we faced a problem when trying to concatenate an arc-bounded path corresponding to $g(vw)=5, D[vw][x]=3$ and an arc bounded path corresponding to $g(xy)=3, D[xy][z] = 3$. By simply dropping a single unit of charge at $x$ (the concatenation point), we are now able to concatenate the two paths and therefore assign $D[vw][z] = (D[vw][x] -1) +D[xy][z] = 5$, see Figure~\ref{fig:charge-drop-fixes-issue-2}.
\end{enumerate}

We incorporate charge drops in our algorithm in the following places.

\begin{enumerate}
    \item When computing all short monotone paths, if a path $P$ (of length $2$ or $3$) starts by a negative gain arc, we will always apply charge drop schedule and create a descending path out of $P$. For example, if $P = v_1 v_2 v_3 v_4$ and $g(v_1 v_2) = -5, g(v_2 v_3) = 2, g(v_3 v_4) = -1$, then we record a descending path from $v_1$ to $v_4$ of gain $-5$ (this corresponds to dropping one unit of charge at $v_4$).
    \item In the computation of long monotone paths. Recall that we consider tuples $u,v,w,x\in V$ and we extend the arc-bounded path that corresponds to $D[uv][w]$ by the arc $wx$. We incorporate charge drops in the following case: If $g(uv) < 0 $ and $D[uv][w]+g(wx)\in [g(uv) , 0]$ (that is the concatenated path remains arc-bounded), we record a descending path from $u$ to $x$ of gain $g(uv)$. This corresponds to performing a charge drop at $x$ that drops $D[uv][w]+g(wx)-g(uv)$ charge. 
    \item In the concatenation procedure, whenever the concatenation of the two arc bounded paths does not yield an arc-bounded path, we perform a charge drop to force the result to be arc-bounded. That is, for every $v,w,x,y,z\in V$, if $g(vw)>g(xy)>0$ and $D[vw][x] + D[xy][z] > g(vw)$, we set $D[vw][z] = g(vw)$. This corresponds to performing the smallest possible charge drop at $x$ such that the concatenated path is arc-bounded, see Figure~\ref{fig:charge-drop-fixes-issue-2}$(b)$. 
\end{enumerate}

\subsection{Main technical lemma}
In this section, we prove a simplified version\footnote{We address only ascending paths.} of our main lemma (Lemma~\ref{lemma:paths-shrink}). Recall our algorithm: We perform $\tilde{\Theta}(\sqrt{n})$ iterations. In each iteration we find all short monotone path and shortcut them (this results in a modified graph with larger arc gains). Moreover, in each iteration, with probability $\tilde{\Theta}(\frac{1}{\sqrt{n}})$ we additionally call $\LS$ which finds long monotone paths in the current graph, shortcuts them, and returns a modified graph.

\begin{lemma}\label{lemma:paths-shrink-simple-version}
    Let $P= v_1\ldots v_k$ be a simple ascending path in $G$. Let $G'$ be the modified graph after $\sqrt{n}$ iterations of the modified algorithm and let $g'$ be its gain function. If $|P| \le  \sqrt{n}$, then $g'(v_1v_k) \ge g(P)$. If $|P| > \sqrt{n}$, then w.h.p.\ there is an ascending path $P'$ in $G'$ from $v_1$ to $v_k$ in that satisfies $g'(P')\ge g(P)$ and $|P'| \le (1- 1/\Omega(\log n))\cdot |P|$.
\end{lemma}

 Lemma \ref{lemma:paths-shrink-simple-version} is  derived from Lemma~\ref{lemma:long-shortcuts-shrinks-simple-version}, which is our main technical lemma. It provides guarantees about $\LS$, when run on a graph with an ascending path that contains few double-funnels.

\begin{lemma}\label{lemma:long-shortcuts-shrinks-simple-version}
    Let $P = e_1 \ldots e_k$ be a simple ascending path in $G$ from $x$ to $y$. 
    Let $t(\ge 1)$ be the number of double-funnels in $P$ that are maximal with respect to inclusion. Let $G'$ be the updated graph resulted from $\LS(G)$.\footnote{Note that every non empty path contains at least one double-funnel.} If $t \le k / \sqrt{n} $, then w.h.p. there is an ascending path $P'$ in $G'$ from $x$ to $y$ that satisfies $g^{G'}(P')\ge g^{G}(P)$ and $|P'| \le (1-1/\Omega(\log n))\cdot |P|$. 
\end{lemma}

We prove Lemma~\ref{lemma:long-shortcuts-shrinks-simple-version} at the end of this section. The derivation of Lemma~\ref{lemma:paths-shrink-simple-version} is now straightforward.

\begin{proof}[Proof of Lemma~\ref{lemma:paths-shrink-simple-version}]
    Let $r= \sqrt{n}$ and let $G_0(=G),G_1,\ldots, G_r$ be the graphs throughout the $r$ iterations of the algorithm. Let $P_0=P,P_1,\ldots, P_r$ be a series of monotone paths, where $P_i$ is the shortest path in $G_i$ from $v_1$ to $v_k$ that has no smaller
    gain (with respect to $G_i$) than $P_{i-1}$ (with respect to $G_{i-1}$). We split the proof into cases.
    
     \textbf{Case $|P|\le r$:} Since in each of the $r$ rounds we compute all the short monotone paths, 
     and since every monotone path contains a short monotone path, we get that for every $1\le i < r$, if $|P_i|>1$ then $|P_{i+1}| < |P_i|$. Thus, $|P_r|=1$ and the lemma follows.

    \textbf{Case $|P| > r$:} 
    If $P_r \le |P|/2$, then we are done. Otherwise $P_r > |P|/2$ and therefore for at least $r/2$ indices $0\le i< r$, it holds that $|P_i|-|P_{i+1}| \le |P|/r$. This mean that, for each such index $i$, $P_i$ has at most $|P|/r$ disjoint short shortcuts as subpaths. 
    Thus,
    by our arguments in the previous sections (see Figure~\ref{fig:alternation-funnel-shortcut}),
    $P_i$ contains $O(|P|/r) = O(|P_i|/r)$  double-funnels that are maximal with respect to inclusion. 
    Therefore,
    w.h.p.\ we run $\LS(G_i)$ at an iteration $i$ such that $P_i$ contains $O(|P_i|/r) = O(|P_i|/\sqrt{n})$ double-funnels. Hence, the conditions of Lemma~\ref{lemma:long-shortcuts-shrinks-simple-version} are satisfied and we are done. 
\end{proof}

Before proving Lemma~\ref{lemma:long-shortcuts-shrinks-simple-version}, we need to introduce the following structural definitions.
These definitions allow us to measure how many applications of $\Concat$ are needed in order to dominate an arc bounded path.

\begin{definition}
    Let $P=e_1 \ldots e_k$ be a path in $G$. For every $1\le i \le k$ we define
    $s^P(i)\ge i$ to be the maximal index such that $e_i \ldots e_{s^P(i)}$ is first arc-bounded.
    When $P$ is clear from the context, we abbreviate and write $s(i)$. 
\end{definition}

\begin{definition}\label{def:funnel-distance-simple-version}
    Let $P= e_1 \ldots e_k$ be a path in $G$. For every $i$, we define $f^P(i)$ as the number of first arc-bounded funnels in $e_i \ldots e_{s(i)}$ that are maximal with respect to inclusion.
    When $P$ is clear from context, we abbreviate and write $f(i)$.
\end{definition}

The following lemma proves that for every path $P=e_1\ldots e_k$, the set of paths $\{e_i \ldots e_{s(i)}\mid 1\le i \le k\}$ is laminar. We defer the proof of this lemma to the appendix (see Lemma~\ref{lemma:laminar}).

\begin{lemma}\label{lemma:laminar-simple-version}
    Let $P=e_1 \ldots e_k$ be a path in $G$, then the set of intervals $\{(i, s(i)) \mid 1\le i \le k \}$ is laminar.
\end{lemma}

We are now ready to prove Lemma~\ref{lemma:long-shortcuts-shrinks-simple-version}.

\begin{proof}[Proof of Lemma~\ref{lemma:long-shortcuts-shrinks-simple-version}]
     Let $F_1,\ldots F_t$ be the disjoint double-funnels in $P$. By the discussion in Section~\ref{sec:technical-review}, there are $O(t)=o(k)$ arcs in $P$ that are not contained in the double-funnels (see Figure~\ref{fig:alternation-funnel-shortcut}). Every double-funnel can be decomposed into at most $2$ funnels (last arc-bounded followed by first arc-bounded). Let $F'_1,\ldots, F'_{t'}$, where $t \le t' \le 2t$, be the corresponding funnels. We distinguish between funnels that are first-arc bounded to those which are last-arc bounded. Assume that the majority of the arcs of $P$ belong to first-arc bounded funnels. The analysis for the other case is symmetric. Therefore, these funnels (first-arc bounded) contain at least $k/3$ arcs.\footnote{The choice of $3$ and not $2$ is due to the subtlety that the disjoint double-funnels do not necessarily cover all of $P$.} Among these funnels, we consider only funnels of length at least $\sqrt{n}/6$. Note that at least $k/6$ arcs belong to such funnels (if more than $k/6$ arcs belong to funnels of length at most $\sqrt{n}/6$ then we need at least $t > k / \sqrt{n} $ funnels to accommodate them, a contradiction). 
     Denote these arcs by $e_{i_1},\ldots e_{i_r}$ ($r \ge k/6$).
    
    By Lemma~\ref{lemma:laminar-simple-version}, the set $A=\{(i_j, s(i_j)) \mid 1\le j \le r \}$ is laminar. We refer to each item in $A$ as an \emph{interval}. Recall that each interval $(i_j,s(i_j))$ corresponds to a monotone path of the same length (A maximal arc bounded path extended by a single arc is monotone), see Section~\ref{sec:arc-bounded-to-mono-review}. 
    Moreover, in order for $\LS$ to shortcut the monotone path corresponding to $(i_j,s(i_j))$, $\LS$  has to sample $v\in e_{i_j}=(v,w)$ and then perform $f(i_j)$ concatenations from $v$.

In the rest of the proof, we prove that $\LS$ finds enough disjoint monotone paths of total length $\Omega(k/\log k)$. To this end, we partition $A$ into disjoint sets $A_1,\ldots, A_{\log \sqrt{n}}$, where $A_i = \{(i_j,s(i_j)) \mid f(i_j) \in  [2^i, 2^{i+1}) \}\subseteq A $ correspond to all intervals/monotone paths that require $c\in [2^i, 2^{i+1})$ concatenations in order to be realized. We then prove that $A_{i^\star}$, the largest of these sets (hence of size $\Omega(k/\log n)$),  contains a collection of disjoint \emph{chains} (a chain is a set of nested intervals) $B'_1,\ldots,B'_{q'}\subseteq A_{i^\star}$ such that:
\begin{enumerate}
    \item The chains are pairwise internally disjoint.
    That is, for every $1\le j_1<j_2\le q'$ and $(\ell_1,r_1)\in B'_{j_1},$ $(\ell_2,r_2)\in B'_{j_2}$ it holds that $(\ell_1,r_1)\cap (\ell_2,r_2) = \emptyset$.
    \item $|B'_j| = \Omega\left( \frac{\sqrt{n} 2^{i^\star}}{\log n} \right)$, for $j=1,\ldots,q'$. This property is crucial for the sampling to ``hit" $B'_j$.
    \item $|\bigcup_{i=1}^{q'}  B'_i| = \Omega(|A_{i^\star}|)=\Omega(k/\log n)$.
\end{enumerate}
Finally, by Property $(2)$, we show that w.h.p., for every $j=1,\ldots,q'$, $\LS$ realizes an interval from $B'_j$ whose length is at least $|B'_j|/2$. By combining these disjoint (Property $(1)$) shortcuts, we reduce the size of $P$ by $\sum_{i=1}^{q'}  |B'_i|/2 = \Omega(k/\log n)$.

    We now show the lower bound on the size of $A_{i^*}$ and prove that it contains a collection of chains $B'_1,\ldots, B'_{q'}$ that satisfy the above poperies.
Since $i^\star$ is such that $|A_{i^\star}| \ge |A_i|$ for every $1\le i \le \log \sqrt{n}$ and $|A| \ge k/6$ (by the laminarity of $A$ each interval contains an edge which is not in any other interval) it follows that $|A_{i^\star}|\ge \frac{k}{6 \log{\sqrt{n}}}$. 
    Observe that for every $1\le i \le \log \sqrt{n}$, $A_{i}$ is laminar as a subset of $A$. Moreover, each interval in $A_{i}$ cannot contain two disjoint intervals in $A_i$. Indeed,  assume $(i_{j_1},s(i_{j_1})), (i_{j_2},s(i_{j_2})) \subseteq (i_{j_3},s(i_{j_3}))$ and $(i_{j_1},s(i_{j_1})) \cap (i_{j_2},s(i_{j_2})) = \emptyset$, where all intervals belong to $A_{i}$. Therefore $f(i_{j_3}) \ge f(i_{j_1}) + f(i_{j_2}) \ge 2^{i} + 2^{i} = 2^{i+1}$, so $(i_{j_3},s(i_{j_3})) \notin A_{i}$, a contradiction. It follows that we can decompose $A_i$ (and in particular $A_{i^*}$ ) into a collection of internally disjoint chains. 

   Let $B_1,\ldots, B_q$ be the decomposition of  $A_{i^\star} $ into internally disjoint chains ($A_{i^\star}= \cup_{i=1}^{q} B_i$). Since the $B_i$'s are internally disjoint (and so are the funnels in them), $q\cdot 2^{i^\star} \le t$.  
   Let $A'_{i^\star}$ be the union of the~$B_i$'s that satisfy $|B_i| \ge \frac{k}{12q \log  \sqrt{n}}$.
   It follows that

    \begin{align}\label{eq:main_lemma-simple-version}
        |A'_{i^\star}| 
    \ge  |A_{i^\star}| - q \cdot \frac{k}{12q \log  \sqrt{n}}
    \ge \frac{k}{12 \log \sqrt{n}}.
    \end{align}

    Let $B'_1,\ldots, B'_{q'} $ be the chains of $A'_{i^\star}$.  Let $B'_j \subseteq A'_{i^\star}$, it holds that
    \begin{align*}
        |B'_j| \ge \frac{k}{12q \log \sqrt{n}} \numge{1} \frac{k \cdot 2^{i^\star}}{12t \log \sqrt{n}} \numge{2} \frac{\sqrt{n} 2^{i^\star}}{12 \log \sqrt{n}} = \Omega \left( \frac{\sqrt{n} 2^{i^\star}}{\log n} \right),
    \end{align*}
    where Inequality~$(1)$ follows since $q\cdot 2^{i^\star}\le t$ and Inequality~$(2)$ follows since $t \le k / \sqrt{n}$. 
     
     Recall that $\LS(M)$ samples vertices to $S_{i^\star}$ i.i.d.\ with probability $p_{i^\star}=\Theta(\frac{\log^2 n}{2^{i^\star}\sqrt{n}})$. Since $\LS$ performs $2^{i^\star}$ concatenations from every vertex in $S_{i^\star}$, 
     every interval in $A_{i^\star}$ has a probability of $p_{i^\star}$ to be realized. 
     Let $B'_j \subseteq A'_{i^\star}$. 
    Since  $|B'_j| = \Omega 
    \left( \frac{\sqrt{n} 2^{i^\star}}{ \log n} \right)$, it follows by the Chernoff bound that w.h.p.\ we realize an interval from $B'_j$ of length at least $0.5 |B'_j|$.

    Since $B'_1,\ldots, B'_{q'} $ are internally disjoint, then the above realized shortcuts (one from every $B'_j$) are also disjoint. Hence, by shortcutting the realized intervals we get an ascending path $P'$ in $G'$ of length:
    
    \begin{align*}
       |P'| & \le k - \sum_{j=1}^{q'} {0.5 |B'_j|} =
       k - 0.5 |A'_{i^\star}| 
       \numle{1} k -0.5 \frac{k}{12\log \sqrt{n}} \\&= 
       \left(1 - \Omega\left(\frac{1}{\log n} \right) \right) \cdot k = 
       \left(1 - \Omega\left(\frac{1}{\log k} \right) \right) \cdot |P|,
    \end{align*}
    where Inequality $(1)$ follows from Equation~(\ref{eq:main_lemma-simple-version}) and the last equality holds because, according to the statement of the lemma, $\sqrt{n} \le t \sqrt{n} \le k < n$.
\end{proof}

\section{Concluding remarks}\label{S-concl}

We presented a randomized $\tilde{O}(n^{3.5})$-time algorithm for the finding optimal energetic paths between all-pairs of vertices in a weighted directed $n$-vertex graph with positive and negative gains that may contain positive-gain cycles. This improves upon a previous $\tilde{O}(mn^{2})$-time algorithm by Dorfman et al.~\cite{DorfmanKTZ23}. The new algorithm is quite involved and requires the introduction of many new ideas. Improving the running time of the algorithm is a natural open problem. 

\bibliographystyle{plain}
\bibliography{bibliography}

\appendix

\section{Full Version}

This appendix contains the full technical details of the paper and is organized as follows. In Appendix~\ref{S:prelim} we begin with some preliminary material. Appendix~\ref{S-overview} then gives an overview of the algorithm. The new algorithm is composed of two stages. In Stage I, described in Appendix~\ref{S-shortucts}, sufficiently many shortcuts are found. The correctness of Stage I is proved in Appendix~\ref{sec:correctness}. Stage II, described in Appendices~\ref{S-relating} and~\ref{S-algorithm}, uses the shortcuts found in stage I to find the $\alpha_B(s,t)$ values, and an implicit representation of the optimal energetic paths.

\section{Preliminaries}\label{S:prelim}

Let $G=(V,A,g)$, where $g:A\to\RR$ is a gain function.
Fix the battery capacity $B> 0 $.
Suppose we traverse a path
$P=v_1 u_2\ldots v_k$ 
starting with a charge of $b$ at $v_1$. We define
$\alpha_b(P)\le B$ to be the amount of charge with which we reach $v_k$. If $P$ cannot be traversed with this initial charge, we let $\alpha_b(P) = -\infty$. For $s,t\in V$ and $b\in [0,B]$, define $\alpha_b(s,t)=\max \{ \alpha_b(P) \mid \text{$P$ is a path from $s$ to $t$}\}$, i.e., the maximal final charge possible at $t$ when starting at $s$ with  $b$ charge. It is proved in~\cite{DorfmanKTZ23} that the $\max$ in this definition is well-defined. (Note that the maximum is over a possibly infinite collections of paths, since the paths are not necessarily simple.) A path $P=v_1\ldots v_k$ is optimal if $\alpha_B(v_1,v_k)=\alpha_B(P)$. The all-pairs maximum final charge problem is to compute $\alpha_B(s,t)$ for every pair $s,t\in V$. We say that a path $P$ is \emph{traversable} if $\alpha_B(P)\ge 0$, that is there is some energy level  that we can start with and traverse $P$. We say that a path $P$ is strongly traversable if $\alpha_0(P)\ge 0$. We let $|P|$ be the \emph{length} of $P$, i.e., the number of arcs in $P$.

The gain of an arc $uv\in A$ is $g(uv)$. The gain of a vertex $v$ in a path $P$ is the sum of gains of the arcs  that lead to $v$ in $P$. During our analysis we allow ourselves to dispose of some charge while traversing a path. This leads to the following definition of gains on paths that takes into account charge drops, see Figure~\ref{fig:example_charge_drop}.

\begin{definition}[Gain]\label{def:gain}
    Let $G=(V,A,c)$. Let $P=v_1\ldots v_k$ be a path in $G$ and let $C=(0,d_2,\ldots, d_k)\in \RR^k_{\ge 0}$   be a charge drop schedule.\footnote{Note that $d_1=0$, i.e., we do not drop charge at the first vertex.} The \emph{gain} at $v_i$ with respect to $P$ and $C$, denoted as $g^{P,C}_{v_i}$ is defined as $g_{v_i}^{P,C} =\sum_{t=1}^{i-1} g(v_{t} v_{t+1}) - \sum_{t=2}^{i} d_t$, for $2\le i\le k$ and $g_{v_1} = 0$ otherwise. That is, $d_t$ is the charge drop performed at $v_t$ for $2\le t \le k$. 
    We omit $P$ and $C$ and write $g_{v_i}$ when $P,C$ are clear from the context. The gain of $P$ with respect to $C$, denoted $g^C(P)$, is defined to be $g^{P,C}_{v_k}$. When no charge drop schedule is introduced, then we assume that the schedule is zero: $C = (0,\ldots,0) \in \RR^k_{\ge 0}$. 
\end{definition}

\begin{figure}[t]
    \centering
    \includegraphics[width=1.00\textwidth]{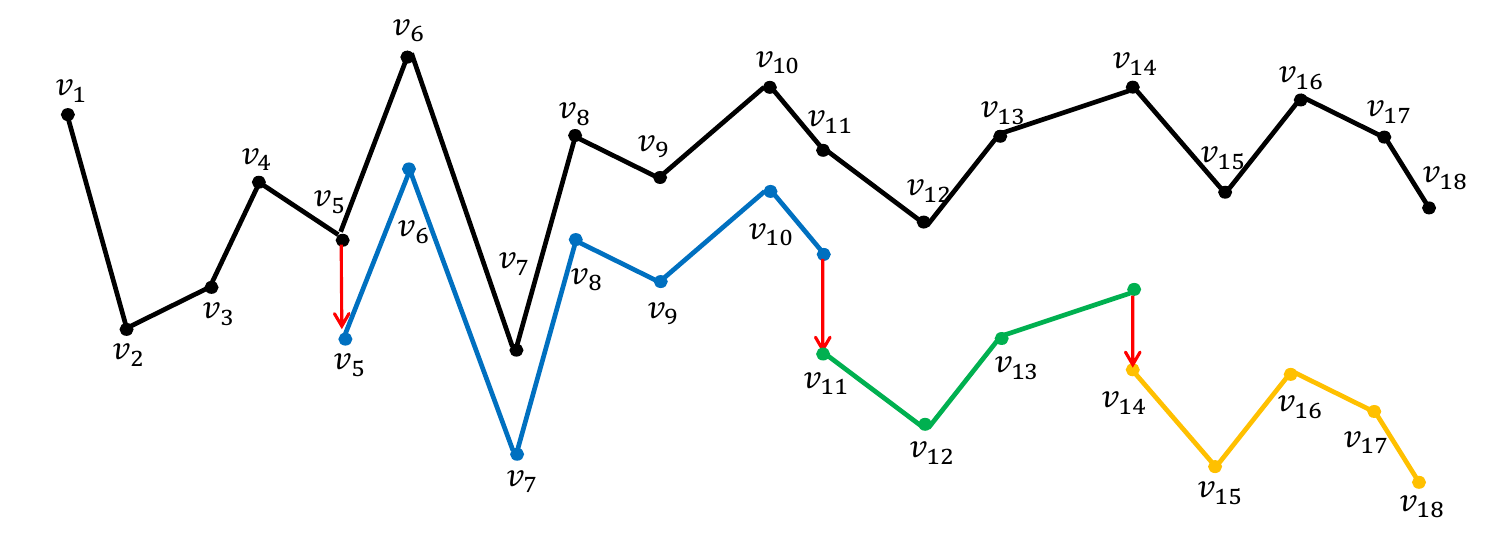}
    \caption{
    In black: The original gains of $P=v_1\ldots v_{18}$. Red downward arrows correspond to charge drops. In color: the gains of $P$ with respect to the charge drops. After each charge drop we switch color. Note that each colored path has matching gains to those of a corresponding subpath of~$P$.}
    \label{fig:example_charge_drop}
\end{figure}

Note that unlike the definition of the charge level of the electric car, the above definition allows the gains of vertices on a path $P$ to be larger than $B$ and smaller than $-B$.\footnote{Note that a path of gain smaller than $-B$ is not traversable.} 
Our algorithm, however, does not 
compute paths (and even subpaths) of gain smaller than $-B$.

Throughout this paper charge drops are used by the algorithm only twice, in Appendices~\ref{sec:extend} and~\ref{section:long}. 
It may be instructive for a reader to first think of 
the case where all charge drops are $0$. In the following sections we define path structures that are studied throughout the paper.

\subsection{Monotone Paths and Shortcuts}

\begin{definition}[Monotone path]\label{def:monotone-path}
    Let $P=v_1\ldots v_k$ be a traversable path in $G$ and let $C$ be a charge drop schedule for $P$.
    \begin{itemize}
        \item We say that $P$ is ascending with respect to $C$ if $0 = g^{C}_{v_1} \le g^C_{v_i} \le g^C_{v_k}$, for every $1\le i\le k$. See Figure~\ref{fig:example_funnel_monotone}$(a)$.
        \item We say that $P$ is descending with respect to $C$ if $0 = g^C_{v_1} \ge g^C_{v_i} \ge g^C_{v_k}$, for every $1\le i\le k$. See Figure~\ref{fig:example_funnel_monotone}$(b)$.
    \end{itemize}
    We say that $P$ is monotone with respect to $C$ if it is either ascending or descending with respect to $C$. We say that $P$ is monotone if it is monotone with respect to the zero schedule.
\end{definition}

Note that all ascending paths are strongly traversable. Also note that an ascending path might have a descending subpath and vice versa.

\begin{lemma}\label{lem:ascending_no_CDS}
    If a path $P=v_1\ldots v_k$ is ascending with respect to a charge drop schedule $C$, then $P$ is ascending with respect to the zero schedule.
\end{lemma}

\begin{proof}
    Observe that $g^{P,C}_{v_1} = g^{P}_{v_1}=0$. Since $P$ is ascending with respect to $C$, we get that $g^{P,C}_{v_i} \le g^{P,C}_{v_k}$ for every $1\le i\le k$. Since $v_k$ is the last vertex (and therefore encounters the largest charge drop), we get that  $g^{P}_{v_i} \le g^{P}_{v_k}$ for every $1\le i \le k$. Therefore, for every $1 \le i \le k$, it holds that 
    $$g^{P}_{v_1} = g^{P,C}_{v_1} \le g^{P,C}_{v_i} \le g^{P}_{v_i} \le g^{P}_{v_k}.$$
\end{proof}

\begin{definition}[Shortcut]
    We define an arc $e=xy$ (not necessarily in $A$) to be a $k$-\emph{shortcut} in $G$  if there is a path $P=v_1\ldots v_k$ from $x$ to $y$ in $G$ which is monotone with respect to a charge drop schedule~$C$.
    We say that the gain of the shortcut is 
     $g(e)=g^C(P)$. We say that~$e$ is a shortcut in $G$ if it is a $k$-shortcut in $G$ for some $k$.
    The shortcut $e$ is ascending if $P$ is ascending and descending if $P$ is descending. We say that~$e$ is a short shortcut if it is a $k$-shortcuts for $k\in \{ 2,3\}$.
\end{definition}

Note that we may have parallel shortcuts corresponding to different paths, but in this case we  only keep the one of largest gain.

It is convenient to think of $A$ as a clique where some arcs may have gain $-\infty$.
Our algorithms are going to compute sets of shortcuts in some base graph $G$.
Based on such a set of shortcuts $S$, it constructs a new graph $G'$ in which $g(xy)$ for 
every arc $xy$ is the maximum between 
$g(xy)$ in $G$ and the gain of the shortcut $xy$ in $S$. Our definitions of gain apply to the original graph or any graph that we obtain when using this procedure.

The following lemma states a core concept of our shortcutting algorithm: Every monotone path has a subpath that is a short monotone path.

\begin{lemma}\label{lemma:mono-has-shortcut}
    Every monotone path $P=v_1 \ldots v_k$ with respect to a charge drop schedule $C$, where $t>1$, contains a short shortcut with respect to $C$.
\end{lemma}

\begin{proof}
    Let $g_i = g^{P,C}_{v_i}$ for every $1\le i \le k$ and denote $g^e_i = g(v_{i-1}v_{i})-C(v_{i})$ for $1 < i \le k$. 
    Observe that $g^{P,C}_{v_i}= \sum_{j=2}^i g^e_j$, for $1 < i \le k$.
    
    By contradiction, assume that $P$ does not contain a short shortcut with respect to $C$. In particular $k > 4$. Moreover, $sign(g^e_i)\neq sign(g^e_{i+1})$ and $g^e_i\neq 0$ for every $1 < i < k$ (otherwise $v_{i-1} v_i v_{i+1}$ is monotone with respect to a sub-schedule of $C$).
    
    Assume $P$ is descending with respect to $C$, the other case is symmetric. We prove by induction that $|g^e_2|\ge|g^e_3| > \ldots > |g^e_{k}|$.\footnote{Note that the first inequality is weak. This is similar to the definition of funnels in the next subsection (see Definition~\ref{def:funnel}),} The base case holds since otherwise $P$ is not descending with respect to $C$. Let $2<i < k$, we prove that $|g^e_i| > |g^e_{i+1}|$. By contradiction, assume $|g^e_i| \le |g^e_{i+1}|$. It is easy to see that $v_{i-2}v_{i-1}v_i v_{i+1}$ is monotone with respect to $C$.

    Thus, $|g^e_2|\ge|g^e_3| > \ldots |g^e_{k}|$. Since $P$ is descending with respect to $C$, we get $g^{P,C}_{v_k} \le g^{P,C}_{v_{k-1}}$ and $sign(g^e_{k}) < 0 < sign(g^e_{k-1})$. Therefore $|g^e_{k}| \ge |g^e_{{k-1}}|$, a contradiction.
\end{proof}

The following lemma shows the relation between paths that reach full charge when stating with zero charge, to ascending paths.

\begin{lemma}\label{lemma:observation-monotone}
    Let $P$ be a path from $x$ to $y$. If $\alpha_0(P) = B$, i.e., $P$ is strongly traversable and it reaches $y$ with full charge, then $P$ is ascending.
\end{lemma}

\begin{proof}
    Denote $P = v_0 \ldots v_k$. Since $P$ can be traversed with no initial charge then $g_{v_0}=0 \le  g_{v_i}$  for every $1\le i \le k$. By contradiction, assume there is $i<k$ such that $g_{v_i} > g_{v_k}$. This means that $g(v_i \ldots v_k) < 0$ and therefore we reach $v_k$ with strictly less charge than $v_i$,   contradicting the assumption that we can reach $v_k$ with full charge.
\end{proof}

\subsection{Arc-Bounded Paths}

We next define arc-bounded paths, a core structure of our algorithm. A path $P=v_1\ldots v_k$ is first-arc bounded if the gain of every $v\in P$ is between the gains of the first two vertices, see Figure~\ref{fig:bounded_and_funnel_example}(a)-(b). We also defined arc-bounded paths with respect to charge drop schedules, see Figure~\ref{fig:bounded_and_funnel_example}(c)-(d). 

\begin{figure}[t]
    \centering
\includegraphics[width=1.00\textwidth]{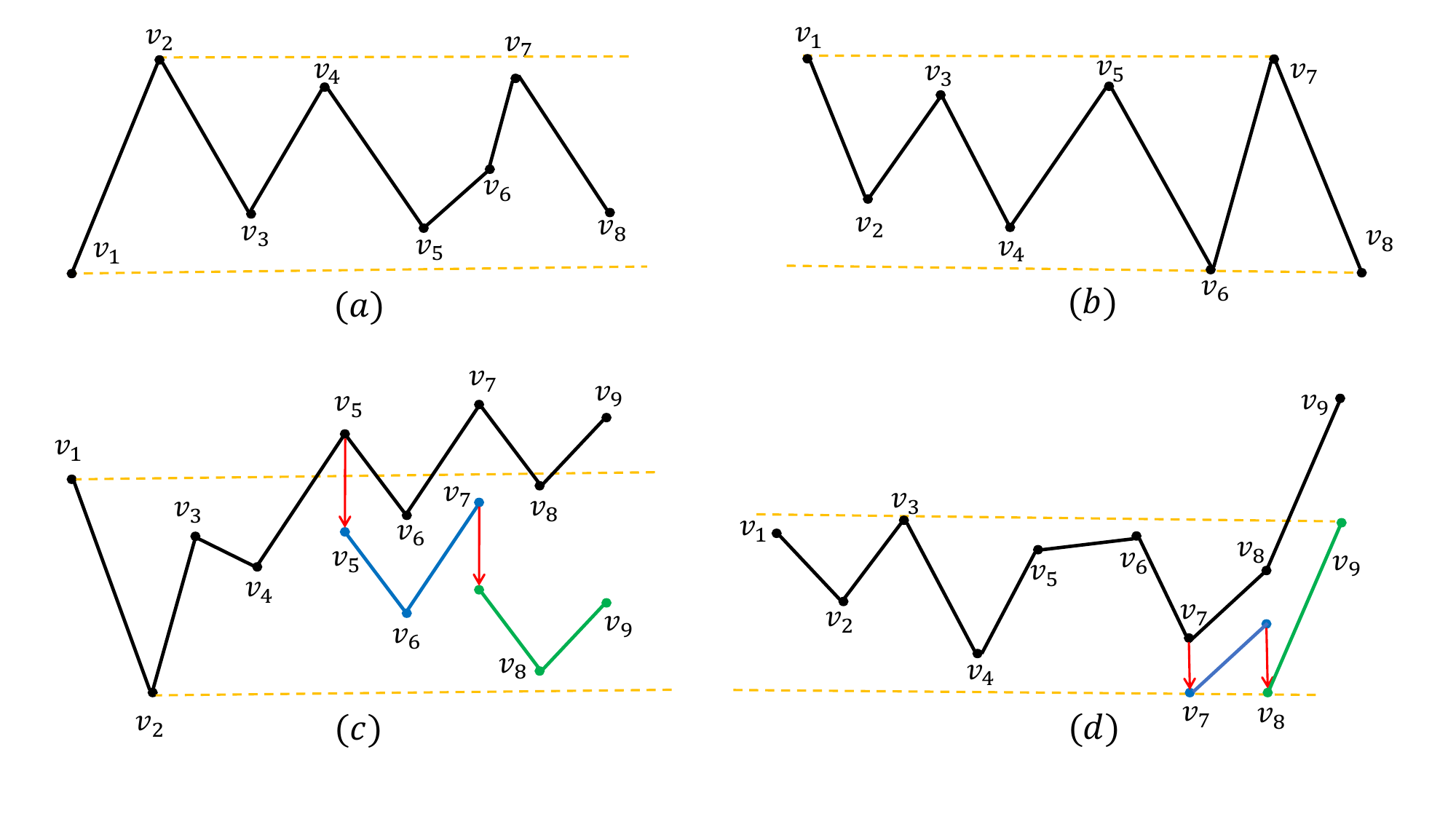}
\caption{
On the top: First-arc (left) and last-arc (right) bounded paths. Both paths are arc bounded paths with respect to the zero schedule. 
On the bottom: First-arc (left) and last-arc (right) bounded paths with respect to different charge drop schedules.}
\label{fig:bounded_and_funnel_example}
\end{figure}

\begin{definition}[Arc-bounded path]\label{def:arc-bounded}
    A path $P=v_1\ldots  v_k$ is first-arc-bounded, or 
    alternatively $v_1 v_2$-bounded with respect to a charge drop schedule $C$ if $C$ does not drop charge at  $v_2$\footnote{Recall Definition~\ref{def:gain} which states that we don't drop charge at $v_1$.} and if one of the following holds
    \begin{itemize}
        \item $g(v_1 v_2) \ge 0 $ and $0 = g^{P,C}_{v_1} \le g^{P,C}_{v_i} \le g^{P,C}_{v_2} = g(v_1 v_2)$, for every $1\le i\le k$.  We say that $P$ is a $ \overunderline{{v_1}{v_2}{v_k}}{2-2}{1-1}$ path with respect to $C$.
        \item  $g(v_1 v_2 )\le 0$ and $g(v_1 v_2)=g^{P,C}_{v_2} \le g^{P,C}_{v_i} \le g^{P,C}_{v_1} = 0$, for every $1\le i\le k$. We say that $P$ is a $\overunderline{{v_1}{v_2}{v_k}}{1-1}{2-2}$ path with respect to $C$.
    \end{itemize}
    Similarly, $P$ is last arc-bounded, or alternatively $v_{k-1} v_k$-bounded with respect to $C$ if $C$ does not drop charge at $v_1$ and  $v_{k-1}$ and $v_k$ and if one of the following holds
    \begin{itemize}
        \item $g^{C}(v_{k-1} v_k) \ge 0 $ and $g^{P,C}_{v_{k-1}} \le g^{P,C}_{v_i} \le g^{P,C}_{v_k}$, for every $1\le i\le k$.  We say that $P$ is a $\overunderline{{v_1}{v_{k-1}}{v_k}}{3-3}{2-2}$ path with respect to $C$.
        \item  $g^C(v_{k-1} v_k )<0$ and $g^{P,C}_{v_k} \le g^{P,C}_{v_i} \le g^{P,C}_{v_{k-1}}$, for every $1\le i\le k$. We say that $P$ is a $\overunderline{{v_1}{v_{k-1}}{v_k}}{2-2}{3-3}$ path with respect to $C$.
    \end{itemize}
     We say that $P$ is arc-bounded if it is either first-arc-bounded or last-arc-bounded. We say that $P$ is negative arc-bounded if the ``bounding" arc is of negative gain.
\end{definition}

\subsection{Funnels}

The following definition defines the structure \emph{funnel}, see Figure~\ref{fig:example_funnel_monotone}(c)-(f). Funnels are defined with respect to the zero charge drop schedule. 

\begin{definition}[Funnels]\label{def:funnel}
    A path $P$ is said to be a \emph{funnel} if it is arc-bounded with respect to the zero schedule 
    and does not contain any monotone path of length $2$ or $3$.
\end{definition}

\begin{lemma}\label{lemma:funnel-zigzag-structure}
    Let $P = v_0 \ldots v_k$ and denote $e_i =  v_{i-1}v_i$ for $i=1,\ldots k$.  $P$ is a funnel if and only if the following two conditions hold.
    \begin{enumerate}
        \item 
        \begin{itemize}
        \item If $P$ is $e_1$-bounded then $|g(e_1)|\ge|g(e_2)| > \ldots |g(e_{k})|>0$, or
        \item If $P$ is $e_{k}$-bounded then $|g(e_{k})| \ge |g(e_{k-1})| > \ldots > |g(e_1)| > 0$.
    \end{itemize}
    Note that all inequalities are strict except the first.
    \item The sign of the arc gains are alternating, i.e., $sign(g(e_i)) = (-1)^{i+1}\cdot sign(g(e_1))$ for
    every $1\le i \le k$.
    \end{enumerate}
\end{lemma}

\begin{proof}
    Assume $P$ is a funnel and that it is $e_1$ bounded.
    The proof for the case that $P$ is $e_{k}$-bounded is symmetric.
    The second property is immediate since a funnel does not contain $2$-shortcuts. Since $P$ is $e_1$-bounded it follows that $|g(e_1)| \ge |g(e_2)|$. We prove by induction that $|g(e_i)| > |g(e_{i+1})|$ for $1<i\le k-1$. The base case ($|g(e_2)| > |g(e_3)|$, note the strict inequality) follows since otherwise $e_1e_2e_3$ is a $3$-shortcut.  The inductive step is similar. 

    For the other direction, assume $P$ is $e_1$-bounded and $|g(e_1)|\ge|g(e_2)| > \ldots |g(e_{k})| > 0$ and also  $sign(g(e_i)) = (-1)^{i+1}\cdot sign(g(e_1))$ for every $1\le i\le k$. The second property guarantees that $P$ does not contain $2$-shortcuts and together with the first property we get that $P$ does not contain $3$-shortcuts. 
\end{proof}

As an immediate corollary of the above structural lemma, we observe that a subpath of a funnel is also a funnel.

\section{Relating the Path Structures to \texorpdfstring{$\alpha(\cdot, \cdot)$}{Alpha(.,.)}}
We present several lemmas that relate monotone paths, arc-bounded paths and funnels to the $\alpha(\cdot,\cdot)$ values of $G$. We start with the following lemma that characterizes traversable paths. It states that a path is traversable if and only if it has no subpath that loses more than $B$ gain (charge). Moreover, the lemma shows how to calculate $\alpha_b(P)$ of a path $P$, where $b\in [0,B]$, using the largest gain of a vertex on $P$ and $g(P)$, see Figure~\ref{fig:alpha_general}.

\begin{figure}
    \centering
    \includegraphics[width=0.8\textwidth]{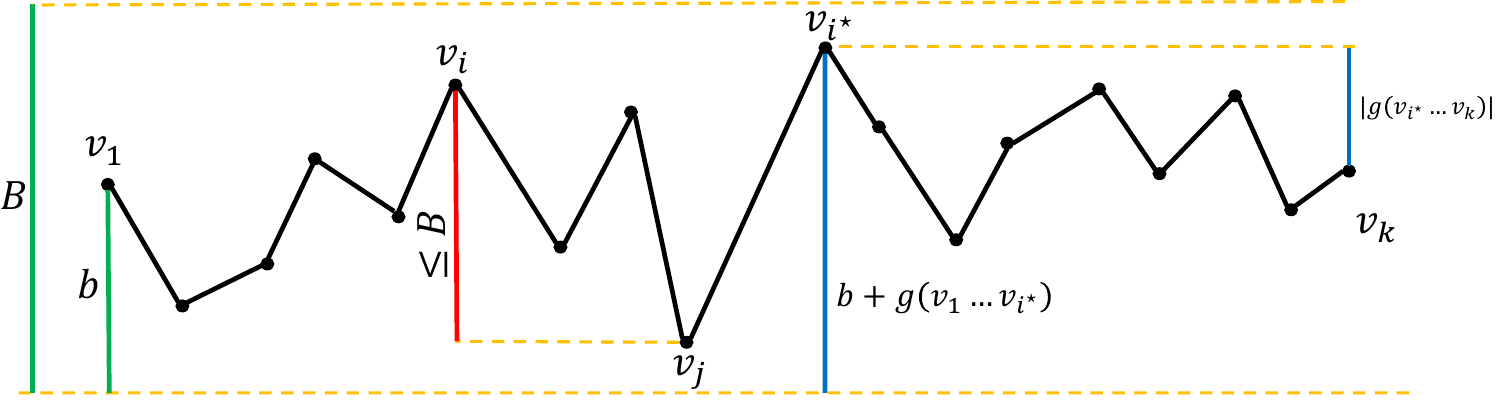}
    \par\smallskip 
    (a)
    
    \vspace{0.5cm} 
    
    \includegraphics[width=0.8\textwidth]{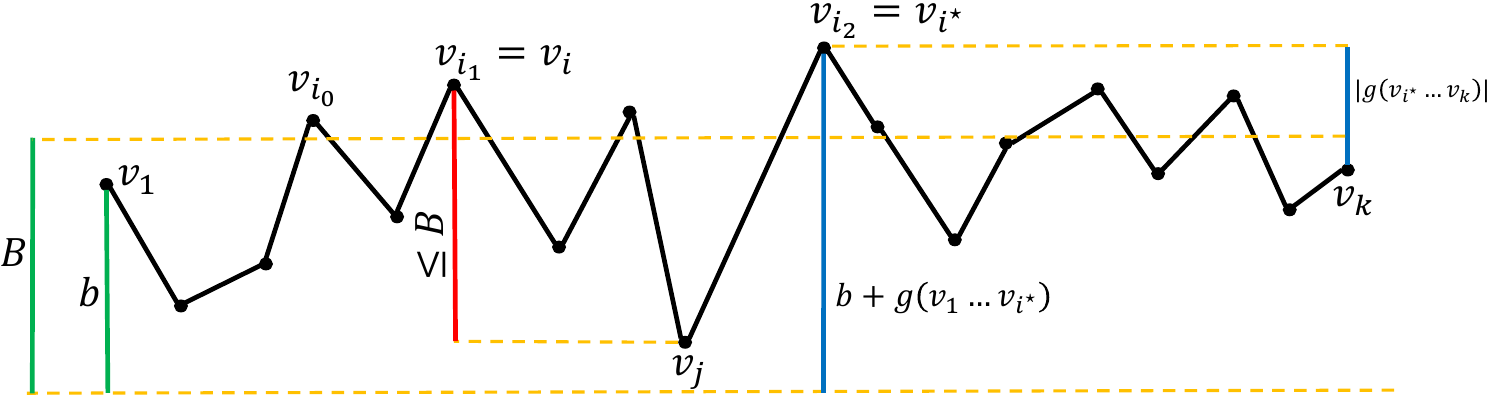}
    \par\smallskip 
    (b)
\caption{Illustration of Lemma~\ref{lem:alpha_general}. We start at $v_1$ with $b$ charge. The subpath $v_i\ldots v_j$ has the lowest gain $g(v_i\ldots v_j) = \min_{i'<j'} g(v_{i'}\ldots v_{j'})$. As depicted, $|g(v_i\ldots v_j)|\le B$. Moreover, every prefix $v_1\ldots v_t$ has gain $g(v_1\ldots v_t)\ge -b$ and therefore $
\alpha_b(P)\neq -\infty$. The vertex of maximum gain is $v_{i^\star} = \argmax_{v_i} g(v_1\ldots v_i)$. If $b+g(v_1\ldots v_{i^\star})\le B$, see Figure~$(a)$, then $\alpha_b(v_1\ldots v_{i^\star}) = b + g(v_1\ldots v_{i^\star})$. Otherwise, see Figure~$(b)$, $\alpha_b(v_1\ldots v_{i^\star}) = B$. In both cases $\alpha_b(v_1\ldots v_{k}) = \alpha_b(v_1\ldots v_{i^\star}) + g( v_{i^\star} \ldots v_k)$.}  
\label{fig:alpha_general}
\end{figure}

\begin{lemma}\label{lem:alpha_general}
 Let $P = v_1\ldots v_k$ and $b\in [0,B]$. Then
 \begin{itemize}
     \item $\alpha_b(P) \ge 0 $ if and only if for every $1\le j \le k$ it holds that $g(v_1\ldots v_{j})\ge -b$ and for every $1\le j_1 \le j_2 \le k$ it holds that $g(v_{j_1}\ldots v_{j_2})\ge -B$.
     \item If $\alpha_b(P) \ge 0 $ then $\alpha_b(P) =  \min \{B, b + g(v_1\ldots v_{i^\star}) \} + g(v_{i^\star} \ldots v_k)$, where $v_{i^\star} = \argmax_{v_i}g(v_1\ldots v_i)$. 
 \end{itemize}
\end{lemma}

\begin{proof}
    We begin by proving the first claim. The first direction (in which we assume $\alpha_b(P)\ge 0$) is trivial. Assume that for every $1\le j \le k$ it holds that $g(v_1\ldots v_{j})\ge -b$ and for every $1\le j_1 \le j_2 \le k$ it holds that $g(v_{j_1}\ldots v_{j_2})\ge -B$. We split into the following cases.

    \textbf{Case $1$: $b + g_{v_i}< B$ for $i=1,\ldots ,k$ (Figure~\ref{fig:alpha_general}(a)):}
    We prove by induction on $i=1,\ldots, k $ that $\alpha_b(v_1\ldots v_i) = b + g_{v_i}$. The base case is immediate $\alpha_b(v_1) = b = b+ g_{v_1}$. Let $1 < i \le k$ and assume that $\alpha_b(v_1\ldots v_{i-1}) = b + g_{v_{i-1}}$. Note that
    $$\alpha_b(v_1\ldots v_{i-1})  + g(v_{i-1} v_i) = b+ g_{v_{i-1}} + g(v_{i-1} v_i) = b + g_{v_i} \ge 0,$$
    where the last inequality holds by the assumption. It follows that $\alpha_b(v_1\ldots v_i)\ge 0$. Since in this case we assume $b+g_{v_i} < B$ for all $i=1,\ldots,k$, it follows that $\alpha_b(v_1\ldots v_{i}) = b + g_{v_{i}}$. We conclude that $\alpha_b(P) = b + g_{v_k} \ge 0$.  

    \textbf{Case $2$: There is an $1\le i \le k$ such that $b + g_{v_i} \ge B$ (Figure~\ref{fig:alpha_general}(b)):} 
    Let $1 \le i_0 \le k$ be minimal such that $b+ g_{v_{i_0}} \ge B$. Similarly to Case $1$, we get that for every $1\le i < i_0$,
    $\alpha_b(v_1 \ldots v_i) = b + g_{v_i}$. Note that 
    $$\alpha_b(v_1 \ldots v_{i_0 - 1}) + g(v_{i_0-1} g_{v_{i_0}})
    = b+ g_{v_{i_0-1}} + g(v_{i_0-1} v_{i_0}) 
    = b+ g_{v_{i_0}} \ge B, 
    $$
    and therefore $\alpha_b(v_1 \ldots v_{i_0}) = B$. Let $i_1 > i_0$ be minimal such that $g_{v_{i_1}} \ge g_{v_{i_0}}$.
    Note that for every $i_0< j < i_1$, it holds that 
    \[
    0 \le  B + g(v_{i_0}\ldots v_j) = B + (g_{v_j} - g_{v_{i_0}}) \le B,
    \]
    where the first inequality follows by the statement of the lemma and the last inequality follows since $i_0 < j  < i_1$.    
    Thus, for every $i_0< j < i_1$, it holds that $\alpha_b(v_1\ldots v_j) = B + ( g_{v_j} - g_{v_{i_0}}) $ and $\alpha_b(v_1\ldots v_{i_1}) = B$. By continuing this process we  get a sequence of indices $i_0<i_1< \ldots <i_t$ for which $\alpha_b(v_1 \ldots v_{i_j}) = B$, for every $1\le j \le t$, and $v_{i_{t}}$ has the largest gain in $P$ (that is $i_t=i^\star$ from the second statement of the lemma). We prove by induction on $i$, for $i_t \le i \le k$, that $\alpha_b(v_1 \ldots v_i) = B +   ( g_{v_i} - g_{v_{i_t}}) \ge 0$ and in particular $\alpha_b(P) \ge 0$. The base case $i = i_t$ holds since we already proved that $\alpha_b(v_1\ldots v_{i_t})=B$. Let $i_t < i \le k$. By the inductive hypothesis, we get that
    \[
    \alpha_b(v_1 \ldots v_{i-1}) + g(v_{i-1} v_i) = B + ( g_{v_{i-1}}-g_{v_{i_t}}) + g(v_{i-1} v_i) = B + ( g_{v_i} - g_{v_{i_t}} ) = B + g(v_{i_t} \ldots v_i) \ge 0, 
    \]
    where the inequality holds by the statement of the lemma. Thus, $\alpha_b(v_1 \ldots v_{i}) \ge 0$. Moreover 
    $\alpha_b(v_1 \ldots v_{i-1}) + g(v_{i-1} v_i) = B + ( g_{v_i} - g_{v_{i_t}} ) \le B$ and therefore 
    $\alpha_b(v_1 \ldots v_{i}) = B + (g_{v_i} - g_{v_{i_t}})$.
    
    We now prove the second statement of the lemma. Assume $\alpha_b(P) \ge 0$. We split to the same cases as before.

    \textbf{Case $1$: $b + g_{v_i}< B$ for $i=1,\ldots ,k$ (Figure~\ref{fig:alpha_general}(a)):}
    As we have seen $\alpha_b(P) = b + g(P)$. Thus
    \begin{align*}
        \alpha_b(P) &= b + g(P) 
        \\&=  b + g(v_1\ldots v_{i^\star}) + g(v_{i^\star}\ldots v_k) 
        \\&= \min \{B, b + g(v_1\ldots v_{i^\star})\} + g(v_{{i^\star}}\ldots v_k).
    \end{align*}

    \textbf{Case $2$: there is $1\le i \le k$ such that $b + g_{v_i} \ge B$ (Figure~\ref{fig:alpha_general}(b)):} 
    Recall the sequence of prefix maxima $v_{i_1},\ldots v_{i_t}$ on $P$ with respect to the gains from the first part of the proof. It follows that $i^\star = i_t$ and we proved that   $\alpha_b(P) = B + (g_{v_k} - g_{v_{i^\star}})$. Therefore
    \begin{align*}
        \alpha_b(P) = B +(g_{v_k} - g_{v_{i^\star}})
        = B + g(v_{i^\star} \ldots v_k)
        = \min\{B,  b+g_{v_{i^\star}}\} + g(v_{i^\star} \ldots v_k).
    \end{align*}
\end{proof}

The following lemma is used extensively in order to lower bound $\alpha_b(P)$, for a monotone path $P$ and $b\in [0,B]$, by the gain of $P$. 

\begin{lemma}\label{lemma:alpha-of-monotone}
    Let $P=v_1\ldots v_k$ be a monotone path with respect to a charge drop schedule $C$. Let $b\in [0,B]$.
    \begin{itemize}
        \item If $P$ is descending with respect to $C$ and $g^C(P)\ge -b$, then   $\alpha_b(P) \ge b+g^C(P)$.  
        \item If $P$ is ascending with respect to $C$, then  $\alpha_b(P) = \min\{B,b+g(P)\} \ge \min\{B,b+g^C(P)\}$. In particular, $P$ is strongly traversable.
    \end{itemize}
\end{lemma}

\begin{proof}
    We begin by proving the first claim. Since $P$ is descending with respect to $C$, we get that for every $1\le i \le k$, it holds that $g_{v_{i}} \ge g^{P,C}_{v_i} \ge g^C(P) \ge -b$. Let $1\le j_1 \le j_2 \le k$. Observe that
    \[g(v_{j_1} \ldots v_{j_2}) \ge g^C(v_{j_1} \ldots v_{j_2})  = g^{P,C}_{v_{j_2}} - g^{P,C}_{v_{j_1}} 
    \ge g^{P,C}_{v_{k}} - g^{P,C}_{v_{1}} = g^C(P) \ge -b \ge -B.
    \]
    Therefore, by Lemma~\ref{lem:alpha_general}, $\alpha_b(P)\ge 0$. Let $v_{i^\star}$ be the vertex with the largest gain $g_{v_{i^\star}}$ in $P$. By Lemma~\ref{lem:alpha_general}, it holds that 
    \begin{align*}
        \alpha_b(P) 
        &=  \min \{B, b + g(v_1\ldots v_{i^\star}) \} + g(v_{i^\star} \ldots v_k) 
        \\ &\ge
        \min \{B, b + g^C(v_1\ldots v_{i^\star}) \} + g^C(v_{i^\star} \ldots v_k)
        \\&\numeq{1}
        b + g^C(v_1\ldots v_{i^\star})  + g^C(v_{i^\star} \ldots v_k) 
        \\&= b + g^C(P),
    \end{align*}    
    where Equality $(1)$ holds since $P$ is descending with respect to $C$, so  $g^C(v_1\ldots v_{i^\star}) \le g^{P,C}_{v_1} = 0$.
    
    We now prove the second claim.
    Let $P=v_1\ldots v_k$ be an ascending path with respect to $C$. By Lemma~\ref{lem:ascending_no_CDS}, it follows that $P$ is ascending with respect to the zero schedule.
    Since $P$ is traversable, it follows by Lemma~\ref{lem:alpha_general} that for every $1\le j_1 \le j_2 \le k$ it holds that $g(v_{j_1}\ldots g_{v_{j_2}}) \ge -B$. Since $P$ is ascending, for every $1\le i \le k$ it holds that $g_{v_{i}} \ge g_{v_1} = 0 \ge -b$. Therefore,
    by Lemma~\ref{lem:alpha_general}, we get that 
        $\alpha_b(P) = \min \{B, b+ g(v_1\ldots v_k) \}$ (note that since $P$ is ascending, then $v_k$ is the vertex of largest gain in $P$).
\end{proof}

Let $P= v_1 \ldots v_k$ be a negative arc-bounded path. The following lemma states that if the (negative) bounding arc of $P$ has gain at least $-B$ then $P$ us traversable.

\begin{lemma}\label{lemma:alpha-of-arc-bounded}
    Let $P=v_1\ldots v_k$ be a negative arc-bounded path with respect to a charge drop schedule $C$. Let $b\in [0,B]$.
    \begin{itemize}
        \item If $P$ is first arc-bounded with respect to $C$ and $g(v_1 v_2)\ge -b$, 
        then $\alpha_b(P) \ge b+g^C(P)$.
        \item If $P$ is last arc-bounded with respect to $C$ and $g(v_{k-1} v_k)\ge -B$ and $ g^C(P)\ge -b$, then   $\alpha_b(P) \ge \min \{  b+g^C(P), B + g(v_{k-1} v_k) \}$.
    \end{itemize}
\end{lemma}

\begin{proof}
     We begin by proving the first claim. Since $P$ is first-arc bounded with respect to $C$, we get that for every $1\le i \le k$, it holds that $g_{v_i} \ge g^{P,C}_{v_i} \ge g^{P,C}_{v_2} = g(v_1 v_2) \ge -b$. Let $1\le j_1 \le j_2 \le k$. Observe that
    \[g(v_{j_1} \ldots v_{j_2}) \ge g^C(v_{j_1} \ldots v_{j_2})  = g^{P,C}_{v_{j_2}} - g^{P,C}_{v_{j_1}} 
    \numge{1} g^{P,C}_{v_{2}} - g^{P,C}_{v_{1}} = g(v_1 v_2) -0\ge  -b \ge -B,
    \]
    where Inequality $(1)$ follows since $P$ is first-arc bounded with respect to $C$.
    Therefore, by Lemma~\ref{lem:alpha_general}, $\alpha_b(P)\ge 0$. Let $v_{i^\star}$ be the vertex with the largest gain $g_{v_{i^\star}}$ in $P$. By Lemma~\ref{lem:alpha_general}, it holds that 
    \begin{align*}
        \alpha_b(P) 
        &=  \min \{B, b + g(v_1\ldots v_{i^\star}) \} + g(v_{i^\star} \ldots v_k) 
        \\ &\ge
        \min \{B, b + g^C(v_1\ldots v_{i^\star}) \} + g^C(v_{i^\star} \ldots v_k)
        \\&\numeq{1}
        b + g^C(v_1\ldots v_{i^\star})  + g^C(v_{i^\star} \ldots v_k) 
        \\&= b + g^C(P),
    \end{align*}
    where Equality $(1)$ holds since  $g^C(v_1\ldots v_{i^\star}) \le g^{P,C}_{v_1} = 0$.

    We now prove the second claim. Assume $P$ is last arc bounded with respect to $C$. Since $P$ is $v_{k-1}v_k$-bounded (with respect to $C$) and $g(v_{k-1}v_k) \ge -B$, it follows that there is no subpath $v_{j_1}\ldots v_{j_2}$ of $P$ of gain $g(v_{j_1}\ldots v_{j_2}) < - B$.\footnote{Indeed, by contradiction assume such $j_1 < j_2$ exist. Then $g^C(v_{j_1}\ldots v_{j_2}) \le g(v_{j_1}\ldots v_{j_2}) < - B$ so $|g^{C,P}_{v_{j_1}} - g^{C,P}_{v_{j_2}}| > B$. Together with the inequality $|g^{C,P}_{v_{k-1}} - g^{C,P}_{v_{k}}| \le B$, we get a contradiction to $P$ being arc-bounded with respect to $C$.} Moreover, since $P$ is last arc-bounded, for every $1\le  i \le k$ it holds that $g^P_{v_i} \ge g^{P,C}_{v_k} = g^C(P) \ge -b$, where the last inequality holds by the assumption of the lemma. 
    Thus, by Lemma~\ref{lem:alpha_general}, $\alpha_b(P) \ge 0$.
    Since $P$ is  negative arc-bounded with respect to $C$, it follows that $g^{P,C}_{v_i} \le g^{P,C}_{v_{k-1}}$ for every $i=1,\ldots,k$. In particular, since $v_{k-1}$ accumulated the largest charge drop, 
    we get that $g^{P}_{v_i} \le g^{P}_{v_{k-1}}$ for every $i=1,\ldots,k$.
    Therefore, by Lemma~\ref{lem:alpha_general}, 
    \begin{align*}
    \alpha_b(P) &=  \min \{B, b + g(v_1\ldots v_{k-1}) \} + g(v_{k-1} v_k) \\
    &= \min \{B + g(v_{k-1} v_k), b + g(v_1\ldots v_{k-1}) + g(v_{k-1}v_k)\} \\
    &= \min \{B + g(v_{k-1} v_k), b + g(P) \}.        
    \end{align*}
\end{proof}

The following structural definition and lemma are from Dorfman et al.~\cite{DorfmanKTZ23}.

\begin{definition}[Entry-exit pairs~\cite{DorfmanKTZ23}]\label{D-entry-exit}
Let $C$ be a positive gain cycle in $G=(V,A,g)$ and let $B$ be the capacity of the battery. A pair of vertices $(x,y)$ on~$C$ is an \emph{entry-exit} pair of~$C$ if the car can start at~$x$ with an empty battery and eventually get to~$y$, possibly after going several times around the cycle, with a full battery, i.e., with a charge of~$B$.
\end{definition}
 
The following lemma characterise the structure of optimal paths, see Figure~\ref{fig:opt-path}.

\begin{figure}[t]
    \centering
    \includegraphics[scale=0.45]{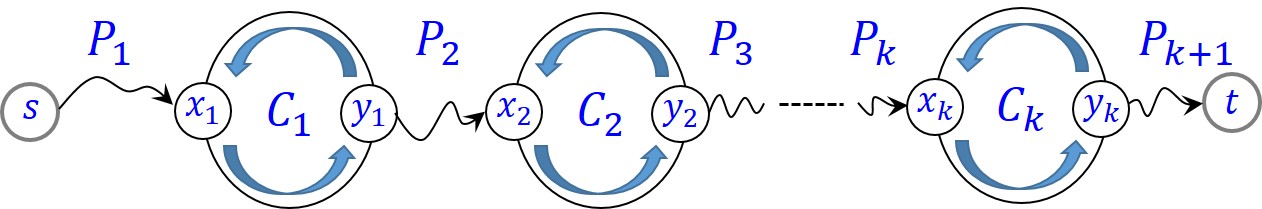}
    \caption{Generic structure of minimum energetic paths in the presence of negative cycles. If $\alpha_b(s,t)\ge -\infty$, then there is a minimum energetic path from~$s$ to~$t$ of the form shown, where $C_1,\ldots,C_k$ are simple negative cycles and $(x_i,y_i)$ is an entry-exit pair on~$C_i$, for $i=1,2,\ldots,k$. All entries $x_1,x_2,\ldots,x_k$ are distinct and all exits $y_1,y_2,\ldots,y_k$ are distinct. The paths $P_1,P_2,\ldots,P_{k+1}$ are simple but necessarily disjoint from the cycles $C_1,C_2,\ldots,C_k$.
    }
    \label{fig:opt-path}
\end{figure}

\begin{lemma}[Lemma 2.6 of \cite{DorfmanKTZ23}]\label{lemma:optimal-structure}
 If there is a traversable path~$P$ from $s$ to $t$ in $G$, then there is a traversable path~$P'$ from $s$ to $t$ such that $\alpha_b(P')\ge \alpha_b(P)$, for every $b \in [0,B]$, where $P'$ has the following form: either $P'$ is simple, or there is a sequence $C_1,C_2,\ldots,C_k$ of simple positive gain cycles, where $k<n$, with entry-exit pairs $(x_1,y_1),(x_2,y_2),\ldots,(x_k,y_k)$ on them, such that $P'$ is composed of a simple path from~$s$ to~$x_1$, followed by sufficiently many traversals of~$C_1$ that end in~$y_1$ with a full battery, followed by a simple path from~$y_1$ to~$x_2$, followed by sufficiently many traversals of~$C_2$ that end in~$y_2$ with a full battery, and so on, and finally a simple path from~$y_k$ to~$t$. Furthermore, all entries $x_1,x_2,\ldots,x_k$ are distinct, and all exits $y_1,y_2,\ldots,y_k$ are distinct.
\end{lemma}

\section{Overview of the Algorithm}\label{S-overview}
The algorithm is composed of two stages. The goal of first stage, which is described in Appendix~\ref{S-shortucts}, is to store information about shortcuts that correspond to simple paths. In the second stage, which is described in Appendix~\ref{S-algorithm}, we use the stored information and compute $\alpha_B(s,t)$ for every $s,t\in V$.

\subsection{Stage I}
This stage performs $O\left(n^\alpha\log^2 n\right)$ iterations.
To clarify the presentation we partition these iterations into 
$O\left(\log^2 n\right)$ outer-iterations, which are performed in $\ComputeS$ (see Figure~\ref{figure:compute-shortcuts}), where each of them calls the procedure $\UpdateS(M)$ which performs 
$O(n^\alpha)$
inner-iterations.

In each inner-iteration we take $M\in \RR^{n\times n}$, our current table of shortcuts, and improve it several times. These improvements happen using two procedures $\Simple$ and $\LS$, which take~$M$ and return an improved table $M'$. Both procedures perform computations solely on $G^M$, which is the complete graph whose arc gains are defined according to $M$, i.e., $g(uv)=M[u][v]$ for every $u,v\in V$.

At inner-iteration $i$, we store the shortcuts (and more information) in a data structure $D$. The data structure is a union of $3$ tables. Let $x,y,z\in V$ and let $G^M$ be the current graph of interest. The values in $D$ are defined with respect to $G^M$.

\begin{itemize}
    \item $D[x][y]$ is the maximum gain of a monotone path from $x$ to $y$ we have encountered.\footnote{This value might be negative, or even $-\infty$. This is true also for the next entries.}
    \item $D[xy][z]$ is the maximum gain of a $\overunderline{xyz}{1-1}{2-2}$ path or a $\overunderline{xyz}{2-2}{1-1}$ path. 
    \item $D[x][yz]$ is the maximum gain of a $\overunderline{xyz}{2-2}{3-3}$ path or a $\overunderline{xyz}{3-3}{2-2}$ path.
\end{itemize}

We show in Corollary~\ref{cor:M-G-connection} that these values also correspond to paths in $G$ with at least as much gain.

The following definition  helps us to measure the quality of the values stored in $D$

\begin{definition}
    Let $P=v_1 \ldots v_k$ be a path in $G^M$. We say that the data structure $D$ \emph{dominates} $P$ with respect to a charge drop schedule $C$ if 
    \begin{itemize}
        \item If $P$ is $v_1 v_2$-bounded with respect to $C$, then $D[v_1 v_2][v_k] \ge g^C(P)$.
        \item If $P$ is $v_{k-1} v_k$-bounded with respect to $C$, then $D[v_1][v_{k-1} v_k] \ge g^C(P)$.
        \item If $P$ is monotone with respect to $C$, then $D[v_1][v_k]\ge g^C(P)$.
    \end{itemize}
    If $C$ is the zero schedule we just say that $D$ dominates $P$.
\end{definition}

Let $M$ be the final table of shortcuts computed during Stage I.
Theorem~\ref{theorem:shortcut}, states that for any simple monotone path $P=v_0\ldots v_k$ in $G$, w.h.p.,  $M[v_0][v_k]\ge g(P)$. 
That is, the gain of the arc $(v_0,v_k)$ in
$G^M$ is larger than 
$g(P)$.
This theorem follows from 
Lemma~\ref{lemma:paths-shrink} which shows 
that if $P$ is a monotone path from $v$ to $w$ in $G^{M_i}$
where $M_i$ is the shortcuts table at the beginning of outer-iteration~$i$,
 then there exists a monotone path $P'$ from $v$ to $w$ in $G^{M_{i+1}}$,
where 
$M_{i+1}$ is the shortcuts table at the beginning of outer-iteration $i+1$,
such that  $g(P')\ge g(P)$
and 
 $|P'|=\left(1-\Omega\left( \frac{1}{\log n}\right) \right)|P|$.

\subsection{Stage II}
We begin by utilizing the shortcuts obtained from Stage I and build an auxiliary graph $H = (V^0 \cup V^B, E(H))$, where $V^b = \{v^b \mid v\in V \}$ for $b=0,B$ represents that we are at $v$ with at least $b$ charge.  An arc $u^{b_1}v^{b_2}\in E(H)$ represents that $\alpha_{b_1}(u,v)\ge b_2$. We add to $E(H)$ arcs that we can easily deduce by the shortcuts of Stage I. We compute the transitive closure $H^\star$  of $H$ that has even stronger relations. 
We prove in Theorem~\ref{theorem:B-B-strong-full}, that for every $s,t\in V$ we have $\alpha_B(s,t)=B$ if and only if $s^B t^B \in E(H^\star)$.

Recall the ``cycle-hopping" structure of optimal cycles given by Lemma~\ref{lemma:optimal-structure}. Let $s,t\in V$ and let~$P$ be an optimal path from $s$ to $t$ (i.e., $\alpha_B(P)=\alpha_B(s,t)$) that is structured as in Lemma~\ref{lemma:optimal-structure}. Let $(x_1,y_1),\ldots (x_k,y_k)$  be the entry-exit pairs as in Lemma~\ref{lemma:optimal-structure}. By the  discussion above, $s^B y_k^B \in E(H^\star)$, thus $H^\star$ allows us to skip all of the cycle-hoping and focus on the last path from $y_k$ to $t$. This path is simple and we start traversing it with  full charge. We prove in Lemma~\ref{lemma:alpha-simple-B-0} that $\alpha_B(y_k,t)$ can be derived from a value in $D$ that correspond to a funnel in $G^M$, where $M$ is taken from the last iteration of Stage I.

\section{Stage I - Algorithm for finding 
shortcuts}\label{S-shortucts}
The goal of algorithm $\ComputeS(G)$  (see Figure~\ref{figure:compute-shortcuts}) is to find shortcuts corresponding to monotone simple paths in $G$. The algorithm proceeds in $\log^2 n$ outer-iterations. Each iteration is implemented using the procedure $\UpdateS(M)$, which gets the shortcuts table $M$ of the previous iteration and computes new shortcuts, based on the graph defined by $M$. We claim (see Lemma~\ref{lemma:paths-shrink}) that in each iteration, a monotone path $P$, consisting of shortcuts of the previous iteration, can be replaced by a shorter path $Q$ (consisting of new shortcuts) of length shorter by a factor of $1-\frac{c}{\log n}$, for some constant $c$.

\begin{figure}[t]
\begin{center}
\hspace*{-0.5cm}
\begin{minipage}{2.6in}
\begin{algorithm}[H]
  \Fn{$\ComputeS(G=(V,A,g))$}{
  \BlankLine
  $M\gets ConstMarix(n,n,-\infty)$ \;
  \BlankLine
  \For{$i=1\ldots n$}{
  $M[i][i] \gets 0$\;
  \For{$(i,j)\in E$}{
    $M[i][j] \gets g(i,j)$
  }
  }
  \For{$t=1\ldots \Theta(\log^2 n)$}{
    $M\gets \UpdateS(M)$
  }
  \BlankLine
  \Return $M$
  }
\end{algorithm}
\end{minipage}
\hspace*{-0.45cm}
\begin{minipage}{3.5in}
\begin{algorithm}[H]
  \Fn{$\UpdateS(M)$}{
  \BlankLine
  $r \gets \Theta(n^\alpha)$ \;
  $M' \gets M$\;
  \For{$i=1\ldots r$}{
  $rand \sim U[0,1]$\;
  \If{$rand < \log (n)/r$}{ 
  $M' \gets \max\{M',\LS(M)\}$\;
  }
  $M \gets \Simple(M)$\;
  }
  \Return $\max\{M,M'\}$
  }
\end{algorithm}
\end{minipage}
\end{center}
\caption{Main procedure. Finds shortcuts corresponding to monotone simple paths. It combines many rounds of finding short shortcuts, with rare applications of finding shortcuts that correspond to longer paths (``long shortcuts")}\label{figure:compute-shortcuts}\label{figure:update-shortcuts}
\end{figure}

The procedure $\UpdateS(M)$, which implements an outer-iteration, proceeds as follows. We perform $\tilde{O}(n^\alpha)$ rounds, which we view as inner-iterations, where $\alpha= 0.5$ is a constant that we  set later. 
In each round, we call the procedure $\Simple(M)$ which finds all short shortcuts in~$G^M$ and updates $M$ accordingly, see Figure~\ref{figure:simple-shortcuts}. Also, with a small probability $p= \tilde{\Theta}\left( n^{-\alpha} \right)$ during a round, we also call the procedure $\LS(M)$ which aims to find some $k$-shortcuts in~$G^M$, where $k>3$.
This procedure also updates $M$.

Intuitively, given a monotone path $P$ in $G^M$, $\UpdateS(M)$ aims to reduce its length by computing shortcuts in $G^M$ that can replace monotone subpaths of $P$. Let $P_1,\ldots, P_{n^\alpha}$ be the corresponding shortcutted versions of $P$ with respect to the updated shortcuts tables $M_1,\ldots, M_{n^\alpha}$.  If we succeeded in reducing the length of $P$ (say by a constant factor) by the $\tilde{O}(n^\alpha)$ applications of $\Simple$ (i.e. $|P_{n^\alpha}| \le c|P|$), then we achieved our goal. Otherwise, in most of the iterations we did not find many short shortcuts on $P_i$ in $G^{M_i}$, and therefore $P_i$ mostly consists of funnels (and some of them can be long). For this reason, with small probability (enough to ``hit" such a round) we call $\LS$ which finds some shortcuts that correspond to monotone paths that contain funnels. We prove in Lemma~\ref{lemma:long-shortcuts-shrinks}, which is the central lemma of this paper, that these shortcuts are enough.

Each of these procedures computes a data structure
$D$   storing information about monotone and arc-bounded paths in  $G^M$.
At the end of each such call we update $M$ with new shortcuts based on~$D$.
 
The algorithm maintains the following invariant.

\begin{invariant}\label{invariant}
Let $M$ be the shortcuts table of the current inner-iteration. The following holds throughout the inner-iteration:

\begin{enumerate}[label=(\Alph*)]
    \item \label{1a} If $D[xy][z]\neq -\infty$ then there is a traversable path $P = xy \ldots z$ in $G^M$ and a charge drop schedule~$C$ such that $P$ is $xy$-bounded with respect to $C$ and $g^C(P) = D[xy][z]$.
    We say that $P$ is realizing $D[xy][z]$.
    \item \label{1b} If $D[x][yz]\neq -\infty$ then there is a traversable path $P = x \ldots yz$ in $G^M$ and a charge drop schedule~$C$ such that $P$ is $yz$-bounded with respect to $C$ and $g^C(P) \ge D[x][yz]$.
    We say that $P$ is realizing $D[x][yz]$.
    \item \label{1c} If $D[x][y]\neq -\infty$ then there is a traversable path $P = x \ldots y$ in $G^M$ and a charge drop schedule~$C$ such that $P$ is monotone with respect to $C$ and $g^C(P) = D[x][y]$. Moreover, if $D[x][y]\ge 0$, then~$P$ is strongly traversable. We say that $P$ is realizing $D[x][y]$.
\end{enumerate}
    
\end{invariant}

The full details of the procedures $\Simple$ and $\LS$ are explained in Appendices~\ref{section:simple} and \ref{section:long}, respectively.

\subsection{Initializing the data structure}
At the beginning of $\Simple(M)$ and $\LS(M)$, we get a shortcuts table $M$ and initialize the data structure $D$ with respect to $G^M$. This initialization creates trivial paths: For every $x,y\in V$ we set $D[x][y]=M[x][y]$ and $D[x][xy]=D[xy][y] = M[x][y]$.

\subsection{Short Shortcuts}\label{section:simple}
The goal of this procedure is to find shortcuts that dominate all (ascending/descending) monotone paths in $G^M$ of length at most $3$, see Figures~\ref{figure:simple-shortcuts} and~\ref{fig:trivial_shortcuts}. Finding $2$-shortcuts is easy. 
We find them all by simply checking for every triplet $x,y,z\in V$ whether $xyz$ is an ascending path in $G^M$, see Figure~\ref{fig:trivial_shortcuts}$(a2)$, and if not, we create a descending shortcut by dropping the right amount of charge, see Figure~\ref{fig:trivial_shortcuts}$(b1)$-$(b3)$. 
We classify monotone paths $xyaz$ of length $3$ according to the eight possibilities for the \emph{sign} of $M[x][y], M[y][a]$ and $M[a][z]$. In seven out of the eight cases we compute shortcuts that dominate these paths similarly to the computation of $2$-shortcuts: For every $x,y,z\in V$ we concatenate the arc $xy$ with the length $2$ shortcut from $y$ to $z$ that were previously computed. We then check whether this results in an ascending shortcut, see Figure~\ref{fig:trivial_shortcuts}$(a1)$, and otherwise we perform charge drops to get a descending shortcut, see Figure~\ref{fig:trivial_shortcuts}$(c1)-(c6)$.
This is done in
 $\Trivial(M,D)$, see Figure~\ref{fig:trivial_shortcuts}. This procedure captures almost all of the possible monotone paths of length at most $3$, except for three cases shown in Figure~\ref{figure:simple-shortcuts}. 

The three special cases of monotone paths $xyaz$ are the following. In all cases the signs of the arc gains alternate between positive and negative

\textbf{Case $1$:} In this case the \emph{sign} pattern of $M[x][y],M[y][a],M[a][z]$ is the same as in Figure~\ref{fig:trivial_shortcuts}$c(3)$. That is, $M[x][y],M[a][z]\ge 0$ and $M[y][a]\le 0$.
Here the path $xyaz$ is ascending and therefore we create an ascending shortcut from $x$ to $z$, see Figure~\ref{figure:simple-shortcuts}$(a)$.

The last two cases are associated with the sign pattern $M[x][y],M[a][z]\le 0$ and $M[y][a]\ge 0$.

\textbf{Case $2$:} The path $xyaz$ satisfies  $g_a\in [g_y,g_x]$. In this case either $g_z \le g_y$, meaning that $xyaz$ is descending, or $g_z \in [g_y,g_a]$. In the latter case we can perform a charge drop at $z$ to make $xyaz$ descending with respect to the appropriate schedule, see Figure~\ref{figure:simple-shortcuts}$(b)$.

\textbf{Case $3$:} The path $xyaz$ satisfies $g_a > g_x (= 0)$. In this case, in order to make $xyaz$ descending, we perform a charge drop at $a$ such that its gain after the drop is the same as the gain of $x$ (which is zero). If after this charge drop $z$ has larger gain than $y$, then we also perform a charge drop at $z$ so that it matches the gain of $y$, see Figure~\ref{figure:simple-shortcuts}$(c)$.

The computation of these three cases of shortcuts is done in $\Simple(M)$ (see Figure~\ref{figure:simple-shortcuts}) as follows. For every triplet $x,y,z\in V$ we do the following. In all cases we aim to compute shortcuts with gains as large as possible.

\begin{figure}[ht!]
\begin{algorithm}[H]
  \Fn{$\Simple(M)$}{
  $D \gets Init\mathtt{-}DS(M)$\;
  \BlankLine
    $\Trivial(M,D)$  \tcp*[h]{\rm Adding to $D$ both $2$-shortcuts and easy $3$-shortcuts in $G^M$}\\
  \For(\tcp*[h]{Creating $2D$ range trees for $yaz$ paths}){$y,z\in V$}{
    $T_{yz}\gets RT(M[y][\cdot],M[y][\cdot]+M[\cdot][z])$ \\
    $T'_{yz}\gets RT(M[y][\cdot], M[\cdot][z])$ 
  }
  \For(\tcp*[h]{\rm $3$-shortcuts}){$x,y,z\in V$}{ 
        \If{$M[x][y]\ge0$}{
        \BlankLine
        $(-,k_2) \gets T_{yz}.range(k_1\in [-M[x][y],0]).max\_ k_2()$ \\ \tcp*{largest gain at $z$ without going below $x$}
        \If(\tcp*[h]{\rm ascending shortcut}){$M[x][y] + k_2 \ge M[x][y]$}{
            $D[x][z] \gets \max\{D[x][z],M[x][y]+k_2\}$
        }
      }
      \If{$M[x][y] \le 0$}{
        \BlankLine
        $(-,k_2) \gets T_{yz}.range(k_1\in[0,|M[x][y]|]).max\_ k_2()$ \\ \tcp*{Largest gain at $z$ without going above $x$}
        \If(\tcp*[h]{\rm descending shortcut}){$M[x][y]\ge M[x][y]+k_2 \ge -B$}{
            $D[x][z] \gets \max\{D[x][z],M[x][y]+k_2\}$
        }
        \BlankLine
        \If(\tcp*[h]{\rm descending shortcut, charge drop is needed at $z$}){$M[x][y]+k_2 \ge M[x][y]$}{
            $D[x][z] \gets \max\{D[x][z],M[x][y]\}$
        }
        \BlankLine
        $(-,M[a][z]) \gets T'_{yz}.range(k_1\in[|M[x][y]|,\infty)).max\_ k_2()$ \\ \tcp*{Largest gain of last arc while going above $x$}
        $D[x][z]\gets  \max\{D[x][z], \min \{M[x][y], M[a][z] \} \}$ \tcp*{Charge drop at $a$ and $z$}
      }
  }
  \BlankLine
  \Return $ D.shortcuts$
  }
\end{algorithm}

\begin{center}
    \includegraphics[width=1\textwidth]{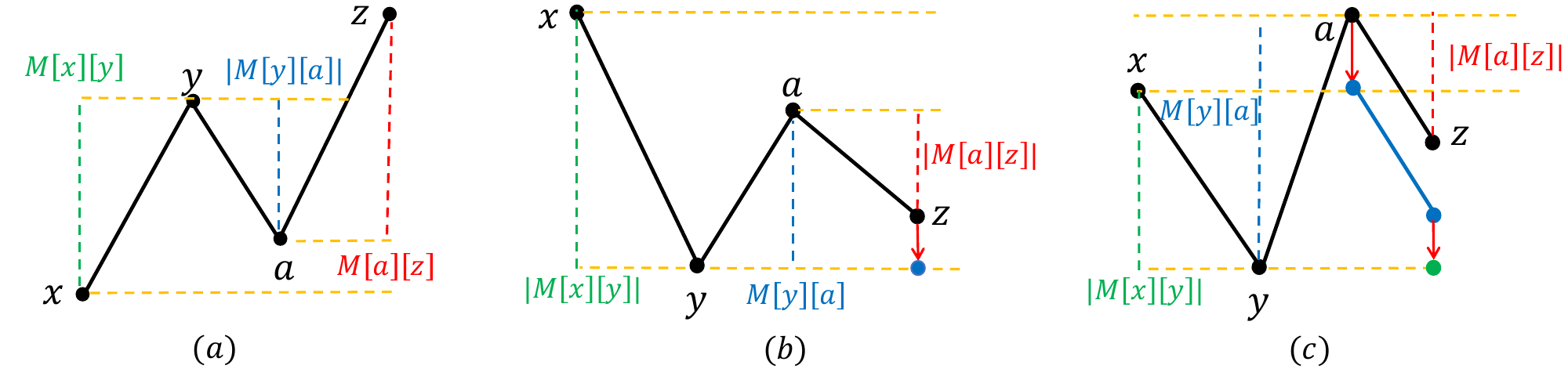}
\end{center}

\caption{Hardest cases to compute $3$-shortcuts. In these cases the signs of the arc gains alternate between positive and negative.}\label{figure:simple-shortcuts}
\end{figure}

\begin{figure}[ht!]
\begin{algorithm}[H]
  \Fn{$\Trivial(M, D)$}{
  $M_2,M_3 \gets ConstMatrix(n,n,-\infty)$ \tcp*{\rm Matrices for shortcuts of paths of length $2$ or $3$}
  \For(\tcp*[h]{\rm $2$-shortcuts}){$x,y,z\in V$}{ 
    \If(\tcp*[h]{\rm ascending shortcuts}){$M[x][y] \ge 0 \wedge M[y][z]\ge 0$}{
        $M_2[x][z] \gets \max\{M_2[x][z],M[x][y]+M[y][z]\}$
    }\Else(\tcp*[h]{\rm descending shortcuts with charge drop}){
    \If{$\min\{M[x][y],0\}+\min\{M[y][z],0\} \ge -B$}{
        $M_2[x][z] \gets \max\{M_2[x][z],\min\{M[x][y],0\}+\min\{M[y][z],0\}\}$
    }
    }
    }
    \For(\tcp*[h]{\rm easy $3$-shortcuts}){$x,y,z\in V$}{ 
      \If(\tcp*[h]{\rm ascending shortcuts}){$M[x][y] \ge 0 \wedge M_2[y][z]\ge 0$}{
        $M_3[x][z] \gets \max\{M_3[x][z],M[x][y]+M_2[y][z]\}$
        }\Else(\tcp*[h]{\rm descending shortcuts with charge drop}){
            \If{$\min\{M[x][y],0\}+\min\{M_2[y][z],0\} \ge -B$}{
            $M_2[x][z] \gets \max\{M_2[x][z],\min\{M[x][y],0\}+\min\{M_2[y][z],0\}\}$
            }
        }
    }
    \For{$x,z\in V$}{
        $D[x][z]\gets \max\{D[x][z], M_2[x][z], M_3[x][z] \}$
    }
  }
\end{algorithm}

\begin{center}
    \includegraphics[width=1\textwidth]{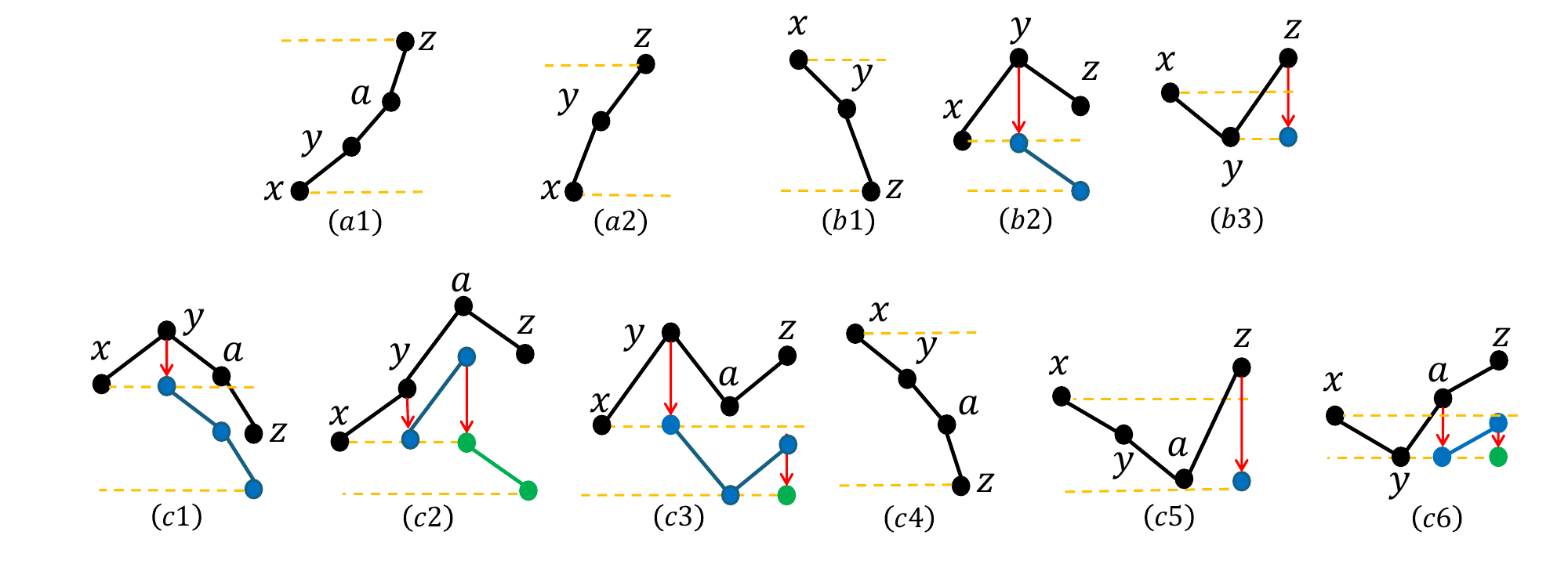}
\end{center}

\caption{
All cases of easy shortcuts, the dashed yellow lines represent the maximum and minimum gains in the paths (with respect to the charge drop schedule). The depicted charge drop schedules are optimal, i.e., the paths are descending with largest possible gain. Figures $(a1),(a2)$ are the only ascending shortcuts. Figures $(b1)$-$(b3)$ are descending $2$-shortcuts with respect to a charge drop schedule that cancels every positive gain arc. Figures $(c1)$-$(c6)$ are descending $3$-shortcuts. Note that the suffix $yaz$ of these paths (except for $(c3)$) has the same schedule as in $(b1)$-$(b3)$. Case $(c3)$ is special since if $z$ was higher, then $xyaz$ was ascending. This is handled in $\Simple(M)$. 
}\label{fig:trivial_shortcuts}
\end{figure}

Assume $M[x][y]>0$, we try to find shortcuts corresponding to Case $1$. That is, we find $a\in V$ such that $xyaz$ is ascending and $a$ satisfies $M[y][a]\le 0$ and $M[a][z]\ge 0$, see Figure~\ref{figure:simple-shortcuts}$(a)$. The computation of $a\in V$ is done as follows. Among all nonpositive gain arcs $ya$ whose gain in absolute value is smaller than the gain of $xy$, we want to find the one which maximized
$M[y][a]+M[a][z]$ and $M[y][a]+M[a][z] \ge 0$. To make this search efficient we store all the pairs $(M[y][a], M[y][a]+M[a][z])$, for $a\in V$, in a $2D$ range tree $T_{yz}$. Creating such a range tree $T_{yz}$ can be done in $O(n\log^2 n)$ time.
Finally, our update for the triplet $x,y,z\in V$ will find amongst pairs in $T_{yz}$ in which the first key $M[y][a]$ satisfies $M[y][a]\in [-M[x][y],0]$, the pair in which its second key $M[y][a]+M[a][z]$ is maximal.\footnote{It is enough to use a range tree in which the secondary structures are heaps rather than search trees because we only need to find the maximum in every secondary data structure. This saves a $\log n$ factor.} This is done in $O(\log^2 n)$ time both in the construction and initialization.

Assume $M[x][y]<0$,
finding shortcuts corresponding to Case $2$ is done similarly to shortcuts corresponding to Case $1$ by utilizing the range tree $T_{yz}$, see Figure~\ref{figure:simple-shortcuts}$(b)$. We find shortcuts corresponding to Case $3$ as follows. Among all $a\in V$ such that $M[y][a] \ge |M[x][y]|$, we find $a\in V$ such that $M[a][z]$ is maximized, see Figure~\ref{figure:simple-shortcuts}$(c)$. To make this search efficient we store all the pairs $(M[y][a], M[a][z])$, for $a\in V$, in a $2D$ range tree $T'_{yz}$, for every pair $y,z\in V$.

The pseudocode of $\Simple(M)$ is given in Figure~\ref{figure:simple-shortcuts}. This pseudocode, and the pseudocodes of the algorithms in the next sections, use range trees as follows. Let $K_1$ and $K_2$ be arrays of length $n$. We denote by $RT(K_1[\cdot],K_2[\cdot])$ the operation of creating a $2D$ range tree with key pairs $(K_1[i],K_2[i])$, for $1\le i \le n$.

\begin{lemma}
    Let $P$ be a monotone path in $G^M$ of length $k\in\{2,3\}$ with respect to a charge drop schedule $C$. Then at the end of $\Simple(M)$, $D$ dominates $P$ with respect to $C$.
\end{lemma}

\begin{proof}
    Assume $P=xyz$ is of length $2$. Since $P$ is traversable we get that $M[x][y]+M[y][z]\ge -B$. We split into cases according to the signs of the arc gains of $P$, see Figure~\ref{fig:trivial_shortcuts} Cases $(a2)$ and Cases $(b1)-(b3)$.
    
    \textbf{Case $1$: $M[x][y],M[y][z]\ge 0$ (Figure~\ref{fig:trivial_shortcuts}$(a2)$):}  Clearly $P$ is ascending with respect to the zero schedule and clearly from the pseudocode of $\Trivial(M,D)$ (see Figure~\ref{fig:trivial_shortcuts}) $D[x][z] \ge M_2[x][z] \ge M[x][y]+M[y][z]=g(P)$.

    \textbf{Case $2$: $M[x][y],M[y][z]< 0$ (Figure~\ref{fig:trivial_shortcuts}$(b1)$):}
    In this case $P$ is descending with respect to the zero schedule. By the pseudocode of $\Trivial(M,D)$, we get that $D[x][z] \ge M_2[x][z] \ge M[x][y]+M[y][z]=g(P)$.

    \textbf{Case $3$: $M[x][y]\ge0 $ and $M[y][z] < 0$ (Figure~\ref{fig:trivial_shortcuts}$(b2)$):} Therefore $P$ must be descending with respect to $C$. Since $g^{P,C}_y \le g^{P,C}_x=0$, it follows that $C(y)\ge M[x][y]$ and therefore $g^C(P) = (M[x][y]-C(y)) + (M[y][z]-C(z)) \le M[y][z]$. Therefore, from the pseudocode of $\Trivial(M,D)$,
    \begin{align*}
        D[x][z] \ge M_2[x][z] \ge \min\{M[x][y],0\}+\min\{M[y][z],0\} = M[y][z] \ge g^C(P).
    \end{align*}

    \textbf{Case $4$: $M[x][y]< 0 $ and $M[y][z]\ge 0$ (Figure~\ref{fig:trivial_shortcuts}$(b3)$):} 
    Therefore $P$ must be descending with respect to $C$. 
    Since $g^{P,C}_y \ge g^{P,C}_z$, it follows that $C(z)\ge M[y][z]$ and therefore $g^C(P) = (M[x][y]-g(y)) + (M[y][z]-C(z)) \le M[x][y]$. Therefore, from the pseudocode of $\Trivial(M,D)$,
    \begin{align*}
        D[x][z] \ge M_2[x][z]  \ge \min\{M[x][y],0\}+\min\{M[y][z],0\} = M[x][y] \ge g^C(P).
    \end{align*}

    Assume $P=xyaz$ is of length $3$. Since $P$ is traversable, we get by Lemma~\ref{lem:alpha_general} that $M[x][y] + M[y][a] + M[a][z] \ge -B$ and $M[y][a]+M[a][z]\ge -B$.\footnote{We mention these properties since $\Trivial$ checks if the assign values are at least $-B$.} We split the rest of the proof into cases according to the signs of the arc gains of $P$, see Figure~\ref{fig:trivial_shortcuts} cases $(a1)$ and cases $(c1)-(c6)$ and Figure~\ref{figure:simple-shortcuts} cases $(a)-(c)$.

    \textbf{Case $1$: $M[x][y],M[y][a],M[a][z]\ge 0$ (Figure~\ref{fig:trivial_shortcuts}$(a1)$):}
    In this case $P$ is ascending with respect to the zero schedule and moreover $yaz$ is ascending. By Case $(a2)$ it holds that after the first for loop $M_2[y][z]\ge g(yaz)$. Therefore, after the second for loop in $\Trivial(M,D)$, we get that 
    $$D[x][z] \ge M_3[x][z] \ge M[x][y]+M_2[y][z]\ge M[x][y] + g(yaz)=g(P).$$

    \textbf{Case $2.1$: $M[x][y] \ge 0$ and $M[y][a] < 0$ and $M[a][z]\ge 0$ and $P$ is ascending (Figure~\ref{figure:simple-shortcuts}$(a)$):}
    Since $P$ is ascending with respect to $C$, it follows by Lemma~\ref{lem:ascending_no_CDS} that $P$ is ascending with respect to the zero schedule. Therefore, $|M[y][a]|\le M[x][y]$ and $M[y][a]+M[a][z]\ge 0$. Thus, the pair $(M[y][a'], M[y][a'] +M[a'][z])$ in $T_{yz}$ with largest $k_2=M[y][a']+M[a'][z]$ that satisfies $M[y][a'] \in [-M[x][y],0]$, satisfies $k_2 \ge M[y][a] +M[a][z] \ge 0$. Therefore, from the pseudocode of $\Simple(M)$,
    \begin{align*}
        D[x][z] \ge M[x][y] + k_2 \ge
        M[x][y] + M[y][a] + M[a][z] = g(P).
    \end{align*}

    \textbf{Case $2.2$: $M[x][y]\ge 0$ and $M[y][a] < 0$ and $M[a][z]\ge 0$ and $P$ is descending (Figure~\ref{fig:trivial_shortcuts}$(c3)$):}
    Since $P$ is descending with respect to $C$ and $M[x][y],M[a][z]\ge 0$ and these are the first and last arcs of $P$, it follows that $C(y)\ge M[x][y]$ and $C(z)\ge M[a][z]$.
    Therefore, from the pseudocode of $\Trivial(M,D)$,
    \begin{align*}
      D[x][z] \ge M_3[x][z] \ge \min\{M[x][y],0\} + \min\{M_2[y][z],0\} = \min\{M_2[y][z],0\}.
    \end{align*}
    To see why $g^C(P) \le \min\{M_2[y][z],0\}$, observe that $g^C(P) \le 0$ by monotonicity and that 
    \begin{align*}
      g^C(P) &\le 
      (M[x][y]-C(y)) + M[y][a] + (M[a][z]-C(z)) \\
      &\le
      M[y][a] = \min\{M[y][a],0\} +  \min\{M[a][z],0\} \le M_2[y][z],
    \end{align*}
    where the last inequality also follows from the pseudocode of $\Trivial(M,D)$.

     \textbf{Case $3$: $M[x][y],M[y][a]\ge 0$ and $M[a][z]< 0$ (Figure~\ref{fig:trivial_shortcuts}$(c2)$):} 
     Since the last arc satisfies $M[a][z]<0$, it follows that $P$ is descending with respect to $C$. Therefore, from the pseudocode of $\Trivial(M,D)$,
     \begin{align*}
      D[x][z] \ge M_3[x][z] \ge \min\{M[x][y],0\} + \min\{M_2[y][z],0\}  = \min\{M_2[y][z],0\}
    \end{align*}
    To see why $g^C(P) \le \min\{M_2[y][z],0\}$, observe that $g^C(P) \le 0$ by monotonicity. Moreover, since the first two arcs have nonnegative gain we get that $C(y) + C(a)\ge M[x][y] + M[y][a]$. Therefore, from the pseudocode of $\Trivial(M,D)$, 
    \begin{align*}
      g^C(P) &\le 
      (M[x][y]-C(y)) + (M[y][a]-C(a)) + M[a][z] \\
      &\le
      M[a][z] = \min\{M[y][a],0\} +  \min\{M[a][z],0\} \le M_2[y][z],
    \end{align*}
    where the last inequality also follows from the pseudocode of $\Trivial(M,D)$.

    \textbf{Case $4$: $M[x][y]\ge 0$ and $M[y][a],M[a][z]<0$ (Figure~\ref{fig:trivial_shortcuts}$(c1)$):} 
    Since the last arc satisfies $M[a][z]<0$, it follows that $P$ is descending with respect to $C$. Therefore, from the pseudocode of $\Trivial(M,D)$,
     \begin{align*}
      D[x][z] \ge M_3[x][z] \ge \min\{M[x][y],0\} + \min\{M_2[y][z],0\}  = \min\{M_2[y][z],0\}.
    \end{align*}
    To see why $g^C(P) \le \min\{M_2[y][z],0\}$, observe that $g^C(P) \le 0$ by monotonicity. Since the first arc $xy$ has nonnegative gain, $C(y)\ge M[x][y]$. Therefore, 
    \begin{align*}
      g^C(P) &\le 
      (M[x][y]-C(y)) + M[y][a] + M[a][z] \\
      &\le
      M[y][a] + M[a][z] = \min\{M[y][a],0\} +  \min\{M[a][z],0\} \le M_2[y][z],
    \end{align*}
    where the last inequality follows from the pseudocode of $\Trivial(M,D)$.
    
    \textbf{Case $5$: $M[x][y], M[y][a],M[a][z] < 0$ (Figure~\ref{fig:trivial_shortcuts}$(c4)$):} 
    Since the first arc $xy$ has negative gain, it holds that $P$ is descending with respect to $C$. It follows from the pseudocode of $\Trivial(M,D)$ that $M_2[y][z] \ge \min\{M[y][a],0\} +  \min\{M[a][z],0\}  = M[y][a]+M[a][z]$. Therefore, 
    \begin{align*}
        D[x][z] \ge M_3[x][z] \numge{1} \min\{M[x][y],0\} +  \min\{M_2[y][z],0\} \ge M[x][y]+M[y][a]+M[a][z] = g(P) \le g^C(P),
    \end{align*}
    where Inequality $(1)$ follows from the pseudocode of $\Trivial(M,D)$.

    \textbf{Case $6$: $M[x][y], M[y][a] < 0$ and $M[a][z] \ge 0$ (Figure~\ref{fig:trivial_shortcuts}$(c5)$):} 
    Since the first arc $xy$ has negative gain, it holds that $P$ is descending with respect to $C$. Since the last arc $az$ is of positive gain, we get that $C(z) \ge M[a][z]$. It follows from the pseudocode of $\Trivial(M,D)$ that $M_2[y][z] \ge \min \{M[y][a],0\}+ \min \{M[a][z],0\} = M[y][a]$. Therefore,
    \begin{align*}
         D[x][z] &\ge M_3[x][z] \numge{1} \min\{M[x][y],0\} +  \min\{M_2[y][z],0\} \ge M[x][y] +   M[y][a] \\
         &\ge M[x][y] + M[y][a] + (M[a][z] - C(z)) \ge
         g^C(P).
    \end{align*}
    where Inequality $(1)$ follows from the pseudocode of $\Trivial(M,D)$.
    
    \textbf{Case $7$: $M[x][y] < 0$ and $M[y][a],M[a][z] \ge 0$ (Figure~\ref{fig:trivial_shortcuts}$(c6)$):} 
    Since the first arc $xy$ has negative gain, it holds that $P$ is descending with respect to $C$. It follows from the pseudocode of $\Trivial(M,D)$ that $M_2[y][z] \ge M[y][a]+M[a][z] \ge 0$. Therefore,
    \begin{align*}
         D[x][z] \ge M_3[x][z] \numge{1} \min\{M[x][y],0\} +  \min\{M_2[y][z],0\} = 0 \ge g^C(P),
    \end{align*}
    where Inequality $(1)$ follows from the pseudocode of $\Trivial(M,D)$.

    \textbf{Case $8$: $M[x][y] < 0 $ and $M[y][a] \ge 0$ and $M[a][z] < 0$ (Figure~\ref{figure:simple-shortcuts}$(b)-(c)$):} 
     We split into sub-cases

    \textbf{Sub-Case $8.1$: $M[y][a]\in [0, |M[x][y]|]$:} Therefore, the pair $(k_1,k_2) = (M[y][a'], M[y][a'] +M[a'][z])$ in $T_{yz}$ with largest $k_2 = M[y][a']+M[a'][z]$ that satisfies $k_1=M[y][a'] \in [0,M[x][y]|]$, satisfies $k_2 \ge M[y][a] +M[a][z]$. Thus, by the two inner-if statements in $\Simple(M)$, we get that
    \begin{align*}
        D[x][z] &\ge \min \{ M[x][y], M[x][y]+k_2 \}\\
        &\ge
        \min \{ M[x][y], M[x][y]+M[y][a]+M[a][z] \}\\
        &=
        \min \{ M[x][y], g(P) \} \ge g^C(P),
    \end{align*}
    where the last inequality follows since $P$ is descending with respect to $C$ so $M[x][y] \ge g^{P,C}_y \ge g^C(P)$.

    \textbf{Sub-Case $8.2$: $M[y][a]\ge |M[x][y]|$:} Therefore, the pair $(k_1,k_2) = (M[y][a'], M[a'][z])$ in $T'_{yz}$ with largest $k_2=M[a'][z]$ that satisfies $k_1=M[y][a'] \ge |M[x][y]|$, satisfies $k_2 \ge M[a][z]$. Thus, by last assignment to $D$ in $\Simple(M)$, we get that
    \begin{align}\label{eq:simple_case_c1}
        D[x][z] &\ge \min \{ M[x][y], k_2 \}
        \ge
        \min \{ M[x][y], M[a][z] \}.
    \end{align}
    Since $P$ is descending with respect to $C$, it holds that $$(M[x][y]-C(y)) + (M[y][a]-C(a)) = g^{P,C}_a \le g^{P,C}_x = 0$$
    and therefore 
    \begin{align}\label{eq:simple_case_c2}
      g^C(P) = (M[x][y]-C(y)) + (M[y][a]-C(a)) + (M[a][z]-C(Z))  \le M[a][z]-C(Z) \le M[a][z].  
    \end{align}
    
    Similarly, since $P$ is descending with respect to $C$, we get that
    \begin{align}\label{eq:simple_case_c3}
       g^C(P) \le g^{P,C}_y =  (M[x][y]-C(y))  \le M[x][y]. 
    \end{align}
    By combining Equations~(\ref{eq:simple_case_c1}),(\ref{eq:simple_case_c2}),(\ref{eq:simple_case_c3}), we get that
    \begin{align*}
        D[x][z] \ge
        \min \{ M[x][y], M[a][z] \} \ge 
        g^C(P).
    \end{align*}
\end{proof}

\begin{lemma}\label{lemma:trivial-shortcuts-invariant}
    Procedure $\Trivial(M,D)$ maintains Invariant~\ref{invariant}\ref{1c}.
\end{lemma}

\begin{proof}
    We prove that every time Algorithm $\Trivial(M,D)$ makes an assignment to $D[x][z]$ then there is a path $P$ from $x$ to $z$ and a charge drop schedule $C$ such that $P$ is monotone with respect to $C$ and $g^C(P) = D[x][z]$. We split the proof into cases:

    \textbf{Case $1$: $D[x][z] = M_2[x][z] = M[x][y]+M[y][z]$:}
    The algorithm performs this assignment when $M[x][y],M[y][z]\ge 0$. In particular $xyz$ is ascending with respect to the zero schedule and $g(xyz)=D[x][z]$. Moreover, by Lemma~\ref{lemma:alpha-of-monotone}, $P$ is strongly traversable.

    \textbf{Case $2$: $D[x][z] = M_2[x][z]  = \min\{M[x][y],0\}+\min\{M[y][z],0\} \ge -B$:}
    The algorithm performs this assignment when either $M[x][y]<0$ or $M[y][z] < 0$. Let $P = xyz$. We apply the following charge drop schedule $C$: If $M[x][y] \ge 0$, then $C(y)=M[x][y]$ and if $M[y][z] \ge 0$, then $C(z)=M[y][z]$. It is easy to see that $P$ is descending with respect to $C$ and that $g^C(P) = \min\{M[x][y],0\}+\min\{M[y][z],0\}$.

    \textbf{Case $3$: $D[x][z] = M_3[x][z]  = M[x][y]+M_2[y][z]$:}
    The algorithm performs this assignment when $M[x][y],M_2[y][z]\ge 0$. By Case $1$ above, $M_2[y][z]\ge 0$ implies that there is an ascending path $yaz$, with respect to the zero schedule, of in $G^M$. Therefore, $xyaz$ is ascending with respect to the zero schedule. 

    \textbf{Case $4$: $D[x][z] = M_3[x][z]  = \min\{M[x][y],0\}+\min\{M_2[y][z],0\}\ge -B$:}
    The algorithm performs this assignment when either $M[x][y]<0$ or $M_2[y][z] < 0$.
    We split into sub-cases:
    
    \textbf{Case $4.1$: $M_2[y][z] \ge 0$:}
    This means that $M[x][y]<0$ and therefore $D[x][z] = M[x][y]$. Since $M_2[y][z] \ge 0$, it follows by the pseudocode of $\Trivial(M,D)$ that $M_2[y][z] = M[y][a] + M[a][z]$ where $a\in V $ and $M[y][a],M[a][z]\ge 0$. Let $P = xyaz$ and consider the charge drop schedule $C$ where $C(a) = M[y][a]$ and $C(z) = M[a][z]$. It is easy to see that $P$ is descending with respect to $C$ ($P$ is traversable since $g^C(P) = M[x][y] = D[x][z]\ge -B$).
    
    \textbf{Case $4.2$: $M_2[y][z] < 0$:}
    Therefore $M_2[y][z] = \min \{M[y][a],0 \} + \min \{M[a][z],0 \}$ for some $a\in V$. Moreover $D[x][z] = \min \{M[x][y],0 \} + M_2[y][a]$. Let $P = xyaz$ and consider the charge drop schedule $C$ where $C(y) =  \max \{M[x][y],0 \}$ and $C(a) = \max \{M[y][a],0 \}$ and $C(z) =\max \{M[a][z],0 \}$. Observe that
    \begin{align*}
        g^C(P) &= (M[x][y] - \max \{M[x][y],0 \}) + 
        (M[y][a] - \max \{M[y][a],0 \}) + 
        (M[a][z] - \max \{M[a][z],0 \}) 
        \\ &= \min \{M[x][y],0 \} + \min \{M[y][a],0 \} + \min \{M[y][z],0 \} \\
        &= D[x][z] \ge -B.
    \end{align*}
    Thus, $P$ is traversable.
\end{proof}

\begin{lemma}\label{lemma:simple-shortcuts-invariant}
    Procedure $\Simple(M)$ maintains Invariant~\ref{invariant}\ref{1c}.
\end{lemma}

\begin{proof}
We prove that every time $\Simple(M)$
makes an assignment to $D[x][z]$ then there is a monotone path $P$ with respect to a charge drop schedule $C$ from $x$ to $z$ such that $g^C(P) = D[x][z]$.

By Lemma~\ref{lemma:trivial-shortcuts-invariant} it is enough to consider only assignments made after executing $\Trivial(M,D)$ in $\Simple(M)$. These assignments correspond to monotone paths of length $3$ that contain arcs of positive and negative gain, see Figure~\ref{figure:simple-shortcuts}. Consider such an assignment associated with triplet
 $x,y,z\in V$.
 
Assume $M[x][y]<0$, we split into cases according to the assignment of the algorithm in the pseudocode. 

\textbf{Case $1$: $D[x][z] = M[x][y] +k_2$:} That is, the algorithm assigned $D[x][z] = M[x][y] + M[y][a] + M[a][z]$, where $a\in V$ satisfies $M[y][a] \in [0, |M[x][y]|]$ and $-B \le M[x][y] + M[y][a] + M[y][z] \le M[x][y]$, See Figure~\ref{figure:simple-shortcuts}$(b)$. Consider the path $P=xyaz$, clearly $g(P)=D[x][z]$. It holds that 
\begin{align*}
    &g_x=0,\;\; g_y=M[x][y] < 0 =g_x,\;\; g_a = M[x][y]+M[y][a] \le 0 = g_x,\\
    &g_z = M[x][y]+M[y][a]+M[a][z] \le M[x][y] = g_y \le g_x.
\end{align*}

Thus, $x$ has the Largest gain in $P$. Moreover, $g_a=M[x][y]+M[y][a] \ge M[x][y]=g_y\ge g_z$, so $z$ has the minimum gain in $P$. It is easy to see that $P$ is also traversable, hence, $P$ is descending.

\textbf{Case $2$: $D[x][z] = M[x][y]$:}
That is, there is $a\in V$ such that $M[y][a] \in [0, |M[x][y]|]$ and $M[x][y] + M[y][a] + M[a][z] \ge M[x][y]$.
In particular $M[y][a] + M[a][z] \ge 0$.
Let $P = xyaz$ (observe that $P$ is traversable) and consider the charge drop schedule $C$ that only drops charge at $z$ and $C(z) = M[y][a] + M[a][z]$. We prove that $P$ is descending with respect to $C$. Observe that

\begin{align*}
    &g^C_x=0,\;\; g^C_y=M[x][y] < 0 =g^C_x,\;\; g^C_a = M[x][y]+M[y][a] \le 0 = g^C_x,\\
    &g^C_z = M[x][y]+M[y][a]+(M[a][z] - C(z)) = M[x][y] = g^C_y \le g^C_x.
\end{align*}

Thus, $x$ has the maximum gain in $P$ with respect to $C$. Moreover, $g^C_a=M[x][y]+M[y][a] \ge M[x][y]=  g_z$, so $z$ has the minimum gain in $P$ with respect to $C$. Hence, $P$ is descending with respect to $C$.

\textbf{Case $3$: $D[x][z] = \min \{M[x][y], M[a][z] \}$:} That is $a\in V$ satisfies $M[x][y] \ge |M[x][y]|$, see Figure~\ref{figure:simple-shortcuts}$(c)$. Let $P = xyaz$ (observe that $P$ is traversable) and let $C$ be the schedule that assigns $C(a) = M[x][y] + M[y][a]$ and $C(z) = \max \{0, M[a][z] - M[x][y] \}$.
Observe that

\begin{align*}
    &g^C_x=0,\;\; g^C_y=M[x][y] < 0 =g^C_x,\;\; g^C_a = M[x][y]+(M[y][a]-C(a)) = 0 = g^C_x,\\
    &g^C_z = M[x][y]+(M[y][a]-C(a))+(M[a][z] - C(z)) = \min \{M[a][z], M[x][y] \} \le g^C_x.
\end{align*}

Thus, $x$ has the maximum gain in $P$ with respect to $C$. Moreover, 
\begin{align*}
g^C_a&= 0 \ge \min \{M[a][z], M[x][y] \} =  g^C_z,\\
g^C_y &= M[x][y] \ge \min \{M[a][z], M[x][y] \} = g^C_z,
\end{align*}
 so $z$ has the minimum gain in $P$ with respect to $C$. Hence, $P$ is descending with respect to $C$.

\medskip

We now assume that $M[x][y]\ge 0$, and the algorithm assigned $D[x][z]=M[x][y]+M[y][a]+M[a][z]$ where $M[y][a]+M[a][z] \ge 0$ for some $a\in V$ satisfying $M[y][a] \in [-M[x][y] , 0]$. Consider the path $P=xyaz$. Observe that $P$ is traversable. Similar to before, we get that $P$ is ascending. By Lemma~\ref{lemma:alpha-of-monotone} we conclude that $P$ is strongly traversable.
\end{proof}

\subsection{Building Long Shortcuts}
The procedure $\LS(M)$ aims to find \emph{long shortcuts} in $G^M$ and update $M$ accordingly. \emph{Long Shortcuts} are shortcuts that correspond to monotone paths of length $k>3$. We find such shortcuts by computing arc-bounded paths and then extending them by one arc into monotone paths (i.e  shortcuts). We give the full description of $\LS(M)$ in Appendix~\ref{section:long}. This algorithm uses several sub-algorithm which we list below and elaborate on in the next sections. 

\begin{itemize}
    \item  $\BFS(M,D):$ This procedure aims to discover arc-bounded paths that are longer than the ones stored in $D$. This is done by extending existing arc-bounded paths in $D$ by one arc. This procedure performs updates of the form $D[xy][z] = \max \{D[xy][z], D[xy][a] + M[a][z]\}$. See Figure~\ref{figure:BFS}. 
    \item $\Concat(M,D,U,W,X)$: Given sets $U,W,X\subseteq V$, the procedure aims to discover longer arc-bounded paths than the ones stored in $D$ by concatenating first-arc-bounded paths with first-arc-bounded paths and last-arc-bounded paths with last-arc-bounded paths. This procedure performs updates of the form $D[uv][x] = \max \{D[uv][x], D[uv][w]+D[wa][x]\}$, where $u\in U, w\in W, x\in X$. See Figure~\ref{figure:concat}.
    \item $\ComputeF(M):$ This procedure returns a data structure $D$ that dominates any simple path that is a funnel in $G^M$ w.h.p.\ (see Lemma~\ref{lemma:funnels-optimality}). See Figure~\ref{figure:funnels}.
    \item $\CO(M,D,U,W,X)$: Given sets $U,W,X\subseteq V$, this procedure aims to discover longer arc-bounded paths than the ones stored in $D$ by concatenating first-arc-bounded paths with last-arc-bounded paths. This procedure performs updates of the form $D[uv][x] =\max \{D[uv][x], D[uv][w]+D[w][ax]\}$, where $u\in U, w\in W, x\in X$. See Figure~\ref{figure:concat-opposite}.
    \item $\Extend(M,D,T)$: Given a set $T\subseteq V$, this procedure considers  every arc-bounded path in which the ``bounding" arc contains a vertex of $T$. The goal of this procedure is to extend such a path by a single arc and get a monotone path. This is the procedure that computes the shortcuts for $\LS(M,D)$.
\end{itemize}

The following is the relation between the different algorithms.
Algorithm $\ComputeF(M)$ Is achieved by applying $\BFS{}$ and $\Concat$ several times on a sampled set. 
Algorithms $\LS(M)$ (see Figure~\ref{figure:long-shortcuts}) starts by applying $\ComputeF(M)$, which returns a data structure $D$ that,  dominates every simple path that is a funnel in $G^M$ w.h.p.. The algorithm then tries to elongate some sampled arc-bounded paths. This is done by consecutive applications of $\Concat$ and $\CO$. Finally, $\LS(M)$ calls $\Extend$ in order to transform the arc bounded path stored in $D$ into monotone paths.

\subsubsection{Breadth-Search}\label{section:BFS}
This procedure extend the length of arc-bounded paths dominated by $D$, by concatenating to them a single arc of larger gain. I.e., given $P = v_1 \ldots v_k$, a $v_1 v_2$-bounded path, $\BFS(M,D)$ scans all arcs $xv_1$ and checks if $xv_1 \ldots v_k$ is $x v_1$-bounded path and if so, updates $D[xv_1][v_k]$. The implementation is as follows and its pseudocode is given in Figure~\ref{figure:BFS}. 

For every triplet $x,y,z\in V$, we update $D[xy][z]$ as follows. If $M[x][y]\ge 0$, we consider the values $D[ya][z]$, for all $a\in V$ such that $-M[x][y] \le  M[y][a] \le 0$, and we concatenate $xy$ to the path corresponding to $D[ya][z]$, which results in a $xy$-bounded path to $z$.  That is, for every $y,z\in V$, we find $a \in V$,  that maximizes $M[x][y]+D[ya][z]$ while satisfying  $-M[x][y] \le  M[y][a] \le 0$. To compute such $a\in V$, we store in a range tree $FT_{yz}$ the pairs $(k_1,k_2)=(M[y][a],D[ya][z])$ for every $a\in V$.
To update $D[xy][z]$, we search in $FT_{yz}$ for the pair $(k_1,k_2)=(M[y][a],D[ya][z])$ with largest $k_2$ that satisfies $k_1 \in [-M[x][y],0]$. We then assign $D[xy][z] = \max \{D[xy][z], M[x][y] + D[ya][z] \}$.

The case $M[x][y]\le 0$ and the cases that $P$ is last-arc-bounded are symmetric, See Figure~\ref{figure:BFS}.

\begin{figure}[t!]
\begin{algorithm}[H]
  \Fn{$\BFS(D,M)$}{
      \BlankLine
      \For{$a,b\in V$}{
        $FT_{ab} \gets RT(M[a][\cdot], D[a\cdot][b])$ \tcp*{\rm Range tree of $\ubar{a}\bar{w}b$ and $\bar{a}\ubar{w}b$ paths}
        $LT_{ab} \gets RT(M[\cdot][b], D[a][\cdot b])$ \tcp*{\rm Range tree of $a\ubar{w}\bar{b}$ and $a\bar{w}\ubar{b}$  paths}
      }
      \For{$x,y,z\in V$}{
        \If{$M[x][y]\ge0$}{
            $(-,D[ya][z])\gets FT_{yz}.range(k_1\in [-M[x][y],0]).max\_k_2()$\;
            $D[xy][z] \gets \max \{D[xy][z], M[x][y] + D[ya][z]\}$ \tcp*{\rm We do this if $D[ya][z] \neq -\infty$}
        }\Else{
            $(-,D[ya][z])\gets FT_{yz}.range(k_1 \in [0,|M[x][y]|]).max\_k_2()$\;
            $D[xy][z] \gets \max \{D[xy][z], M[x][y] + D[ya][z]\}$ \tcp*{\rm We do this if $D[ya][z] \neq -\infty$}
        }
        \If{$M[y][x]\ge0$}{
            $(-,D[z][ay])\gets LT_{zy}.range(k_1\in [-M[y][x],0]).max\_k_2()$\;
            $D[z][yx] \gets \max \{D[z][yx], D[z][ay] + M[y][x]\}$ \tcp*{\rm We do this if $D[z][ay] \neq -\infty$}
        }\Else{
            $(-,D[z][ay])\gets LT_{zy}.range(k_1\in [0,|M[y][x]|]).max\_k_2()$\;
            $D[z][yx] \gets \max \{D[z][yx], D[z][ay] + M[y][x]\}$ \tcp*{\rm We do this if $D[z][ay] \neq -\infty$}
        }
      }    
  }
\end{algorithm}

\begin{center}
    \includegraphics[width=0.7\textwidth]{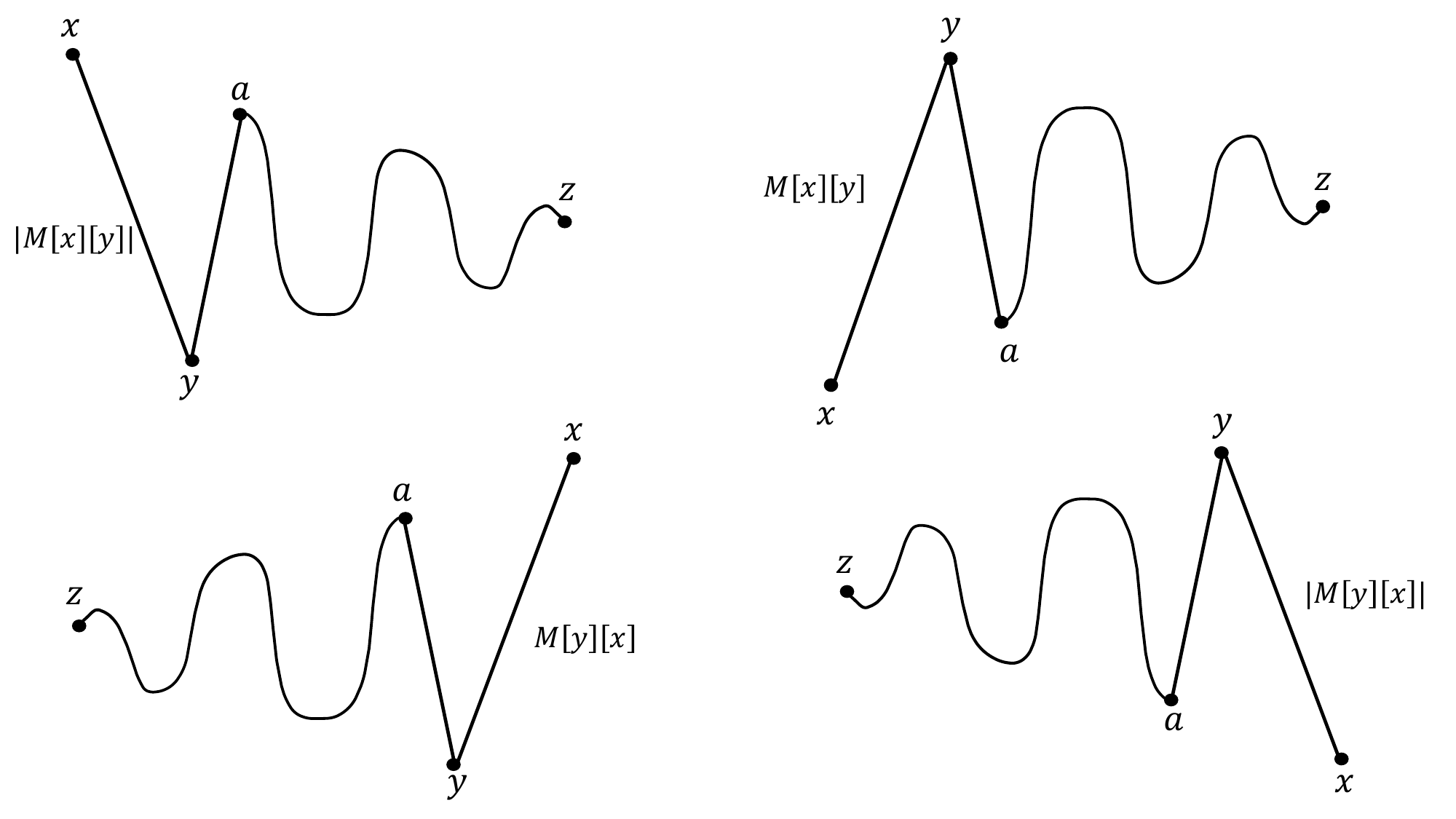}
    \label{figure:BFS}
\end{center}

\caption{The four cases of $\BFS(M,D)$. On the top  we concatenate the arc $xy$ with a first-arc bounded path from $y$ to $z$. On the bottom  we concatenate a last-arc bounded path from $z$ to $y$ with the arc $yx$.}\label{alg2}
\end{figure}

\begin{lemma}\label{lemma:BFS-single-extension}
Let $P = v_1 \ldots v_k$ be an arc-bounded path in $G^M$. If $D$ dominates $P$, then the following holds after $\BFS(D,M)$
\begin{itemize}
    \item If $P$ is $v_1 v_2$-bounded and $P'=v_0 v_1 v_2 \ldots v_k$ is $v_0 v_1$-bounded, then $D$ dominates $P'$.
    \item If $P$ is $v_{k-1} v_k$-bounded and $P'=v_1 \ldots v_k v_{k+1}$ is $v_k v_{k+1}$-bounded, then $D$ dominates $P'$.
\end{itemize}
\end{lemma}

\begin{proof}
    Assume the first case, i.e., $P$ is $v_1 v_2$-bounded. Assume that $M[v_1][v_2] \le 0$, the case $M[v_1][v_2]\ge 0$ is symmetric.
    Let $(M[v_1][a] , D[v_1 a][v_k])$ be the pair in $FT_{v_1 v_k}$ with largest $D[v_1 a][v_k]$ that satisfies $M[v_1][a] \in [- M[v_0][v_1],0]$. Since $P'$ is $v_0 v_1$-bounded, we have $|M[v_1][v_2]| \le M[v_0][v_1]$ and therefore $D[v_1 a][v_k] \ge D[v_1 v_2][v_k]$. Thus, after $\BFS(M,D)$, 
    \begin{align*}
        D[v_0 v_1][v_k] \ge M[v_0][v_1] + D[v_1 a][v_k] \ge M[v_0][v_1] +  D[v_1 v_2][v_k] \ge M[v_0][v_1] + g(P) = g(P').
    \end{align*}

    The proof of the second case where $P$ is $v_{k-1}v_k$-bounded is symmetric.
\end{proof}

Since every funnel is arc-bounded, the following is a direct corollary of Lemma~\ref{lemma:BFS-single-extension}.

\begin{corollary}\label{lemma:BFS-basic}
    Assume that every funnel $P$ in $G^M$ of length at most $k$ is dominated by $D$.
    Then after calling $\BFS(D,M)$, it holds that every funnel $P$ in $G^M$ of length at most $k+1$ is dominated by $D$.
\end{corollary}

\begin{lemma}
    Procedure $\BFS(M,D)$ maintains Invariants~\ref{invariant}\ref{1a} and \ref{invariant}\ref{1b}
\end{lemma}

\begin{proof}
    Assume the invariant holds before $\BFS(M,D)$. We proceed by induction on the changes of $D$. Let $x,y,z\in V$. We split into cases.

    Assume $M[x][y]\ge 0$ and assume the procedure assigned $D[xy][z]=M[x][y]+D[ya][z]$, where $a\in V$ satisfies $0 > M[y][a] \ge  -M[x][y]$. By Invariant~\ref{invariant}\ref{1a}, there is a traversable path $P = ya v_1\ldots v_k z$ in $G^M$ and a charge drop schedule $C$ such that $P$ is $ya$-bounded with respect to $C$ and $g^C(P)= D[ya][z]$. Since $M[x][y]\ge 0$, it follows that $P' = xya v_1\ldots v_k z$ is traversable.  Moreover, since $M[x][y]\ge |M[y][a]|$ and $P$ is $ya$-bounded with respect to $C$, we get that $P'$ is $xy$-bounded with respect to the schedule $C'$ that does not drop charge at $x$ and then goes according to $C$. We get $g^{C'}(P') = M[x][y]+g^C(P) = M[x][y]+D[ya][z] = D[xy][z]$. 

    Assume $M[y][z]\le 0$ and assume the procedure assigned $D[x][yz]=D[x][ay]+M[y][z]$, where $a\in V$ satisfies $0 \le M[a][y] \le  -M[y][z]$. By Invariant~\ref{invariant}\ref{1b}, there is a path $P = x v_1\ldots v_k a y$ in $G^M$ and a charge drop schedule $C$ such that $P$ is $ay$-bounded with respect to $C$ and satisfies $g^C(P) = D[x][ay]$. Since $M[a][y] \le  -M[y][z]$, we get that $P'=x v_1\ldots v_k ayz$ is $yz$-bounded with respect to the charge drop schedule $C'$ that performs charge drops according to $C$ and does not drop charge at the new vertex $z$. By Invariant~\ref{invariant}\ref{1c}, the arc $yz$ is traversable, thus $M[y][z]\ge -B$ which means by Lemma~\ref{lemma:alpha-of-arc-bounded} that $P'$ is traversable. Finally, note that $g^{C'}(P')=g^{C}(P)+M[y][z] = D[x][ay] + M[y][z]=D[x][yz]$.
    
    The other case $M[x][y]\le 0$ is symmetric to the case $M[x][y]\ge 0$ and the case $M[y][z]\le 0$ is symmetric to the case $M[y][z]\ge 0$.
\end{proof}

\subsubsection{Concatenate first-arc bounded paths with first-arc bounded paths}\label{section:concatenate}
In this procedure (see Figure~\ref{figure:concat}) we are given $3$ sets $U,W,X\subseteq V$. For every $u\in U,w\in W,x\in X$ and $v\in V$, we try to concatenate a $\overunderline{uvw}{1-1}{2-2}$ path with some $\overunderline{wax}{1-1}{2-2}$ path, where $a \in V$.
This gives  a (hopefully new or improved gain) $\bar{u}\ubar{v}x$ path. 
We also do the symmetric version: we try to concatenate a $x\bar{a}\ubar{w}$ path to a $w\bar{v}\ubar{u}$ path. 

The choice of focusing on  paths  bounded by a arcs of negative gain  was intentional. To emphasize the difficulty in concatenating   paths  bounded by arcs of positive gain, consider the following example.

Let $P$ be a $\overunderline{uvw}{2-2}{1-1}$ path, where $M[u][v] = 10$ and $g(P)=5$. Let $Q$ be a $\overunderline{wax}{2-2}{1-1}$ path, where $M[w][a] = 5$ and $g(Q)=3$. Clearly $P\mid Q$ is $uv$-bounded with gain $g(P\mid Q)=8$.
However it may be the case where $D$
 dominates both $P$ and $Q$ and stores the values $D[uv][w]=9$ and $D[wa][x]=4$.
 But the concatenation of the paths, say $P'$ and $Q'$, realizing these values is not $uv$-bounded since the gain of $P'\mid Q'$ is
  $D[uv][w] + D[wa][x]=13$ which is larger than $M[u][v]$. For arcs of negative gain  if we replace $P$ by a $\overunderline{uvw}{2-2}{1-1}$ path $P'$ with a larger gain then $P'\mid Q$ is always also 
  $uv$-bounded. We could have addressed this problem by dropping charge at $w$ (see Definition~\ref{def:gain}) but we preferred to  get our desired set of shortcuts without concatenating such paths at all.

In Appendix~\ref{section:CO} we show how to concatenate  $\overunderline{uvw}{1-1}{2-2}$ paths with  $\overunderline{wax}{3-3}{2-2}$ paths. 
This requires a range tree and the ability to drop charges.

We distinguish the cases of concatenating first-arc-bounded paths and last-arc-bounded paths.

\textbf{Concatenating $\overunderline{uvw}{2-2}{1-1}$ and $\overunderline{wax}{2-2}{1-1}$:}
Consider the values D[uv][w] and $D[wa][x]$, where $a\in V$ satisfies $ M[w][a]\le 0$.
It follows from Lemma~\ref{lemma:concat-basic} and Lemma~\ref{lemma:concat-maintains-invariant} that
the concatenation of the paths realizing these values is 
a $uv$-bounded path if and only if
$|M[w][a]| \le D[uv][w] - M[u][v]$. See Figure~\ref{figure:concat}.

Therefore we update $D[uv][x]$ as follows. 
 We find an $a\in V$ that maximises $D[uv][w] + D[wa][x]$ while satisfying $|M[w][a]| \le D[uv][w] - M[u][v]$. This is done by storing, for every pair $w\in W, x\in X$, a Range tree of first-arc-bounded paths $FT_{wx}$ containing the pairs $(k_1,k_2)=(M[w][a], D[wa][x])$, for every $a\in V$. We then find the pair $(k_1,k_2)=(M[w][a], D[wa][x])$ with largest $k_2$ that satisfies $k_1 \in [-( D[uv][w] - M[u][v]),0]$. We then perform the update $D[uv][x] = \max \{ D[uv][x], D[uv][w] + D[wa][x]\}$.

\textbf{Concatenating $\overunderline{xaw}{2-2}{3-3}$ and $\overunderline{wvu}{2-2}{3-3}$:} This case is handled symmetrically. We  perform an update of the form $D[x][vu] = \max \{D[x][vu], D[x][aw] + D[w][uv] \}$, see Figure~\ref{figure:concat}.

\begin{figure}[t]
\begin{algorithm}[H]
  \Fn{$\Concat(M,D,U,W,X)$}{
  \BlankLine
  \For{$(w,x)\in  W \times  X$}{
  $FT_{wx} \gets  RT(M[w][\cdot], D[w \cdot ][x])$ \tcp*{\rm Range tree of $\bar{w}\ubar{a}x$ paths} 
  $LT_{xw} \gets  RT(M[\cdot][w], D[x][\cdot w])$ \tcp*{\rm Range tree of $x\bar{a}\ubar{w}$ paths}
  }
  \For{$(u,v,w,x)\in U\times V \times W \times  X$}{
  \If{$M[u][v]<0$}{
  $(-,D[wa][x]) \gets FT_{wx}.range(k_1 \in [-(D[uv][w]-M[u][v]),0]).max\_k_2()$\;
  $D[uv][x] \gets \max\{D[uv][x], D[uv][w] + D[wa][x]\}$
  }
  \If{$M[v][u] < 0$}{
  $(-,D[x][aw]) \gets LT_{xw}.range(k_1 \in [-(D[w][vu] - M[v][u]),0]).max\_k_2()$ \;
  $D[x][vu] \gets \max \{ D[x][vu], D[x][aw] + D[w][vu] \}$
  }
  }
  }
\end{algorithm}

\begin{center}
    \includegraphics[width=0.49\textwidth]{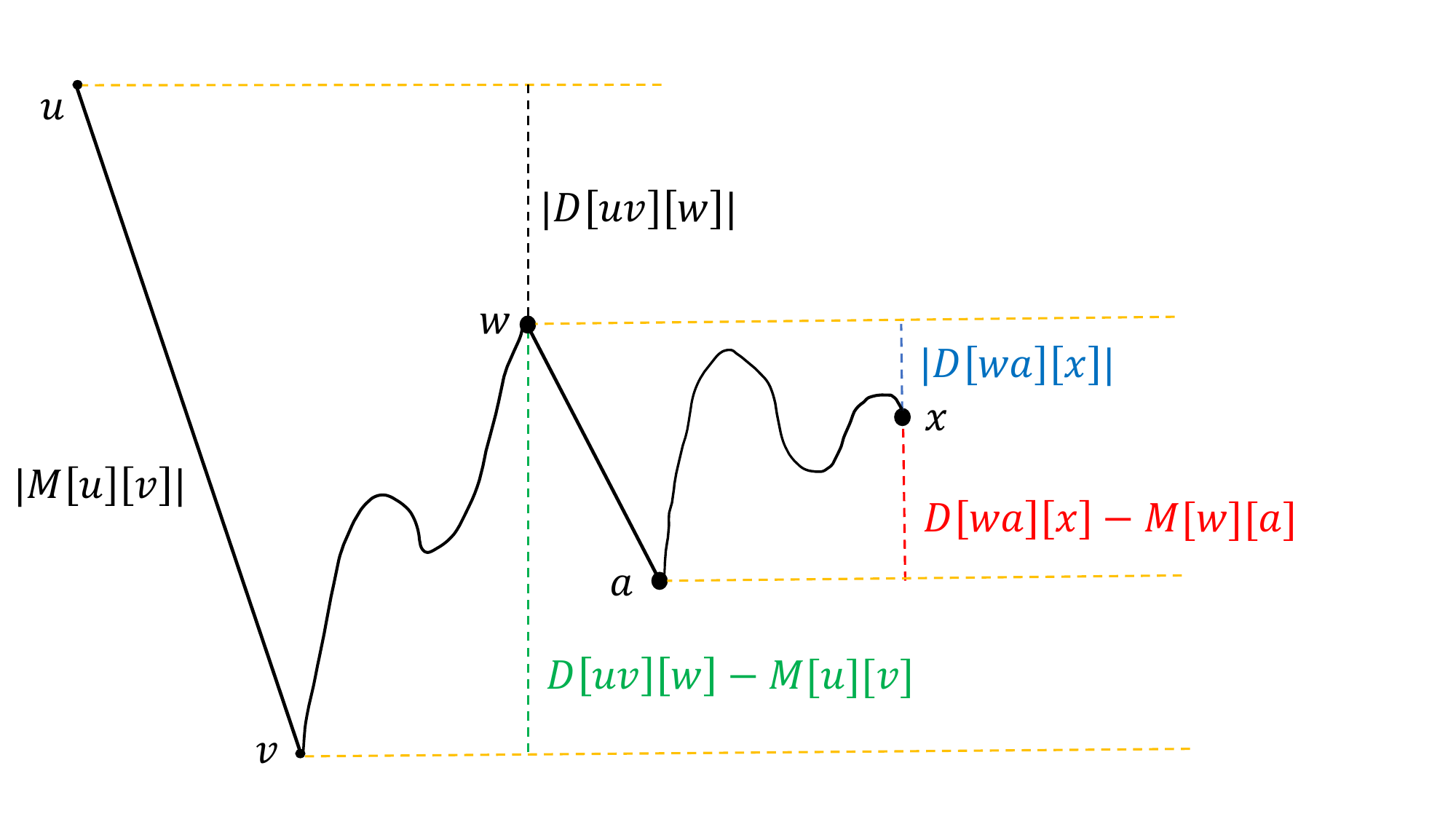}
    \includegraphics[width=0.49\textwidth]{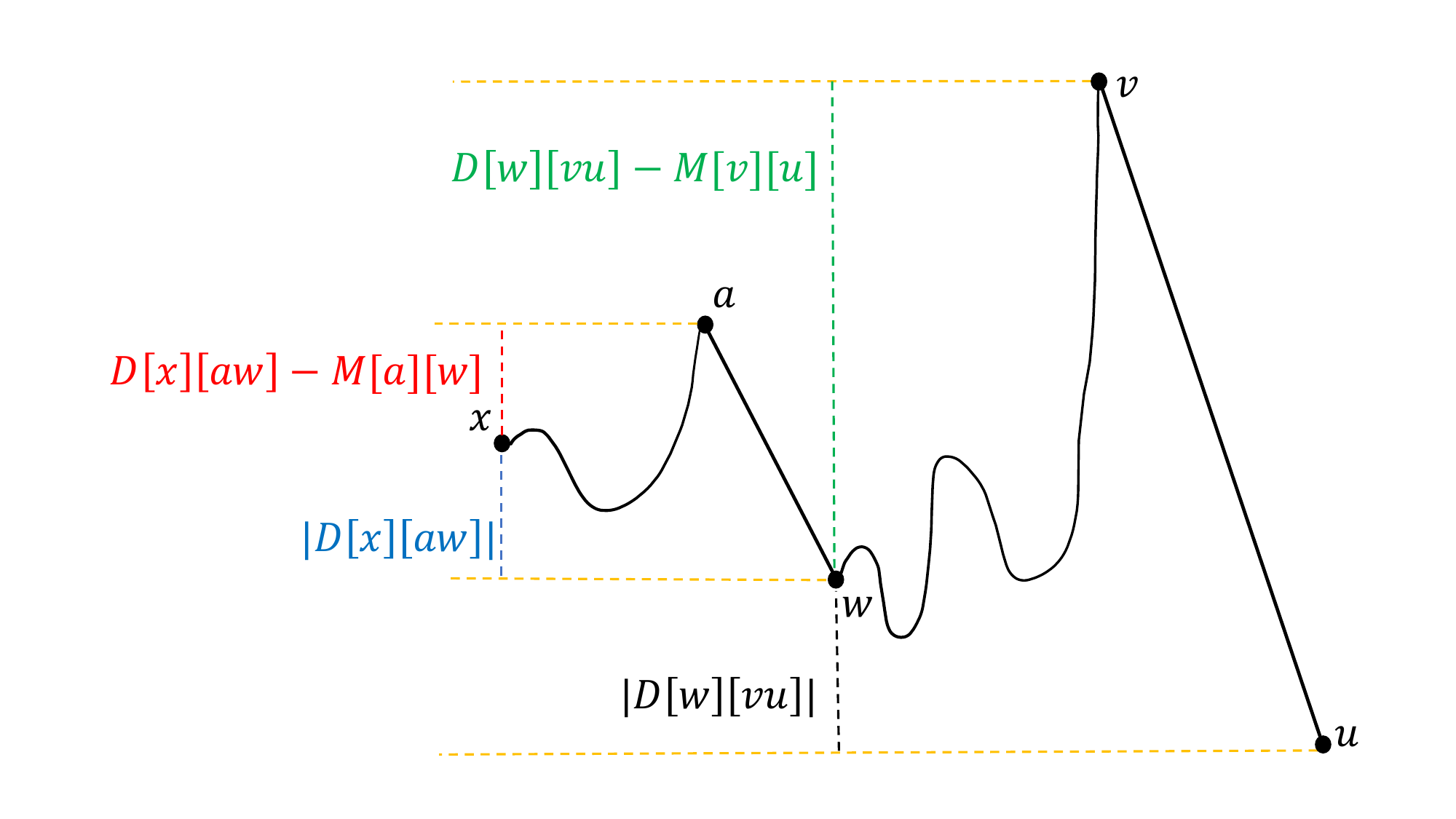}
\end{center}

\caption{On the left: a concatenation of two $\overunderline{ABC}{1-1}{2-2}$ paths.}\label{figure:concat}
\end{figure}

The following lemma proves that after running algorithm $\Concat(M,D)$,
the concatenation of two arc-bounded paths $P,Q$ that match the description above and were dominated by $D$ before
executing $\Concat(M,D)$, is dominated by $D$ after this execution.

\begin{lemma}\label{lemma:concat-basic}
    Let $U,W,X \subseteq V$ and let $u\in U,w\in W,x\in X$ and $v\in V$. Let $P_1$ and $P_2$ be paths in $G^M$ that are dominated by $D$. Assume one of the following holds
    \begin{itemize}
        \item $P_1$ is a $\overunderline{uvw}{1-1}{2-2}$ path, $P_2$ is a $\overunderline{wax}{1-1}{2-2}$ path, and $P=P_1 \mid P_2$ is a $\overunderline{uvx}{1-1}{2-2}$ path.
        \item $P_1$ is a $\overunderline{xaw}{2-2}{3-3}$ path, $P_2$ is a $\overunderline{wvu}{2-2}{3-3}$ path, and $P=P_1 \mid P_2$ is a $\overunderline{xvu}{2-2}{3-3}$ path.
    \end{itemize}
    Then, after $\Concat(M,D,U,W,X)$, $D$ dominates $P$.
\end{lemma}

\begin{proof}
    Assume the first case: Since $D$ dominates $P_1$ and $P_2$ we have $D[u v][w] + D[w a][x] \ge g(P_1) + g(P_2) = g(P)$. Since $P$ is $uv$-bounded, it follows that $g(P_1) + M[w][a] = g_a^P \ge g_{v}^P = M[u][v]$. So by rearranging we get 
    $-M[w][a]=|M[w][a]| \le g(P_1) - M[u][v]$.
    Since $D$ dominates $P_1$ it follows $
    M[w][a]\in [-(D[u v][w] - M[u][v]),0]$.
    Let $(k_1,k_2)=(M[w][a'], D[wa'][x])$ be the pair in $FT_{wx}$ with largest $k_2$ that satisfies $k_1 \in [-(D[u v][w] - M[u][v]),0]$.  Therefore $D[wa'][x] = k_2 \ge D[wa][x]$, so, after the algorithm assigns $D[uv][x]$ a value, we get $D[uv][x] \ge D[uv][w]+D[wa'][x] \ge D[uv][w]+D[wa][x] \ge g(P)$.
    
    The second case in which $P_1$ is a $\overunderline{xaw}{2-2}{3-3}$ path and $P_2$ is a $\overunderline{wvu}{2-2}{3-3}$ path is symmetric.
\end{proof}

\begin{lemma}\label{lemma:concat-maintains-invariant}
    Procedure $\Concat(M,D,U,W,X)$ maintains Invariants~\ref{invariant}\ref{1a} and \ref{invariant}\ref{1b}.
\end{lemma}

\begin{proof}
    Let $u\in U,w\in W,x\in X$ and $v\in V$.

    Assume $M[u][v]<0$ and the algorithm sets $D[uv][x] = D[uv][w] + D[wa][x]$, where $a\in V$ satisfies 
    \begin{align}\label{eq:2}
     -M[w][a]=|M[w][a]|\le D[uv][w]-M[u][v].   
    \end{align}
    By Invariant~\ref{invariant}\ref{1a}, there is a $\overunderline{uvw}{1-1}{2-2}$ path $P_1$ with respect to a charge drop schedule $C_1$ that satisfies $g^{C_1}(P_1) = D[uv][w]$. Similarly there is a $\overunderline{wax}{1-1}{2-2}$ path $P_2$ with respect to a charge drop schedule $C_2$ that satisfies $g^{C_2}(P_2) = D[wa][x]$. Let $P = P_1\mid P_2$ and let $C$ be the concatenation $C_1$ and $C_2$. Clearly $g^C(P) = g^{C_1}(P_1) + g^{C_2}(P_2) = D[uv][x]$. We prove $P$ is $uv$-bounded with respect to $C$ and therefore, by Lemma~\ref{lemma:alpha-of-arc-bounded}, $P$ is traversable. Since $P_1$ is $uv$-bounded with respect to $C_1$ and $P_2$ is $wa$-bounded with respect to $C_2$, it is enough to prove that the gain at $a$ (with respect to $P$ and $C$) is bounded between the gains of $u$ and $v$. Since $M[w][a]<0$ it follows that $g_a^{P,C} \le g_w^{P,C} \le g_u^{P,C} = 0$. Let $d_a$ be the charge drop at $a$ induced by $C_2$. Since $P_2$ is $wa$-bounded with respect to $C_2$, it follows by Definition~\ref{def:arc-bounded} that $d_a = 0$. We get 
    \begin{align*}
        g_a^{P,C} &= g^{C_1}(P_1) + (M[w][a] - d_a) 
        = D[uv][w] + M[w][a]
        \numge{1} M[u][v] =g_v^P,
    \end{align*}    

    Where inequality $(1)$ holds by Equation~(\ref{eq:2}).
    Since $P$ is $uv$-bounded, and by Invariant~\ref{invariant}\ref{1c} $M[u][v]\ge -B$, we conclude that $P$ is traversable.

    The case in which $M[v][u]<0$ and the algorithm set $D[x][vu] = D[x][aw]+D[w][vu]$ is symmetric.
\end{proof}

\subsubsection{Dominating Funnels}\label{sec:funnels}
This procedure returns a data structure $D$ such that every funnel, that is a simple path
in $G^M$, is dominated by $D$ w.h.p.. This is done in $4$ steps. Let $s = \tilde{O}(n^\beta)$, where $\beta=2/3$. The first step is to compute bounded paths that dominate funnels of length $n/s$. This is done by running $\BFS(D,M)$  $n/s$ times. Correctness of this step follows from Corollary~\ref{lemma:BFS-basic}.

In the second step, we sample a set $S$ of $\Theta(s \log n)$ vertices. For every triplet $s_1,s_2,s_3 \in S$ we try to concatenate 
a $\overunderline{{s_1}{a_1}{s_2}}{1-1}{2-2}$ path with a $\overunderline{{s_2}{a_2}{s_3}}{1-1}{2-2}$ path, where $a_1,a_2\in V$. We also concatenate the symmetric paths: a $\overunderline{{s_3}{a_1}{s_2}}{2-2}{3-3}$ path with a $\overunderline{{s_2}{a_2}{s_3}}{2-2}{3-3}$ path.
This is done by applying $\Concat(D,S,S,S)$  $\log n$ times, see Appendix~\ref{section:concatenate}.
Each of these $\log n$ iterations multiplies the length of the funnels between vertices of $S$ that $D$ dominates.
We show that after the second step, $D$ dominates all funnels that are simple paths.
that start and end at  vertices from $S$. Lemma \ref{lemma:concat-sampling-correctness} proves the correctness of this step.

In the third step we call $\Concat(D,S,S,V)$, which for every $s_1,s_2\in S$ and $v\in V$ concatenates
$\overunderline{{s_1}{a_1}{s_2}}{1-1}{2-2}$ paths with $\overunderline{{s_2}{a_2}{v}}{1-1}{2-2}$ paths, where $a_1,a_2\in V$. We also concatenate the symmetric paths:  $\overunderline{{v}{a_1}{s_2}}{2-2}{3-3}$ paths with $\overunderline{{s_2}{a_2}{s_1}}{2-2}{3-3}$ paths.
We show that after 
the third step, $D$ dominates every simple funnel that is a $\overunderline{suv}{1-1}{2-2}$ path or a $\overunderline{vus}{2-2}{3-3}$ path, where $s\in S$ and $u,v\in V$. That is a funnel that starts with a sampled vertex and ends at an arbitrary vertex or a funnel that ends with a sampled vertex and starts at an arbitrary vertex.
This happens since each such funnel that ends at a vertex $v$ contains w.h.p.\ a sampled vertex $s$, such that the funnel from $s$ to $v$ starts with a negative gain arc and is of length at most $n/s$.

Finally, in the fourth step we run $\BFS(D,M)$ again  $n/s$  times. 
This extends w.h.p.\ the funnels that we cover to include all simple funnels of linear length (that start at any vertex). Lemma \ref{lemma:funnels-optimality} proves the correctness of this entire procedure.

\begin{figure}[t]
\begin{algorithm}[H]
  \Fn{$\ComputeF(M)$}{
      \BlankLine
      $D \gets Init\mathtt{-}DS(M)$\;
      $s \gets \Theta(n^\beta)$ \;
       \For(\tcp*[f]{\rm Finding funnels of length $n/s$}){$i=1,\ldots, n/s$}{
        $\BFS(M,D)$ 
      }
      $S\gets Sample(V, p= \log n \cdot s/n)$ \tcp*{\rm Each vertex is sampled i.i.d}
      \For{$iteration =1 \ldots \log n$}{
        $\Concat(M,D,S,S,S)$ \tcp*{\rm Dominate funnels between sampled vertices}
      }
      $\Concat(M,D,S,S,V)$ \tcp*{\rm Compute suffixes of funnels (from sampled vertices)}
      \For(\tcp*[f]{\rm Fully compute funnels}){$i=1,\ldots, n/s$}{
        $\BFS(M,D)$ 
      }
      \Return $D$
  }
\end{algorithm}
\caption{After this procedure every funnel $P$ in $G^M$ is dominated by $D$ w.h.p.}\label{figure:funnels}
\end{figure}

\begin{lemma}\label{lemma:concat-sampling-correctness}
    Let $P = v_1 \ldots v_k$ be a funnel which is negative arc-bounded and let $S$ be the set sampled by the procedure $\ComputeF$. Assume $v_1,v_k \in S$, then  w.h.p. after applying $\Concat(M,D,S,S,S)$ $\log n$ times in $\ComputeF(M)$, $D$ dominates $P$.
\end{lemma}

\begin{proof}
Assume that $P$ is first-arc-bounded path. The case in which $P$ is last-arc-bounded is symmetric.
    If $k \le n/s$ then the claim follows by Corollary~\ref{lemma:BFS-basic}. Assume $k > n/s$. Divide $P$  into continuous segments each of length $n/2s$.
    Let
    $I_t = \{t\cdot n/2s + 1, \ldots  (t+1)\cdot n/2s\}$
    be the set of indices of the vertices of segment $t$
    for 
    $0\le t \le k/(n/2s)-1$.\footnote{We assume  from brevity that $k$ is a multiple of $n/2s$, otherwise the last segment is shorter, but it does not affect the argument.} By the choice of $S$, for every $t$ it holds w.h.p.\ that there exists $i_t \in I_t$ such that $v_{i_t} \in S$ and the arc $v_{i_t}v_{i_t+1}$ has negative gain. Thus, for every $t$, $i_{t+1} -i_t \le n/s$ and therefore (by Corollary~\ref{lemma:BFS-basic}) $D$ dominates the sub-funnel $v_{i_t} \ldots v_{i_{t+1}}$. 
    Therefore, after the first call to $\Concat(M,D,S,S,S)$, by Lemma~\ref{lemma:concat-basic}, $D$ dominates $v_{i_t} \ldots v_{i_{t+2}}$ for every $t< 2s-2$.
    It follows by a simple induction that after the $j$'th call to $\Concat(M,D,S,S,S)$, $D$ dominates $v_{i_a} \ldots v_{i_{b}}$ for every $1\le a < b < 2s$ where $b-a \le 2^j$.
\end{proof}

\begin{lemma}\label{lemma:funnels-optimality}
    Let $P$ be a funnel of length $|P|=O(n)$. After a call to $\ComputeF(M)$, $D$ dominates $P$ w.h.p.
\end{lemma}

\begin{proof}
Denote $P= v_1 \ldots v_k$ and assume $P$ is $v_1v_2$-bounded, the case of a last-arc bounded funnel is symmetric. For every $1\le i \le j\le k$ we denote $P^{ij}=v_i\ldots v_j$.

If $k \le n/s$ then the claim follows by Corollary~\ref{lemma:BFS-basic}. Assume $k > n/s$. Let $A = \{1,\ldots n/2s\},B = \{k-n/2s , \ldots k-1\}$ be sets of  the first $2n/s$ indices and last $2n/s$ indices. By the sampling probability of the nodes to $S$ we get that w.h.p.\ there exists $a \in A$ and $b \in B$ such that $v_{a},v_{b}\in S$ and $M[v_{a}][v_{a + 1}]<0$ and $M[v_{b}][v_{b + 1}]<0$.\footnote{We may assume that $M[v_{a}][v_{a + 1}]<0$ and $M[v_{b}][v_{b + 1}]<0$ since half the arcs in a funnel are of negative gain (Lemma~\ref{lemma:funnel-zigzag-structure}).} By Lemma~\ref{lemma:concat-sampling-correctness}, w.h.p., after the $\log n$ applications of $\Concat(M,D,S,S,S)$, $D$ dominates $P^{a b}$. Since $k - b \le n/s$, by Corollary~\ref{lemma:BFS-basic}, after the first $n/s$ call to $\BFS(D,M)$,~$D$ dominates $P^{b k}$. By applying Lemma~\ref{lemma:concat-basic} on $P_1= P^{a b}$ and $P_2=P^{b k}$, we conclude that after performing $\Concat(M,D,S,S,V)$ it holds that~$D$ dominates $P^{a k}$. Finally, since $a < n/s$, we get by Lemma~\ref{lemma:BFS-single-extension} that after the last $n/s$ calls to $\BFS(D,M)$, $D$ dominates $P$.
\end{proof}

\subsubsection{Concatenating first-arc-bounded paths with last-arc-bounded paths}\label{section:CO}
Similarly to $\Concat(M,D)$, in $\CO(M,D)$ we are given $3$ sets $U,W,X\subseteq V$. For every $u\in U,w\in W,x\in X$ and $v\in V$, we try to create a $\overunderline{uvx}{1-1}{2-2}$ path by concatenating a $\overunderline{uvw}{1-1}{2-2}$ path with a $\overunderline{wax}{3-3}{2-2}$ path, where we optimize over the choices of $a \in V$, see Figure~\ref{figure:concat-opposite}. We also do the symmetric computation: we try to create a $\overunderline{xvu}{2-2}{3-3}$ path by concatenating a $\overunderline{xaw}{2-2}{1-1}$ path with a $\overunderline{wvu}{2-2}{3-3}$ path.
Notice that in either case the new path that we create is negative arc-bounded.

We now elaborate on the case corresponding to concatenating  $\overunderline{uvw}{1-1}{2-2}$ path with a $\overunderline{wax}{3-3}{2-2}$ path.
Assume $M[u][v]<0$, we update $D[uv][w]$ as follows. We consider the values $D[w][ax]$, for every $a\in V$ that satisfies $M[a][x]>0$ and $M[a][x]-D[w][ax]\le D[uv][w]-M[u][v]$. The latter condition guarantees that the gain of $a$ is larger than the gain of $v$ with respect to the concatenation of the paths realizing $D[uv][w]$ and $D[w][ax]$, see Figure~\ref{figure:concat-opposite}.
We distinguish between the following two cases.

\textbf{Case $1$: 
$M[a][x]-D[w][ax] \le D[uv][w]-M[u][v]$ and $D[w][ax] \le |D[uv][w]|$:} 
This case corresponds to Figure~\ref{figure:concat-opposite}(a).
The first condition says that the gain of the concatenated path never goes below the gain of $v$ (i.e $g_v=M[u][v] < 0$) and the second condition says that the gain of the concatenated path never exceeds the gain of $u$ (i.e., $g_u=0$).
In this case we claim that the paths realizing $D[uv][w]$ and $D[w][ax]$ can be concatenated into a $\overunderline{uvx}{1-1}{2-2}$ path of gain $D[uv][w]+D[w][ax]$.  We find such an $a\in V$ with largest $D[w][ax]$ and perform the update $D[uv][x]= \max \{D[uv][x], D[uv][w] + D[w][ax] \}$. To find the best $a\in V$, we store for every $w\in W$ and $x\in X$ the pairs $(M[a][x]-D[w][ax], D[w][ax])$, for $a\in V$ satisfying $M[x][a]>0$, in a \emph{Range Tree} $LRT_{wx}$ of last-arc-bounded paths. We then perform a search in $LRT_{wx}$ for a pair $(k_1,k_2)$ with $k_1 \le D[uv][w]-M[u][v]$ and largest~$k_2$ that satisfies $k_2 \le |D[uv][w]|$. This operation takes $O(\log^2 n)$ time.

\textbf{Case $2$:
  $ D[w][ax] \ge |D[uv][w]|$ and $M[a][x] \le |M[u][v]|$.}
 This case corresponds to Figure~\ref{figure:concat-opposite}(b).
  The first condition implies that the gain at $x$ is larger than the gain at $u$. Note that the condition from Case $1$ $M[a][x]-D[w][ax] \le D[uv][w]-M[u][v]$ can be derived from the two conditions.
 In this case we set $D[uv][x]=0$. To justify this assignment we argue that there is a path $P$ and an associated charge drop schedule $C$ such that $P$ is a $\overunderline{uvx}{1-1}{2-2}$ path with respect to $C$ and $g^C(P)=0$.
 Let $(P_1,C_1)$ and $(P_2,C_2)$ be the paths and charge drop schedules realizing $D[uv][w]$ and $D[w][ax]$, respectively.
Let $P = P_1\mid P_2$. We define a charge drop schedule $C$ for $P$ as follows: Let $d_w$ be the last charge drop in $C_1$ associated with $w$. We get $C$ by concatenating $C_1$ and $C_2$ and changing $d_w$ to be equal to $d_w+g^{C_1}(P_1)+g^{C_2}(P_2)$.

We claim that $P$ is $uv$-bounded with respect to $C$ and $g^C(P) = 0$. The latter is clear since 
\begin{align*}
    g^C(P) = g^{C_1}(P_1)- (g^{C_1}(P_1)+g^{C_2}(P_2)) + g^{C_2}(P_2) = 0.
\end{align*}
Lemma \ref{lemma:concat-opposite-invariant} shows that $P$ is $uv$-bounded with respect to $C$.

We discover whether there exists a vertex  $a\in V$ for which we should apply this case as follows. For every $w\in W$ and $x\in X$,
we store the pairs $(D[w][ax], M[a][x])$, for $a\in V$ satisfying $M[x][a]>0$, in a $2$-dimensional \emph{Range Tree} $LRT'_{wx}$ of values realized by $\overunderline{wax}{3-3}{2-2}$ paths.
We then perform a search in $LRT'_{wx}$ for a pair $(k_1,k_2)$ with $k_1 \ge |D[uv][w]|$ and $k_2\le |M[u][v]|$. This operation is done in $O(\log^2 n)$ time. If we find such a pair, we
apply this case and set
$D[uv][x]=0$.

The symmetric version, i.e., concatenating a $\overunderline{xaw}{2-2}{1-1}$ path with a $\overunderline{wvu}{2-2}{3-3}$ path, is as done analogously, see Figure~\ref{figure:concat-opposite}. We search for $a\in V$, such that $M[x][a]>0$ and one of the following cases is satisfied:

\textbf{Case $3$: 
$M[x][a]-D[xa][w] \le D[w][vu]-M[v][u]$ and $D[xa][w] \le |D[w][vu]|$:} 
This case corresponds to Figure~\ref{figure:concat-opposite}(c).
Similarly to Case $1$, in this case we can concatenate the paths realizing $D[xa][w]$ and $D[w][vu]$. we find $a$ that maximize $D[xa][w]$ and perform the update $D[x][vu] = \max \{D[x][vu], D[xa][w] + D[w][vu] \}$.

\textbf{Case $4$:
 $ D[xa][w] \ge |D[w][vu]|$ and $M[x][a] \le |M[v][u]|$.}
 This case corresponds to Figure~\ref{figure:concat-opposite}(d).
 Similarly to Case $4$, in this case in order to concatenate the paths realizing $D[xa][w]$ and $D[w][vu]$ we have to perform a charge drop at $w$. This will give us a $\overunderline{xvu}{2-2}{3-3}$ path of gain $0$ (with respect to some charge drop schedule), this is the best we can hope for in a negative arc-bounded path. We search if such an $a\in V$ exists using a Range tree $FRT'_{xw}$. If so, we perform the update $D[x][vu] = 0$.

\begin{figure}[t!]
\begin{algorithm}[H]
  \Fn{$\CO(M,D,U,W,X)$}{
  \BlankLine
  \For{$(w,x)\in  W \times  X$}{
    $LRT_{wx} \gets RT()$ \tcp*{\rm $2D$-Range tree of $w\ubar{a}\bar{x}$ paths}
    \For{$a\in V$ s.t $M[a][x]\ge 0$}{
        $LRT_{wx}.insert(M[a][x]-D[w][ax], D[w][ax])$
    }
    \BlankLine
    $LRT'_{wx} \gets RT()$ \tcp*{\rm $2D$-Range Tree of $w\ubar{a}\bar{x}$ paths}
    \For{$a\in V$ s.t $M[a][x]\ge 0$}{
        $LRT'_{wx}.insert(D[w][ax], M[a][x])$ \;
    }
    $FRT_{xw} \gets RT()$ \tcp*{\rm $2D$-Range tree of $\ubar{x}\bar{a}w$ paths}
    \For{$a\in V$ s.t $M[x][a]\ge 0$}{
        $FRT_{xw}.insert(M[x][a]-D[xa][w],D[xa][w])$
    }
    \BlankLine
    $FRT'_{xw} \gets RT()$ \tcp*{\rm $2D$-Range tree of $\ubar{x}\bar{a}w$ paths}
    \For{$a\in V$ s.t $M[x][a]\ge 0$}{
        $FRT'_{xw}.insert(D[xa][w], M[x][a])$
    }
  }
  \For{$(u,v,w,x)\in U\times V \times W \times  X$}{
      \If(\tcp*[h]{Trying to create a $\bar{u}\ubar{v}x$ path}){$M[u][v]<0$}{
          $(-,D[w][ax]) \gets LRT_{wx}.range(k_1\le D[uv][w]-M[u][v], k_2\le |D[uv][w]|).max\_k_2()$\;
          $D[uv][x] \gets \max \{D[uv][x], D[uv][w]+D[w][ax] \}$ \;
          
          $bool \gets LRT'_{wx}.find(k_1\ge |D[uv][w]|,k_2 \le |M[u][v]|)$\;
          \If{bool}{
            $D[uv][x]\gets 0$
          }
      }
      \If(\tcp*[h]{Trying to create a $x\bar{v}\ubar{u}$ path}){$M[v][u]<0$}{
          $(-,D[xa][w]) \gets FRT_{xw}.range(k_1\le D[w][vu]-M[v][u],k_2\le |D[w][vu]|).max\_k_2()$\;
          $D[x][vu] \gets \max \{D[x][vu], D[xa][w] + D[w][vu] \}$\;
          
          \BlankLine
          $bool \gets FRT'_{xw}.find(k_1\ge |D[w][vu]|,k_2\le |M[v][u]|)$\;
          \If{bool}{
            $D[x][vu]\gets 0$
          }
      }
  }
  }
\end{algorithm}

\begin{centering}
    \includegraphics[width=0.85\textwidth]{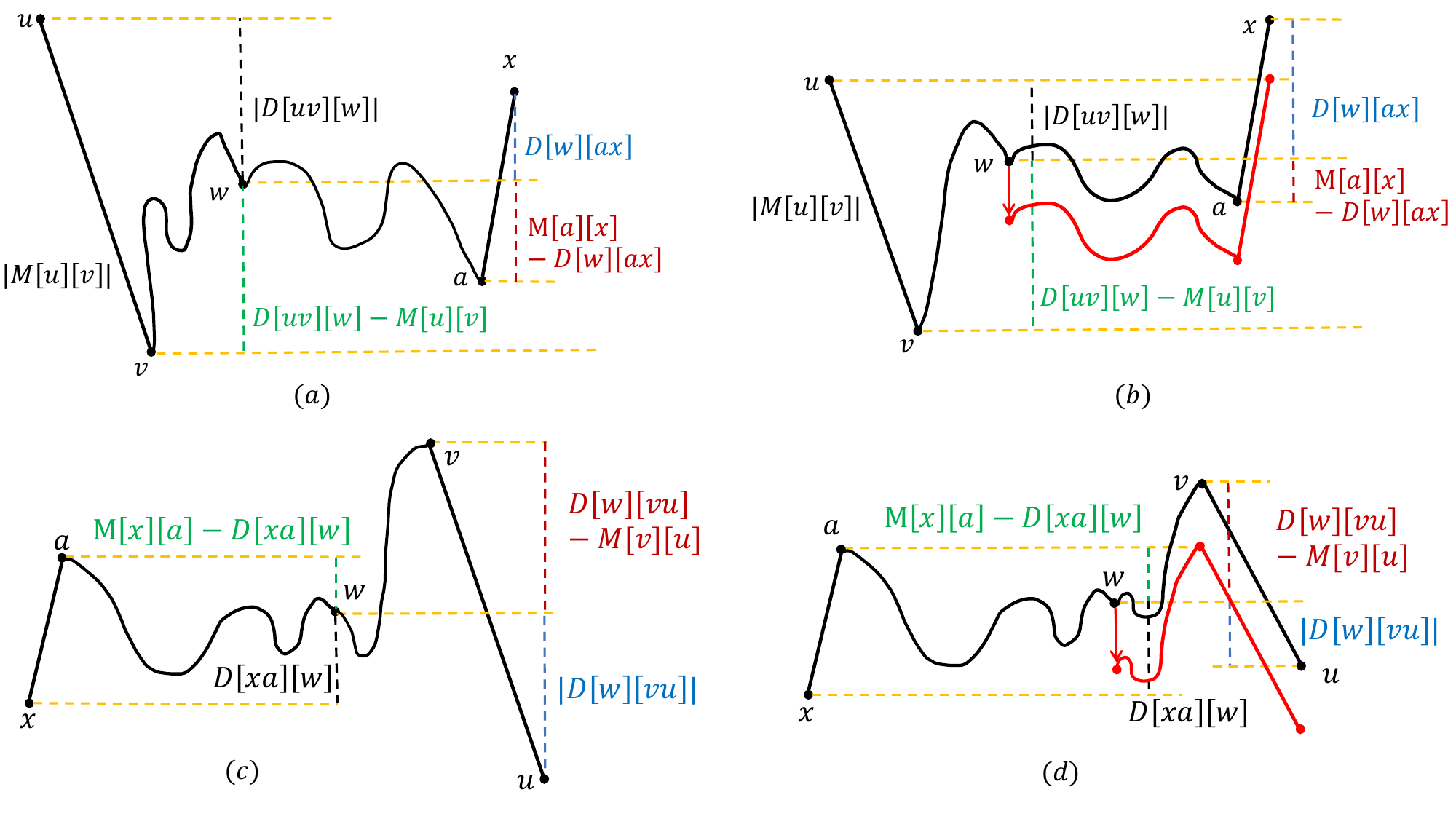}
\end{centering}

\caption{Concatenating  arc-bounded paths, one is first-arc-bounded  and the other is last-arc-bounded.}\label{figure:concat-opposite}
\end{figure}

\begin{lemma}\label{lemma:concat-opposite-bounded}
    Let $P$ and $Q = $ be paths in $G^M$ that are dominated by $D$. Let $U,W,X\subseteq V$ and let $u\in U,w\in W,x\in $ and $v \in V$. If one of the following holds
    \begin{itemize}
        \item $P$ is a $\overunderline{uvw}{1-1}{2-2}$ path, $Q$ is a $\overunderline{wax}{3-3}{2-2}$ path, and $P\mid Q$ is a $\overunderline{uvx}{1-1}{2-2}$ path.
        \item $P$ is a $\overunderline{xaw}{2-2}{1-1}$ path, $Q$ is a $\overunderline{wvu}{2-2}{3-3}$ path, and $P\mid Q$ is a $\overunderline{xvu}{2-2}{3-3}$ path.
    \end{itemize}
    then after $\CO(M,D,U,W,X)$, $D$ dominates $P \mid Q$.
\end{lemma}

\begin{proof}
     Assume the first case, the second case is symmetric. Since $P \mid Q$ is a $\overunderline{{u}{v}{w}}{1-1}{2-2}$ path, we get that $g_{a}\ge g_{v} = M[u][v]$ (gains are with respect to $P \mid Q$). Thus 
     \begin{align*}
     D[u v][w] + D[w][ax]\ge g(P) + g(Q) = 
      g_{x} = g_{a} + M[a][x] \ge g_v+M[a][x] = M[u][v]+M[a][x].
     \end{align*}
     Rearranging the terms, we get $M[a][x]-D[w][ax]\le D[u v][w]-M[u][v]$. We split to cases analogue to the cases in the description of the algorithm.

     \textbf{Case 1: $D[w][ax]\le |D[uv][w]|$.}
     Let $(k_1,k_2)=(M[a'][x]-D[w][a'x], D[w][a'x])$ be the pair in $LRT_{wx}$ with $k_1 \le D[u v][w]-M[u][v]$ and largest $k_2$ that satisfies $k_2 \le |D[uv][w]|$. Since $(M[a][x]-D[w][ax], D[w][ax])$ is also a pair in $LRT_{wx}$, we get $D[w][a'x]=k_2 \ge D[w][ax]$. Consider the tuple $(u,v,w,x)$. Thus, when $D[uv][x]$ is updated according to this tuple, the first If statement performs the following update     \begin{align*}
         D[uv][x] \ge D[uv][w]+D[w][a'x] \ge D[uv][w]+D[w][ax] \ge g(P) + g(Q) = g(P\mid Q).
     \end{align*}

    \textbf{Case 2: $D[w][ax]\ge |D[uv][w]|$ and $M[a][x]\le |M[u][v]$.} 
    In this case the algorithm assigns $D[uv][x]=0$. Since $P\mid Q$ is a $\overunderline{uvx}{1-1}{2-2}$ path, it follows that $g(P\mid Q) \le 0$ and therefore $D[uv][x]\ge g(P \mid Q)$.
\end{proof}

\begin{lemma}\label{lemma:concat-opposite-invariant}
    Procedure $\CO(M,D,U,W,X)$ maintains Invariants~\ref{invariant}\ref{1a} and \ref{invariant}\ref{1b}.
\end{lemma}

\begin{proof}
    We prove the lemma by induction on the assignments of the algorithm. Let $u\in U,w\in W,x\in X$ and $v\in V$ and assume $M[u][v]<0$, the case $M[v][u]<0$ is symmetric.
    
    Assume the algorithm set $D[uv][x] = D[uv][w]+ D[w][ax]$ for some $a\in V$ satisfying
    \begin{align}\label{eq:1}
     M[a][x]-D[w][ax] &\le D[uv][w]-M[u][v], \; \text{and} \\
     D[w][ax] &\le |D[uv][w]|.
    \end{align}
    
    In particular $D[uv][x]\le 0$.
    By Invariant~\ref{invariant}\ref{1a}, there is traversable a $\overunderline{uvw}{1-1}{2-2}$ path $P_1$ with respect to a charge drop schedule $C_1$ that satisfies $g(P_1) = D[uv][w]$. Similarly there is a $\overunderline{wax}{3-3}{2-2}$ path $P_2$ with respect to a charge drop schedule $C_2$ that satisfies $g(P_2) = D[w][ax]$. Consider the path $P = P_1\mid P_2$.  Let $C$ be the charge drop schedule derived by following $C_1$ on $P_1$ and then $C_2$ on $P_2$. We get 
    \begin{align*}
        g^{C}(P) = g^{C_1}(P_1) + g^{C_2}(P_2) = D[uv][w]+D[w][ax] = D[uv][x].
    \end{align*}

    We now prove that $P$ is $uv$-bounded with respect to $C$.  
    Since $P_1$ is $uv$-bounded with respect to $C_1$ and $P_2$ is $ax$-bounded with respect to $C_2$, it is enough to prove that $g^{P,C}_a \ge g^{P,C}_v$ and $g^{P,C}_x\le g^{P,C}_u$. Indeed, $g^{P,C}_x = g^C(P)= D[uv][x] \le 0 = g^{P,C}_u$. Let $d_x\ge 0$ be the charge drop performed at $x$ in $C$. We get
    \begin{align*}
        g^{P,C}_a 
        &= g^C(P)-(M[a][x]-d_x)
        \ge g^{C_1}(P_1) + g^{C_2}(P_2) - M[a][x] \\
        &= D[uv][w] + D[w][ax] - M[a][x]
        \numge{1} M[u][v] 
        \numge{2} g^{P,C}_v,
    \end{align*}
    where inequality $(1)$ holds by the left inequality in Equation~(\ref{eq:1}). Note inequality $(2)$ is not necessarily an equality since there might be a charge drop at $v$. Since $P_1$ is traversable it holds that $M[u][v]\ge -B$ and therefore, By Lemma~\ref{lemma:alpha-of-arc-bounded}, $P$ is traversable.

    Assume the algorithm set $D[uv][x] = 0$ because there is a vertex $a\in V$ such that
    \begin{align}
        M[a][x]\le  |M[u][v]| \\
        D[w][ax]\ge  |D[uv][w]|
    \end{align}
    $M[a][x]\le  |M[u][v]|$ and 
    $D[w][ax]\ge  |D[uv][w]|$ and $M[a][x]-D[w][ax]\le D[uv][w]-M[u][v]$. In particular $D[uv][w] + D[w][ax]\ge 0 $.  Define $P = P_1\mid P_2$ and $C$ as before. Let $\gamma = g^C(P)$, therefore $\gamma =  g^{C_1}(P_1) + g^{C_2}(P_2) = D[uv][w]+D[w][ax] \ge 0$.  Let $d_w$ be the last charge drop in $C_1$ associated with $w$. We define a charge drop schedule $C'$ that differs from $C$ only at $w$ and assigns a charge drop at $w$ of $d_w+\gamma$. We prove that $P$ is a $\overunderline{uvx}{1-1}{2-2}$ path with respect to $C'$ and $g^{C'}(P)=D[uv][w]$. The latter follows by the following calculation

    \begin{align*}
        g^{C'}(P)
        = g^{C_1}(P_1) + (g^{C_2}(P_2)-\gamma)
        = D[uv][w] + D[w][ax] - \gamma
        = 0
        =D[uv][w].
    \end{align*}
    
    We now prove that $P$ is $uv$-bounded with respect to $C'$ and therefore, By Lemma~\ref{lemma:alpha-of-arc-bounded}, $P$ is traversable. Since $P_2$ is a $\overunderline{wax}{3-3}{2-2}$ path with respect to $C_2$ (and also with respect to the new charge drop at its first vertex $w$), it is enough to show that  $g^{P,C'}_a \ge g^{P,C'}_v$ and $g^{P,C'}_x\le g^{P,C'}_u$. Indeed, $g^{P,C'}_x = g^{C'}(P)= 0 = g^{P,C'}_u$. Let $d_x\ge 0$ be the charge drop performed at $x$ in $C'$. We get
    \begin{align*}
        g^{P,C'}_a 
        &= g^{C'}(P)-(M[a][x]-d_x)
        = - M[a][x] +d_x \\
        &\ge M[u][v] +d_x
        \ge M[u][v] 
        \ge g^{P,C'}_v.
    \end{align*}     
\end{proof}

\subsubsection{Build monotone paths from arc-bounded paths}\label{sec:extend}
In this procedure (see Figure~\ref{figure:extend}) we are given a set $T\subseteq V$. For every $u\in T$ and $v,w,x\in V$, we try to concatenate a $\overunderline{uvx}{1-1}{2-2}$ path with the arc $xy$ in order to either get a descending path from $u$ to $y$ or to get an ascending path from $v$ to $y$. 
We also do the opposite:
we try to concatenate the arc $yx$ with a $\overunderline{xvu}{2-2}{3-3}$ path in $G^M$ in order to either get a descending path from $y$ to $u$ or to get an ascending path from $y$ to $v$.

The concatenation of $\overunderline{uvx}{1-1}{2-2}$ paths with the arc $xy$ is done as follows. We distinguish between the following cases.

\textbf{Case 1: $-B \le D[uv][x]+M[x][y]\le M[u][v]$.}
This case corresponds to a concatenation of a $\overunderline{uvx}{1-1}{2-2}$ path with the arc $xy$ that results in a descending path from $u$ to $u$. The algorithm sets $D[u][y] = \max \{   D[u][y],D[uv][x]+M[x][y]\}$.

\textbf{Case 2: $D[uv][x]+M[x][y] \ge M[u][v]$}
This case means that after the concatenation, the gain at $y$ is at least the gain at $v$. In order to make the path descending, we perform a charge drop at~$y$ such that the gain at $y$ matches the gain of $v$, resulting in a descending path. The algorithm sets $D[u][y] = \max \{   D[u][y], M[u][v]\}$.

\textbf{Case 3: $D[uv][x]+M[x][y] \ge 0$}
This case means that after the concatenation, the gain at $y$ is at least the gain at $u$ (which is $0$). This means that $y$ has the maximum gain in the concatenated path. Since $P$ is $uv$-bounded, $v$ has the minimum gain in the concatenated path. Therefore the sub path from $v$ to $y$ is ascending. The algorithm sets $D[u][y] = \max \{   D[u][y], -M[uv] + D[uv][x]+M[x][y]\}$.

The procedure performs similar computations when it concatenates the arc $yx$ with a $\overunderline{xvu}{2-2}{3-3}$ paths.

\begin{figure}[hbt!]
\begin{algorithm}[H]
  \Fn{$\Extend(M,D,T)$}{
    \BlankLine
  \For{$(u,v,x,y)\in T \times V^3$}{
      \BlankLine
      \If(\tcp*[h]{first-arc bounded paths}){$M[u][v]\le 0$}{
          \If(\tcp*[h]{descending shortcuts}){$-B \le D[uv][x]+M[x][y]\le M[u][v]$}{
            $D[u][y] = \max \{   D[u][y],D[uv][x]+M[x][y]\}$ 
          }
          \If(\tcp*[h]{descending shortcuts}){$D[uv][x]+M[x][y] \ge M[u][v]$ }{
            $D[u][y] = \max \{D[u][y], M[u][v] \}$\tcp*{\rm  Drop charge at $y$ to be descending.}
          }
          \If(\tcp*[h]{ascending shortcuts}){$D[uv][x]+M[x][y] \ge 0$ }{
            $D[v][y] = \max \{D[v][y], -M[u][v]+D[uv][x]+M[x][y] \}$ \\
            \tcp*{\rm  Note: Ascending path starts at $v$.}
          }
      }
      \BlankLine
      \If(\tcp*[h]{last-arc bounded paths}){$M[v][u]\le 0$}{
          \If(\tcp*[h]{descending shortcuts}){$-B\le M[y][x]+D[x][vu]\le M[v][u]$ }{
            $D[y][u] = \max \{D[y][u], M[y][x]+ D[x][vu]\}$
          }
          \If(\tcp*[h]{descending shortcuts}){$M[y][x]+D[x][vu]\ge M[v][u]$ }{
            $D[y][u] = \max \{D[y][u], M[v][u] \}$\tcp*{\rm Charge drop at $x$ to make $g^C_y=g^C_v$.}
          }
          \If(\tcp*[h]{ascending shortcuts}){$M[y][x]+D[x][vu]\ge 0$}{
            $D[y][v] = \max \{   D[y][v],M[y][x]+D[x][vu]-M[v][u]\}$ \\
            \tcp*{\rm  Note: Ascending path ends at $v$.}
          }
        }
  }
  }
\end{algorithm}

\begin{center}    \includegraphics[width=1.10\textwidth]{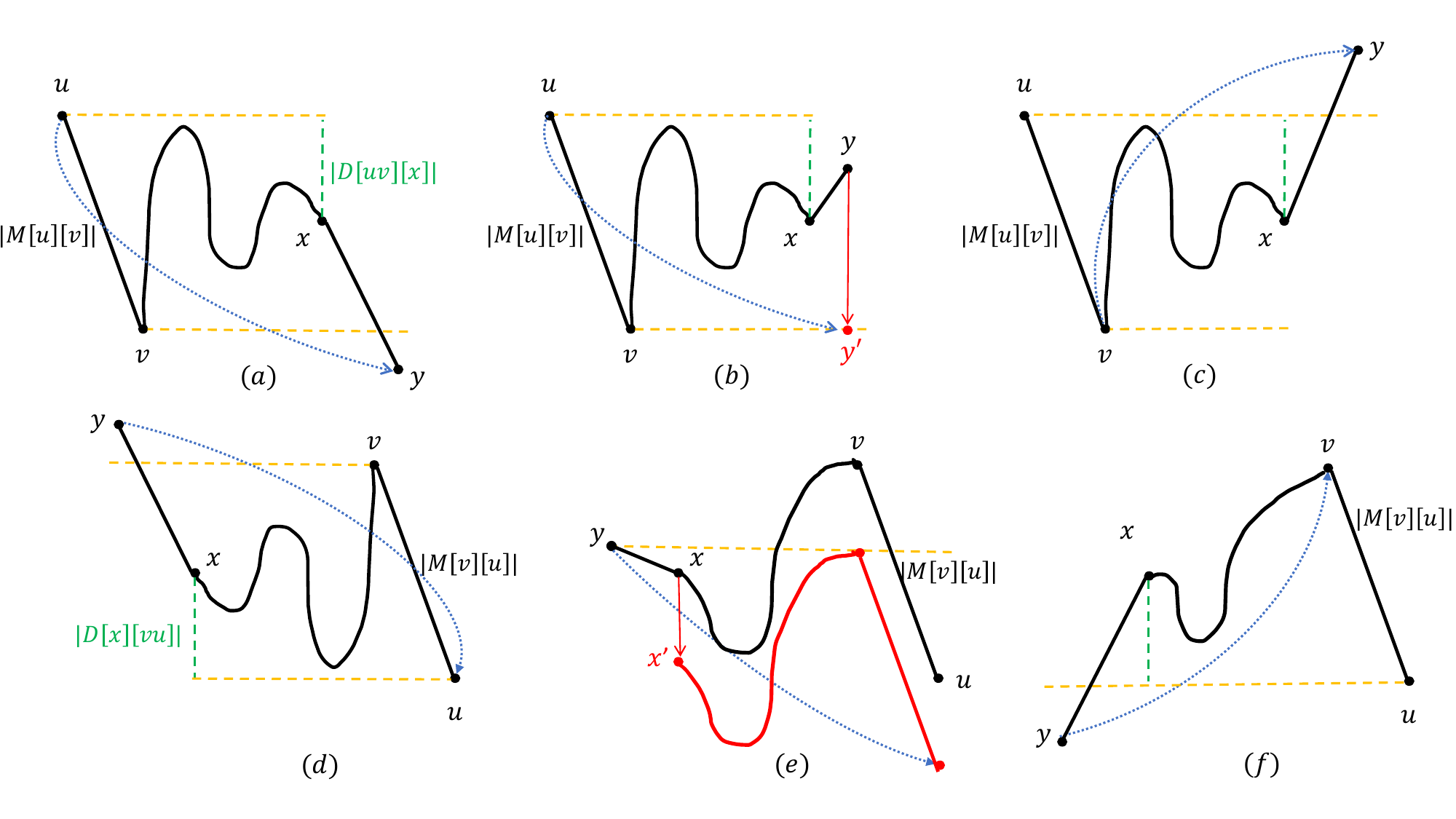}
\end{center}

\caption{The six cases of the Algorithm $\Extend(M,D,T)$. Blue dotted arrows are new shortcuts. Red downwards vertical arrows represent charge drop at a vertex, this affects the gain of all subsequent vertices.}\label{figure:extend}
\end{figure}

The following lemma states that if $D$ dominates a first-arc $uv$-bounded path $P$, where $u\in T$, that can be extended by an arc $xy$ and result in a monotone path $P'$ (that starts either at $u$ or $v$), then after $\Extend(M,D,T)$, $D$ dominates $P'$. The lemma also proves a similar result for last-arc bounded paths.

\begin{lemma}\label{lemma:long-shortcuts-dominating}
   Let $P$ be an arc-bounded path in $G^M$ and assume that $D$ dominates $P$. Denote by $\tilde{P}$ the subpaths of $P$ that excludes the bounding arc of $P$ (either the first arc or the last arc). Then the following holds after $\Extend(M,D,T)$
    \begin{enumerate}
        \item Assume $P$ is a $\overunderline{uvx}{1-1}{2-2}$ bounded and $P' = P\mid y $ is descending such that $g(P')\ge -B$. If $u \in T$, then $D$ dominates $P'$.
        \item Assume $P$ is $\overunderline{uvx}{1-1}{2-2}$-bounded and $P' = \tilde{P}\mid y $ is ascending. If $u \in T$, then $D$ dominates $P'$.
        \item  Assume $P$ is a $\overunderline{uvx}{1-1}{2-2}$ bounded. If $u \in T$, then $D[u][x]\ge M[u][v]$.
        \item Assume $P$ is a $\overunderline{xvu}{2-2}{3-3}$ path and $P' = y \mid P $ is descending such that $g(P')\ge -B$. If $u \in T$, then $D$ dominates $P'$.
        \item Assume $P$ is a $\overunderline{xvu}{2-2}{3-3}$ path and $P' = y \mid \tilde{P} $ is ascending. If $u \in T$, then $D$ dominates $P'$.
        \item Assume $P$ is a $\overunderline{uvx}{2-2}{3-3}$ bounded. If $x \in T$, then $D[u][x]\ge M[v][x]$.
    \end{enumerate}
\end{lemma}

\begin{proof}
    We prove items $1,2$ and $3$ of the lemma, items $4,5$ and $6  $ are symmetric. Since $D$ dominates $P$, we get $D[u v][x] \ge g(P)$.
    
    We begin by proving item $1$ of the lemma. 
    We split to cases according to the pseudocode.
    
     If $D[u v][x]+ M[x][y] \le M[u][v]$ (See Figure~\ref{figure:extend}(a)), then after $\Extend(M,D,T)$ we get $D[u][y]\ge D[u v][x]+ M[x][y]$. Therefore

     \[ D[u][y]\ge D[u v][x]+ M[x][y] \ge g(P) + M[x][y] = g(P').\]
     
     If $D[u v][x]+ M[x][y] \ge M[u][v]$ (See Figure~\ref{figure:extend}(b)), then after $\Extend(M,D,T)$ we get $D[u][y]\ge  M[u][v] \ge g(P')$, since $P'$ is descending. Thus, in both cases $D$ dominates $P'$.

    We now prove item $2$ of the lemma  (See Figure~\ref{figure:extend}(c)). By the assumptions, $|M[u][v]| \le g(P') = g(\tilde{P}) + M[x][y] =  -M[u][v] + g(P)+M[x][y]$. Rearranging the terms, we get $g(P)+M[x][y]\ge 0$. So $D[u v][x]+M[x][y]\ge 0$. Hence, after $\Extend(M,D,T)$ we get 
    \[D[v][y]\ge - M[u][v] + D[u v][x]+M[x][y] 
    \ge - M[u][v] + g(P) +M[x][y]
    = g(P').\]

    We now prove item $3$ of the lemma (See Figure~\ref{figure:extend}(b) and set $x=y$). Since $\ComputeS$ initializes $M[w][w]=0$ for every $w\in V$, and since the values in $M$ are non decreasing, it follows that $M[x][x]\ge 0$. Therefore $D[uv][x]+M[x][x]\ge D[uv][x]\ge M[u][v]$, and by the second inner-if statement in $\Extend(M,D,T)$ we get that $D[u][x]\ge M[u][v]$.

    The proof of items $4,5$ and $6$ follows similarly, see Figures~\ref{figure:extend}(d)-(f).
\end{proof}

\begin{lemma}\label{lemms:long-shortcuts-invariant}
    Procedure $\Extend(M,D,T)$ maintains Invariant~\ref{invariant}\ref{1a}.
\end{lemma}

\begin{proof}
    We prove the lemma by induction on the assignments of the algorithm.

    Let $(u,v,x,y)\in T\times V^3$. Assume $M[u][v]\le 0$, the case $M[v][u]\le 0$ is symmetric.
    By Invariant~\ref{invariant}\ref{1a}, there is a $\overunderline{uvx}{1-1}{2-2}$ path $P$ with respect to a charge drop schedule $C$ that satisfies $g^{C}(P) = D[uv][x]$.
    We split to three cases according to the assignment to $D$ that the $\Extend(M,D,T)$ performs.

    \textbf{Case 1: $D[u][y] = D[uv][x] + M[x][y]$.} We perform this assignment only when 

    \begin{align}\label{eq:bounded-to-mono}
        -B\le  D[uv][x] + M[x][y] \le M[u][v],
    \end{align}
    
    see Figure~\ref{figure:extend}(a).
    Let $P' = P\mid y$ and let $C'$ be the charge drop schedule that concatenates $C$ with the length one schedule that does no drop charge at the last vertex $y$. It holds that $g^{C'}(P') = g^C(P)+M[x][y] = D[uv][x]+M[x][y]=D[u][y]$.  Observe that 
    \[g^{P',C'}_y = g^{C'}(P') = D[uv][x]+M[x][y] \numle{\ref{eq:bounded-to-mono}}  M[u][v] = g^{P',C'}_v  \ .\]
    Since $P$ is traversable
    and $g^{C'}(P') =  D[uv][x] + M[x][y] \numge{\ref{eq:bounded-to-mono}} -B$, it follows that $P'$ is traversable. Since $P$ is also $uv$-bounded with respect to $C$,  
    it follows that $P'$ is descending with respect to $C'$.  

    \textbf{Case 2: $D[u][y] =M[u][v]$.} We perform this assignment only when 
    \begin{align}\label{eq:bounded-to-mono2}
        D[uv][x] + M[x][y] \ge M[u][v],
    \end{align}
    
    see Figure~\ref{figure:extend}(b). Let $P,P'$ and $C$ be as in the previous case. In order to make $P'$ descending we define the charge drop schedule $C'$ that follows $C$ and then performs a charge drop at $y$ of $d_y = (g^C(P)+M[x][y]) - M[u][v]$ (i.e., we drop the gain at $y$ to be equal to the gain at $v$). Note that $d_y$ is indeed non  negative since

    \[d_y = g^C(P) + M[x][y] - M[u][v]
    = D[uv][x] + M[x][y] - M[u][v] \numge{\ref{eq:bounded-to-mono2}} 0. \]

    We claim that $g^{C'}(P')= D[u][y]$($=M[u][v]$) since \[
    g^{C'}(P')
    =g^C(P) + (M[x][y]-d_y)
    = D[uv][x] + (-D[uv][x] + M[u][v])
    =M[u][v]
    = D[u][y].
    \]

   Since $P$ is traversable and $g^{C'}(P') = M[u][v]\ge - B$, it follows that $P'$ is traversable.
    Since $P$ is $uv$-bounded with respect to $C$ and $g^{P',C'}_y = g^{C'}(P)= M[u][v]= g^{P',C'}_v$, it follows that $P'$ is descending with respect to $C'$.
    \textbf{Case 3: $D[v][y] = -M[u][v]+D[uv][x] + M[x][y]$.} We perform this assignment only when 
    
    \begin{align}\label{eq:bounded-to-mono3}
        D[uv][x] + M[x][y] \ge 0,
    \end{align}
    
    see Figure~\ref{figure:extend}(c). Let $Q$ be the suffix of $P$ that skips the first vertex $u$. Let $P'= Q\mid y$.  Since $P$ is traversable, it follows that $Q$ is traversable and therefore ($M[x][y]\ge 0$) $P'$ is traversable. Let $C'$ be the charge drop schedule that follows $C$ (but starts at $v$) and does not drop charge at $y$.
    Observe that 
    \[g^{C'}(P') = -M[u][v] + g^C(P)+M[x][y]  = -M[u][v]+D[uv][x]+M[x][y] = D[v][y].\] 
    
    Moreover, since $D[uv][x] + M[x][y] \numge{\ref{eq:bounded-to-mono3}} 0$, it follows that $g^C(P) + M[x][y] \ge 0$. This means that $y$ has larger gain (with respect to $C$) than all vertices in $P$, so $P'$ is ascending with respect to $C'$.
     It follows by Lemma~\ref{lemma:alpha-of-monotone} that $P$ is strongly traversable. 
\end{proof}

\subsubsection{Long Shortcuts}\label{section:long}

This procedure aims to find ``long shortcuts" in $G^M$. These are shortcuts that correspond to monotone paths of length $k>3$. We find such shortcuts by computing (long) arc bounded paths and then extending them by one arc into monotone paths (i.e  shortcuts) using $\Extend$.

The procedure (See Figure~\ref{figure:long-shortcuts}) starts by running $\ComputeF(M)$ in order get a data structure~$D$ that dominates each funnel in $G^M$ w.h.p. 
The procedure $\LS(M)$ samples sets $T_i$ of size $\Theta(\frac{\log^2 (n) \cdot \kappa }{2^i})$, for every $1\le i \le \log n$,\footnote{Recall the intuition from the Technical review (Section~\ref{sec:technical-review-finding-arc-bounded-paths}), the algorithm interpolates between two extreme cases, see Figure~\ref{fig:example-different-sampling}} where $\kappa=\Theta(n^{1-\alpha})$. In Appendix \ref{sec:correctness}  $\kappa$  would be a bound on the number of funnels which are maximal with respect to inclusion in a studied path $P$. For every $u\in T_i$, we concatenate $2^i$ times arc bounded paths starting at $u$ ($uw$-bounded) with other arc bounded paths. This is done using the two concatenation procedures $\Concat(M,D)$ and $\CO(M,D)$. Intuitively, each such concatenation extends the reach of a $uv$-bounded path $P$ ($u\in T_i$) by an additional funnel.

For example, if $P$ is first-arc bounded, then the procedure $\CO(M,D)$ is used in order to concatenate $P$ with last-arc bounded funnel and 
$\Concat(M,D)$ is used in order to concatenate~$P$ with first-arc bounded funnel.

Finally, after computing these arc  bounded paths, we try to extend them by one arc to get new shortcuts. We do so by running $\Extend(M,D,T)$, where $T = \cup_i {T_i}$ is the set of all sampled vertices.

\begin{figure}[t]
\begin{algorithm}[H]
  \Fn{$\LS(M)$}{
    \BlankLine
    $D \gets \ComputeF(M)$\;
    $T \gets \emptyset$ \tcp*{\rm All vertices sampled for creating shortcuts}
    \For{$i=1\ldots \log\left( {n^{1-\alpha} \log^2 n}\right)$}{
        $s_i \gets \Theta(\frac{\log^2 (n)}{2^i \cdot n^\alpha})$ \tcp*{\rm new sampling probability}
        $T_i \gets Sample(V, p=s_i)$\;
        \RepTimes{$2^i$}{
            $\CO(M,D,T_i,V,V)$ \tcp*{\rm Skip funnels of opposite direction}
            $\Concat(M,D,T_i,V,V)$\tcp*{\rm Skip funnels of the same direction}
        }
        $T \gets T \cup T_i$
    }
  \BlankLine
  $\Extend(M,D,T)$
  \BlankLine
  \Return $ D.shortcuts$
  }
\end{algorithm}
\caption{Procedure $\LS$. We sample vertices and compute arc bounded paths in which those vertices are end points. The lower the sampling probability, the further we extend our search.}\label{figure:long-shortcuts}
\end{figure}

\section{Stage I Correctness}\label{sec:correctness}

In this appendix we prove the main theorem of our shortcutting algorithm.
\begin{theorem}\label{theorem:shortcut}
    Let $P = v_1 \ldots v_k$ be a monotone simple path in $G$. Let $M$ be the shortcuts returned from $\ComputeS(G)$. Then w.h.p.\ $M[v_1][v_k]\ge g(P)$.
\end{theorem}

Theorem~\ref{theorem:shortcut} follows from the following lemma. 

\begin{lemma}\label{lemma:paths-shrink}
    Let $P= v_1\ldots v_k$ be a monotone simple path in $G^M$ with respect to a charge drop schedule $C$. Let $M'$ be the shortcuts table after running $\UpdateS(M)$. If $|P| \le  n^\alpha$, then $M'[v_1][v_k] \ge g^{C}(P)$. If $|P| > n^\alpha$, then w.h.p.\ there is a monotone path $P'$, with respect to a charge drop schedule $C'$, from $v_1$ to $v_k$ in $G^{M'}$ that satisfies $g^{C'}(P')\ge g^{C}(P)$ and $|P'| \le (1- 1/\Omega(\log n))\cdot |P|$.
\end{lemma}

Before proving Lemma~\ref{lemma:paths-shrink}, we need to introduce the concept of funnel decomposition.

\subsection{Funnel Decomposition}
 A \emph{funnel decomposition} of a path $P=e_1\ldots e_k$ in $G^M$ is a partition of $P$ into subpaths $F_1,\ldots,F_t$ which are funnels that are maximal with respect to inclusion. More precisely, the funnel decomposition of $P$ is defined by the following process. We define $F_1 = e_1 \ldots e_r$, where $1\le r\le k$, to be the maximal funnel in $P$ that contains $e_1$. Assume we have constructed $F_1,F_2,\ldots, F_s$ and denote $F_s = e_{\ell} \ldots e_r$. If $\cup_{i=1}^s F_i \neq P$, then we define $F_{s+1}= e_{\ell'}\ldots e_{r'}$ as the maximal funnel in $P$ with largest $r'$ that contains $e_{r+1}$.\footnote{Note that an arc can be contained in at most two maximal funnels. Therefore it is important to specify which maximal funnel we pick.} In particular $\ell'>\ell$. Since every arc is a funnel, it is clear that the funnel decomposition is well defined.

The following lemma proves structural properties on the funnel decomposition.
The lemma states that every two different funnels that are maximal can intersect by at most two arcs. In particular, every two consecutive funnels in the funnel decomposition overlap by at most two consecutive arcs.

\begin{lemma}\label{lemma:funnel-intersection}
    Let $P= e_1 \ldots e_k$ be a path in $G=(V,A,c)$. Let $F_1= e_a\ldots e_b$ and $F_2=e_c \ldots e_d$ be two different funnels in $P$ which are maximal with respect to inclusion. If $a<c$ then $c\ge b-1$. Moreover, if $c=b-1$ then $g(e_{b-1})=-g(e_{b})$
\end{lemma}

\begin{proof}
    If $c>b$ then we are done. 
    Otherwise $F_1$ and $F_2$ intersect and  therefore we may assume that 
    $F_1$ is last-arc bounded and $F_2$ is first-arc bounded (otherwise, by maximality they must be identical).\footnote{Note that a funnel of only two arcs of the same gain in absolute value is both first-arc bounded  and  last-arc bounded.}
    
    By contradiction, assume $c< b-1$. Therefore, $e_{b-2},e_{b-1},e_b \in F_1 \cap F_2$. It follows by the strict inequalities in Lemma~\ref{lemma:funnel-zigzag-structure} that $g(e_{b-2})$ must be both strictly larger and strictly smaller than   $g(e_{b})$
 which is a contradiction.

    Assume $c=b-1$.  Since $F_1$ is last-arc-bounded, by the weak  inequality  in Lemma~\ref{lemma:funnel-zigzag-structure}, we get $|g(e_{b})| \ge |g(e_{b-1})|$. Similarly, since $F_2$ is first-arc bounded, we get $|g(e_{b-1})| \ge |g(e_{b})|$ and therefore $g(e_{b}) = -g(e_{b-1})$.
\end{proof}

\subsection{Proof of Lemma~\ref{lemma:paths-shrink}}

We present the road map of the proof of Lemma~\ref{lemma:paths-shrink}. Let $u,v\in V$ and let $P$ be an ascending path\footnote{The same argument applies for descending paths but we will need to incorporate charge drops to the argument.} from $u$ to $v$ in $G^M$. Let $M_1,\ldots ,M_r$ be the shortcuts tables resulted after each of the $r = n^\alpha$ applications of $\Simple$ during the iterations of $\ComputeS$. During these iterations we called $\Simple$ $n^{\alpha}$ times and $\LS$ $\tilde{\Theta}(1)$ times in expectation. Let $P_i$ be the shortest ascending path from $u$ to $v$ in $G^{M_i}$ of gain larger than the gain of $P$ in $G^M$. Since the shortcuts tables $M_i$ keep increasing their gains it follows that $|P_i|$ is decreasing with $i=1,\ldots , r$.
Let us focus only on the calls of $\Simple(M_i)$ for $i=1,\ldots,r$. If after running $\Simple$ $n^\alpha$ times, the length of $P_{n^\alpha}$ is not smaller by a constant factor than the length of $P$, say $P_{n^\alpha} \ge 0.9 |P|$, it follows that in \textbf{half of the calls} to $\Simple(M_i)$ we have $|P_i| - |P_{i+1}| \le 0.2 |P|/n^{\alpha}$. 

Lets focus on an iteration $i$ such that 
$|P_i| - |P_{i+1}| \le 0.2 |P|/n^{\alpha}$. This means that in $P_i$ there are at most $ 0.2 |P|/n^{\alpha}$ short shortcuts. From this we can deduce that in the funnel decomposition of $P_i$ there are 
at most $0.2 |P|/n^{\alpha}=O(|P|/n^{\alpha})$ funnels (at the end of a maximal funnel there must be a short shortcut by the definition of a funnel).

Since in half of the calls to $\Simple(M_i)$, for
$i=1,\ldots r$, the funnel decomposition of $P_i$
 has $O(|P|/n^{\alpha})$ funnels, we get w.h.p.\,\footnote{The probability is $1-\left( 1- \frac{\log n}{r} \right)^r \ge 1-\frac{1}{n}$} that during $\UpdateS(M)$ we run $\LS(M_j)$, for some $1\le j \le r$, where $P_j$ satisfies the above (i.e., has at most $|P|/n^\alpha$ funnels in its funnel decomposition).

Let $M'$ be the matrix in which we accumulate long shortcuts in $\UpdateS(M)$. We prove in Lemma~\ref{lemma:long-shortcuts-shrinks} that if we run $\LS(M_j)$ where $P_j$ has $t=O(|P_j|/n^\alpha)$ funnels (and therefore $|P_j| = \Omega(t n^\alpha)$) in its decomposition then there is a monotone path $P'$ in $G^{M'}$ from $u$ to $v$ that satisfies $|P'| \le  \left( 1- 1/\log n \right)|P_j|$ and $g^{G^{M'}}(P') \ge g^{M_j}(P_j)$, which proves the Lemma~\ref{lemma:paths-shrink}.

To prove Lemma~\ref{lemma:long-shortcuts-shrinks}, for every arc $e\in P_j$ we consider the furthest first-arc bounded subpath $P_e$ of $P$ that starts at~$e$. Note that similarly to Lemma~\ref{lemma:long-shortcuts-dominating}, if we extend $P_e$ with the next arc in $P$, we get a monotone path of length at least $|P_e|$. We prove in Lemma~\ref{lemma:laminar} that the arc-bounded subpaths $P_e$ for $e\in E$, form a laminar set. We then argue that there is a large subset $B \subseteq \{P_e \mid e\in P_j\}$ that satisfies

\begin{enumerate}
    \item $|B| = \Omega(|P|/\log n)$
    \item $B$ has a stronger structure than laminarity: It is a union of chains $B = \cup_{k=1}^q B_k$, where a chain $B_k$ is a set of subpaths such that for every two paths $P_1,P_2\in B_k$ either $P_1\subseteq P_2$ or $P_2\subseteq P_1$.
    \item There exists $0\le f^\star \le \log n$ such that for every  $P_e\in B$, it holds that the number of maximal funnels in $P_e$ is at least $2^{f^\star}$ and less than $2^{f^\star +1 }$.
    \item Similarly to the bound $|P_j| = \Omega(t n^\alpha)$, i.e., the length of $P_j$ is larger than the number of funnels in $P_j$ by at least a factor of $n^\alpha$, we have $|B_k| = \Omega\left( \frac{2^{f^\star} n^\alpha}{\log n} \right)$ for every $1\le k \le q$. This means that the length of the longest path in $B_k$ is larger by a factor of at least $n^\alpha /\log n$ than the number of funnels inside it.
\end{enumerate}

Let $1\le k \le q$. We finish the argument by saying that because of the chain structure of $B_k$ and because we uniformly sample vertices in $\LS$, we will sample w.h.p.\ a vertex $v \in T_{f^\star}$ (see $\LS$) that is the first vertex of a path $P_e  \in B_k$  
that contains $\Omega(|B_k|)$ of the paths in $B_k$. In particular $|
P_e| = \Omega(|B_k|)$. By Lemma~\ref{lemma:long-shortcuts-dominating}, after $\Extend(M,D)$, $D$ will dominate the monotone path that corresponds to $P_e$, which is of length $|P_e| = \Omega(|B_k|)$. Because $B$ is composed of disjoint chains, it follows that w.h.p. the total shortcutting we perform to $P$ will be of size $\sum_{k=1}^q \Omega(|B_k|) = \Omega(|B|) = \Omega(|P|/\log n)$.

The proof of Lemma~\ref{lemma:long-shortcuts-shrinks} is based on the following structural definitions that formalize the paths $P_e$ in the above explanation.
These definitions allow us to measure how many applications of $\Concat$ and $\CO$ are needed in order to dominate a path $P_e$.

\begin{definition}
    Let $P=e_1 \ldots e_k$ be a path in $G^M$. For every $1\le i \le k$ we define
    \begin{itemize}
        \item $\bar{s}^P(i)\ge i$, the maximal index such that $e_i \ldots e_{\bar{s}^P(i)}$ is $e_i$-bounded.
        \item $\ubar{s}^P(i)\le i$, the smallest index such that $ e_{\ubar{s}^P(i)} \ldots e_i$ is $e_i$-bounded.
    \end{itemize}
    When $P$ is clear from the context, we abbreviate and write $\bar{s}(i),\ubar{s}(i)$.
\end{definition}

\begin{definition}\label{def:funnel-distance}
    Let $P= e_1 \ldots e_k$ be a path in $G^M$ and let $F_1,\ldots, F_t$ be the funnel decomposition of~$P$. For every $i$ we define
    \begin{itemize}
        \item $\bar{f}^P(i)=b-a+1$, where $a$ is maximal such that $e_i \in F_{a}$ and $b$ is minimal such that $e_{\bar{s}(i)}\in F_{b}$.
        \item $\ubar{f}^P(i)=a-b+1$, where  $a$ is minimal such that $e_i \in F_{a}$ and $b$ is maximal such that $e_{\ubar{s}(i)}\in F_{b}$.
    \end{itemize}
    When $P$ is clear from context, we abbreviate and write $\bar{f}(i),\ubar{f}(i)$.
\end{definition}

\begin{remark}
    Some arcs on a path might belong to two funnels (arcs that end/start a funnel). This is the reason Definition~\ref{def:funnel-distance} needs to specify a concrete funnel that contains $e_i,e_{\bar{s}(i)},e_{\ubar{s}(i)}$.
\end{remark}

The following lemma states that for every arc $e_i$ in a path $P=e_1 \ldots e_k$, if we extend the arc bounded path $e_i \ldots e_{\bar{s}(i)}$ by a single arc then we can extract from this path a monotone path, See Figure~\ref{figure:extend}(a),(d) and (c),(f).

\begin{lemma}\label{lemma:-bounded-to-monotone}
    Let $P = e_1 \ldots e_k$ be a path in $G^M$. Then for every  $1 < i < k$
    \begin{itemize}
        \item If $\bar{s}(i) < k$ then either $e_i \ldots e_{\bar{s}(i)+1}$ is monotone or $e_{i+1} \ldots v_{\bar{s}(i)+1}$ is monotone.
        \item If $\ubar{s}(i) > 1$ then either $e_{\bar{s}(i)-1}\ldots e_i$ is monotone or $e_{\bar{s}(i)-1}\ldots e_{i-1}$ is monotone.
    \end{itemize}
\end{lemma}

\begin{proof}
    We prove only the first claim, the second claim is symmetric.
    
    Assume $\bar{s}(i) < k$ and let $(u,v)=e_i$ and $(x,y) = e_{\bar{s}(i)+1}$. By the definition of $\bar{s}(i)$, it holds that $e_i \ldots e_{\bar{s}(i)}$ is $e_i$-bounded and $g_{y} = g(e_i \ldots e_{\bar{s}(i)+1})$ is either strictly larger than $\max( g_{u},g_{v})$ or strictly smaller than $\min( g_{u},g_{v})$. Thus, if $g_{y}> \max( g_{u},g_{v})$ then $y$ creates an ascending path with the vertex of minimum gain, either $u$ or $v$. Similarly, if $g_{y}< \min( g_{u},g_{v})$ then $y$ creates a descending path with the vertex of maximum gain, either $u$ or $v$.
\end{proof}

The following lemma is similar to Lemma~\ref{lemma:-bounded-to-monotone} and addresses the case in which a maximal arc bounded path in $P$ reaches the last arc of $P$, and $P$ is monotone with respect to a charge drop schedule, See Figure~\ref{fig:maximal_si}.

\begin{lemma}\label{lemma:-bounded-to-monotone_v2}
    Let $P = e_1 \ldots e_k$ be a monotone path in $G^M$.

    If $P$ is ascending, then for every $1<i<k$
    \begin{itemize}
        \item If $\bar{s}(i) = k$ then either $e_i \ldots e_{k}$ is ascending or $e_{i+1} \ldots e_{k}$ is ascending.
        \item If $\ubar{s}(i) = 1$ then either $e_{1}\ldots e_i$ is ascending or $e_{1} \ldots e_{i-1}$ is ascending.
    \end{itemize}
    
    If $P$ is descending  with respect to a charge drop schedule $C$, then for every $1<i<k$ 
    
    \begin{itemize}
        \item If $\bar{s}(i) = k$ then, with respect to an appropriate suffix of $C$, either $e_i \ldots e_{k}$ is descending or $e_{i+1} \ldots e_{k}$ is descending.
        \item If $\ubar{s}(i) = 1$ then, with respect to an appropriate prefix of  $C$, either $e_{1}\ldots e_i$ is descending or $e_{1} \ldots e_{i-1}$ is descending.
    \end{itemize}
\end{lemma}

\begin{proof}
    Assume $P$ is descending, the case of an ascending path is simpler since there is no charge drop schedule in play.
    Assume $\bar{s}(i) = k$, the case $\ubar{{s}}(i)=1$ is symmetric. 
    
    Let $(u,v)=e_i$ and $(x,y) = e_{k}$. By the definition of $\bar{s}(i)$, it holds that $P_i=e_i \ldots e_{k}$ is $e_i$-bounded (with respect to the zero schedule).
    Therefore, for every $w\in P_i$ we get $ g^{P}_w \le \max \{g^P_u, g^P_v \}$. Since $u,v$ are the first two vertices of $P_i$, we get for every $w\in P_i$ that $g^{P,C}_w \le \max \{g^{P,C}_u, g^{P,C}_v \}$.
    Since $P$ is descending with respect to $C$, then for every $w\in P_i$ it holds that $g^{P,C}_w \ge g^{P,C}_y$. Thus, if $g^{P,C}_u \ge g^{P,C}_v$ then $P_i=e_i \ldots e_k = uv \ldots y$ is descending with respect to a suffix of $C$ and otherwise $e_{i+1}\ldots e_k = v\ldots y$ is descending with respect to a suffix of $C$.

    The case $\ubar{s}(i)=1$ is symmetric.
\end{proof}

\begin{figure}
    \centering
\includegraphics[width=1.00\textwidth]{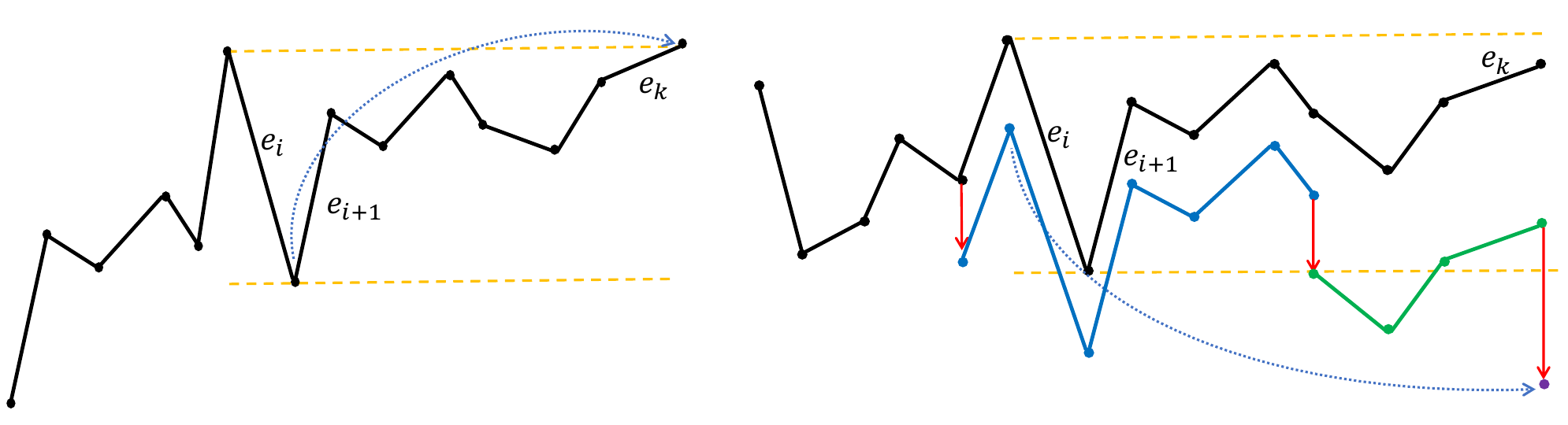}
\caption{Two illustrations of Lemma~\ref{lemma:-bounded-to-monotone_v2}. The dotted blue arrows point from the beginning to the end of the monotone suffixes. On the left we have an ascending path $e_1 \ldots e_k$, where $\bar{s}(i)=k$ and $e_i$ has negative gain. On the right we have a descending path $e_1\ldots e_k$ (the original path is in black) with respect to a charge drop schedule (indicated by the down vertical red arrows), where $\bar{s}(i)=k$ and $e_i$ has negative gain. The two symmetric cases in which $e_i$ is of positive gain are not shown.}
\label{fig:maximal_si}
\end{figure}

\begin{lemma}\label{lemma:funnel-bfs-length}
    Let $P=e_1 \ldots e_k = v_1\ldots v_{k+1}$ be a negative arc bounded path in $G^M$. Let $F_1, \ldots F_t$ be the funnel decomposition of $P$. The following holds w.h.p.\ after the main for-loop in $\LS(M,D)$.
    \begin{itemize}
        \item Assume $P$ is a $\overunderline{{v_1}{v_2}{v_{k+1}}}{1-1}{2-2}$-path in $G^M$. If $v_1$ is sampled to $T_j$, where $2^j\ge t$, then $D$ dominates $P$.
        \item Assume $P$ is a $\overunderline{{v_1}{v_{k}}{v_{k+1}}}{2-2}{3-3}$-path in $G^M$. If $v_{k+1}$ is sampled to $T_j$, where $2^j\ge t$, then $D$ dominates $P$.
    \end{itemize}
\end{lemma}

\begin{proof}
    Throughout the proof we use the procedures $\Concat$ and $\CO$ in order to concatenate funnels. Recall that we only use $\Concat$ on two negative arc bounded paths and we only use $\CO$ on a negative arc bounded path and a positive arc bounded path.  Note that every other arc on a funnel has negative gain.
    
    We prove only the first case since the second case is symmetric. For $b=1\ldots t$, denote $F_b = e_{\ell_{b}}\ldots e_{r_{b}}$. We prove by induction on $b=1\ldots t$ that after iteration $b-1$ (among the $2^j$) of applying the procedures $\Concat(M,D,T_j,V,V)$ and $\CO(M,D,T_j,V,V)$, $D$ dominates $P_b = e_1 \ldots e_{r_{b}}$.

    The base case $b=1$ is immediate by Lemma~\ref{lemma:concat-sampling-correctness} which states that after executing $\ComputeF(M)$,   $D$ dominates each of $F_1,\ldots, F_t$  w.h.p.\ ($\LS$ starts by running $\ComputeF(M)$). Therefore $D$ dominates $F_1$ even before the first iteration. Since a subpath of a funnel is also a funnel (see Lemma~\ref{lemma:funnel-zigzag-structure}), it follows that w.h.p.\ $D$ also dominates all of the subpaths of each of the funnels $F_1,\ldots, F_t$.

    Assume that after the first $b-1$ iterations, $D$ dominates $P_b=e_1 \ldots e_{r_b}$. Consider the next funnel $F_{b+1} = e_{\ell_{b+1}}\ldots e_{r_{b+1}}$ 
    and the $b$'th iteration. We split to the following cases.

    \textbf{Case $F_{b+1}$ is $e_{r_{b+1}}$-bounded:} If $e_{r_{b+1}}$ has nonnegative gain, we apply Lemma~\ref{lemma:concat-opposite-bounded} on $P_b$ and the funnel $F_{b+1} \setminus P_b$: After performing $\CO(M,D,T_j,V,V)$ in iteration $b$ in the inner loop of $\LS$, $D$ dominates $P_{b+1} = e_1 \ldots e_{r_{b+1}}$. If $e_{r_{b+1}}$ has negative gain,  then by Lemma~\ref{lemma:funnel-zigzag-structure}, $e_{r_{b+1}-1}$ has positive gain and $F'_{b+1} = e_{\ell_{b+1}} \ldots e_{r_{b+1}-1}$ is a funnel which is $e_{r_{b+1}-1}$-bounded. Therefore, by Lemma~\ref{lemma:concat-opposite-bounded}, after performing $\CO(M,D,T_j,V,V)$, $D$ dominates $P'_{b+1} \coloneq P_b \mid (F'_{b+1}\setminus P_b) = e_1 \ldots e_{r_{b+1}-1}$. By Lemma~\ref{lemma:concat-basic}, after performing $\Concat(M,D,T_j,V,V)$ (i.e., concatenating $P'_{b+1}$ with the single negative gain arc $e_{r_{b+1}}$), $D$ dominates $P_{b+1}=  P'_{b+1} \mid e_{r_{b+1}}$.
    
    \textbf{Case $F_{b+1}$ is $e_{\ell_{b+1}}$-bounded:} Consider the funnel $F'_{b+1} = F_{b+1}\setminus P_b$ and let $s\ge \ell_{b+1}$ be the index such that $F'_{b+1} = e_s \ldots e_{r_{b+1}}$. By Lemma~\ref{lemma:funnel-zigzag-structure}, since $F_{b+1}$ is first arc bounded then so is $F'_{b+1}$.
    If $e_{s}$ has negative gain, then by Lemma~\ref{lemma:concat-basic}, after performing $\Concat(M,D,T_j,V,V)$, $D$ dominates $P_{b+1}= P_b \mid F'_{b+1}$.
    If $e_{s}$ has nonnegative gain, then by Lemma~\ref{lemma:concat-opposite-bounded}, after performing $\CO(M,D,T_j,V,V)$, $D$ dominates $Q = P_b \mid e_s$.
    By Lemma~\ref{lemma:funnel-zigzag-structure}, $e_{s+1}$ has negative gain and therefore $F''_{b+1} = F'_{b+1}\setminus \{e_s\} = e_{s+1} \ldots e_{r_{b+1}}$ is negative arc bounded. Therefore, by Lemma~\ref{lemma:concat-basic}, after performing $\Concat(M,D,T_j,V,V)$, $D$ dominates $P_{b+1} = Q \mid F''_{b+1}$.
\end{proof}

The following is a corollary of Lemma~\ref{lemma:long-shortcuts-dominating} and Lemmas \ref{lemma:-bounded-to-monotone}, \ref{lemma:-bounded-to-monotone_v2}, \ref{lemma:funnel-bfs-length}.

\begin{corollary}\label{corollary:long-shortcuts}
    Let $P= e_1 \ldots e_k$ be a monotone path in $G^M$ which is either descending with respect to a charge drop schedule $C$ or ascending with respect to the zero schedule. Let $(u,v)=e_i\in P$ be a negative gain arc and let $\bar{P}_i,\ubar{P}_i$ be the monotone\footnote{These paths may have a charge drop schedule assigned to them.} paths corresponding to $e_{\bar{s}(i)},e_{\ubar{s}(i)}$, respectively, given by Lemmas~\ref{lemma:-bounded-to-monotone} and~\ref{lemma:-bounded-to-monotone_v2}.\footnote{Monotone paths given by Lemma~\ref{lemma:-bounded-to-monotone} start either at $u$ or at $v$ and end at $e_{\bar{s}(i)+1}$. Monotone paths given by Lemma~\ref{lemma:-bounded-to-monotone_v2} start either at $u$ or at $v$ and end at $e_k$. Similarly monotone paths given by Lemma~\ref{lemma:-bounded-to-monotone} end either at $u$ or at $v$ and start at $e_{\ubar{s}(i)-1}$. Monotone paths given by Lemma~\ref{lemma:-bounded-to-monotone_v2} end either at $u$ or at $v$ and start at $e_1$.} The following holds at the end of $\LS(M)$.
    \begin{itemize}
        \item Assume $\bar{P}_i$ starts with $e_i$. If $u$ is sampled into $T_j$ and $2^j\ge \bar{f}(i)$, then $D$ dominates $\bar{P}_i$.
        \item Assume $\bar{P}_i$ starts with $e_{i+1}$. If $u$ is sampled into $T_j$ and $2^j\ge \bar{f}(i)$, then $D$ dominates $\bar{P}_i$.
        \item Assume $\ubar{P}_i$ ends with $e_i$. If $v$ is sampled into $T_j$ and $2^j\ge \ubar{f}(i)$, then $D$ dominates $\ubar{P}_i$.
        \item Assume $\ubar{P}_i$ ends with $e_{i-1}$. If $v$ is sampled into $T_j$ and $2^j\ge \ubar{f}(i)$, then $D$ dominates $\ubar{P}_i$.
    \end{itemize}
    The domination is with respect to a sub-schedule of $C$.
\end{corollary}

\begin{proof}
    We prove only the first case, the other cases are simpler. Assume $\bar{P}_i$ starts with $e_i$ and $u \in T_j$. Since $e_i$ has negative gain, it follows that $\bar{P}_i$ is descending with respect to a charge drop schedule $C'$. Let $P_i = e_i \ldots e_{\bar{s}(i)}$. By Lemma~\ref{lemma:funnel-bfs-length}, after the for loop in $\LS(M)$, $D$ dominates $P_i$. Let $b$ be the index such that $\bar{P}_i=e_i \ldots e_b$.  By Lemmas~\ref{lemma:-bounded-to-monotone} and~\ref{lemma:-bounded-to-monotone_v2}, either $b= \bar{s}(i)+1$ or $b= \bar{s}(i)=k$. We split into the following cases.

    \textbf{Case $b= \bar{s}(i)+1$:} By Lemma~\ref{lemma:-bounded-to-monotone}, $\bar{P}_i = P_i \mid e_{\bar{s}(i)+1}$ is monotone and $C'$ is respect to the zero schedule. Therefore, by Lemma~\ref{lemma:long-shortcuts-dominating}, after $\Extend(M,D,T)$, $D$ dominates $\bar{P}_i$.

    \textbf{Case $b= \bar{s}(i)=k$:} Therefore $\bar{P}_i = P_i$ is descending with respect to $C'$. Note that $g^{C'}(\bar{P}_i) \le g^{\bar{P}_i,C'}_v  \le g^{\bar{P}_i}_v = M[u][v]$. Consider the vertex representation of $P_i$ and let $r$ be the index such that $P_i = v_1 v_2 \ldots v_r$, where $v_1=u$ and $v_2=v$. Since $D$ dominates $P_i$, it follows by Lemma~\ref{lemma:long-shortcuts-dominating} that following the application of $\Extend(M,D,T)$ we have that $D[v_1][v_r]\ge M[v_1][v_2]=M[u][v]\ge g^{C'}(\bar{P}_i)$. Thus, $D$ dominates $\bar{P}_i$ with respect to $C'$.
     
\end{proof}

The following lemma proves that for every path $P=e_1\ldots e_k$, the set of paths $\{e_i \ldots e_{\bar{s}(i)}\mid 1\le i \le k\}$ is laminar and similarly $\{  e_{\ubar{s}(i) \ldots e_i}\mid 1\le i \le k\}$ is laminar.

\begin{lemma}\label{lemma:laminar}
    Let $P=e_1 \ldots e_k$ be a path in $G^M$, then the sets of intervals $\{(i, \bar{s}(i)) \mid 1\le i \le k \}$ and $\{(\ubar{s}(i), i) \mid 1\le i \le k \}$ are laminar.
\end{lemma}

\begin{proof}
    We prove the claim only for the first set, the other set is symmetric. Let $1\le i \le k$ and let $j \in (i, \bar{s}(i))$. We show $(j, \bar{s}(j)) \subseteq (i, \bar{s}(i))$ from which the lemma follows. Denote $e_i=(u,v)$ and $e_j = (x,y)$. Since $P_i=e_i \ldots v_{\bar{s}(i)}$ is $e_i$-bounded, we have $g_{w} \in [\min\{ g_u,g_v\}, \max\{ g_u,g_v\}]$ for every $w\in P_i$. In particular $[\min\{g_{x}, g_{y}\}, \max \{ g_{x}, g_{y} \} ] \subseteq [\min\{ g_u,g_v\}, \max\{ g_u,g_v\}]$.  
    
    Since $e_j \ldots e_{\bar{s}(j)}$ is $e_j$-bounded we get that $g_{w} \in [\min\{g_{x}, g_{y}\}, \max \{ g_{x}, g_{y} \} ] \subseteq [\min\{ g_u,g_v\}, \max\{ g_u,g_v\}]$ for every $w\in P_j= e_j \ldots e_{\bar{s}(j)}$. Therefore $e_i\ldots e_{\bar{s}(j)}$ is $e_i$-bounded, so by the maximality of $\bar{s}(i)$ we get that $\bar{s}(i) \ge \bar{s}(j)$, and therefore $(j,\bar{s}(j)) \subseteq (i,\bar{s}(i))$.
\end{proof}

The following lemma easily derives Lemma~\ref{lemma:paths-shrink}. This lemma is our main theoretical contribution and the key to our result. 

\begin{lemma}\label{lemma:long-shortcuts-shrinks}
    Let $P = e_1 \ldots e_k$ be a monotone simple path in $G^M$ with respect to a charge drop schedule $C$, from $s$ to $t$. Let $F_1,\ldots, F_t$ be the funnel decomposition of $P$. Let $\bar{M}$ be the shortcuts table returned from $\LS(M)$. If $t \le k / n^\alpha $ and $k$ is polynomial in $n=|V|$, then w.h.p. there is a monotone path $P'$ in $G^{\bar{M}}$, with respect to a charge drop schedule $C'$, from $s$ to $t$ in $G^{\bar{M}}$ that satisfies $g^{G^{\bar{M}}}(P')\ge g^{G^M}(P)$ and $|P'| \le (1-1/\Omega(\log n))\cdot |P|$. 
\end{lemma}

\begin{proof}
     By the statement of the lemma, $n^\alpha \le k \le n$ and therefore $\log k = \Theta(\log n)$. Consider $F_1,\ldots F_t$, we distinguish between funnels that are first-arc bounded to those which are last-arc bounded. Assume that the majority of the arcs of $P$ belong to first-arc bounded funnels. The analysis for the other case is symmetric. Among these funnels (first-arc bounded), we consider only funnels of length at least $n^\alpha/4$. Note that at least $k/4$ arcs belong to such funnels (if more than $k/4$ arcs belong to funnels of length at most $n^\alpha/4$ then $t > k / n^\alpha $, a contradiction). Among these arcs, we take only those of negative gain. Since every other arc in a funnel is of negative gain (Lemma~\ref{lemma:funnel-zigzag-structure}), we are left with at least $k/10$ arcs.\footnote{The choice of $10$ was arbitrary. If $n^\alpha >> 1$ the number of negative gain arcs in funnel is very close to half of the length of the funnel.} Denote these arcs by $e_{i_1},\ldots e_{i_r}$.
    
    By Lemma~\ref{lemma:laminar}, the set $A=\{(i_j, \bar{s}(i_j)) \mid 1\le j \le r \}$ is laminar. We refer to each item in $A$ as an \emph{interval}. For $i=1,\ldots, \log k$, let $A_i = \{(i_j,\bar{s}(i_j)) \mid \bar{f}(i_j) \in  [2^i, 2^{i+1}) \}\subseteq A $, see Definition~\ref{def:funnel-distance}.  Observe that for every $1\le i \le k$, $A_{i}$ is laminar as a subset of $A$. Moreover, each interval in $A_{i}$ cannot contain two disjoint intervals in $A_i$. Indeed,  assume $(i_{j_1},\bar{s}(i_{j_1})), (i_{j_2},\bar{s}(i_{j_2})) \subseteq (i_{j_3},\bar{s}(i_{j_3}))$ and $(i_{j_1},\bar{s}(i_{j_1})) \cap (i_{j_2},\bar{s}(i_{j_2})) = \emptyset$, where all intervals belong to $A_{i}$. Therefore $\bar{f}(i_{j_3}) \ge \bar{f}(i_{j_1}) + \bar{f}(i_{j_2}) \ge 2^{i} + 2^{i} = 2^{i+1}$, so $(i_{j_3},\bar{s}(i_{j_3})) \notin A_{i}$, a contradiction. It follows that we can decompose $A_i$ into a collection of \emph{chains}. Each chain is a maximal subset of nested intervals in $A_i$.

    Let $i^\star$ be such that $|A_{i^\star}| \ge |A_i|$ for every $1\le i \le \log k$. Thus, $|A_{i^\star}|\ge \frac{k}{10 \log{k}}$. Let $B_1,\ldots, B_q$ be the decomposition of  $A_{i^\star} $ into chains. We have that $A_{i^\star}= \cup_{i=1}^{q} B_i$. Since the $B_i$'s are disjoint, $q\cdot 2^{i^\star} \le t$.  Let $A'_{i^\star}$ be the union of the~$B_i$'s that satisfy $|B_i| \ge \frac{k}{20q \log  k}$. It follows that

    \begin{align}\label{eq:main_lemma}
        |A'_{i^\star}| 
    \ge  |A_{i^\star}| - q \cdot \frac{k}{20q \log  k}
    \ge \frac{k}{20 \log 
  k}.
    \end{align}

    Let $B_j \subseteq A'_{i^\star}$. 
    We have that
    \begin{align*}
        |B_j| \ge \frac{k}{20q \log k} \numge{1} \frac{k \cdot 2^{i^\star}}{20t \log k} \numge{2} \frac{n^\alpha 2^{i^\star}}{20 \log k},
    \end{align*}
    where $(1)$ follows since $q\cdot 2^{i^\star}\le t$ and $(2)$ follows since $t \le k / n^\alpha$.
    Since $\LS(M)$ samples vertices to $T_{i^\star}$ i.i.d.\ with probability $\Theta(\frac{\log^2 n}{2^{i^\star}n^\alpha})$ and $k \ge t \cdot n^\alpha = \Omega(n^\alpha)$, it follows by the Chernoff bound that $T_{i^\star}$ contains $\Omega(\log k) = \Omega(\log n)$  vertices $u\in V$, where $e_{i_a}=(u,v)$ and $(i_a , \bar{s}(i_a))\in B_j$. Furthermore, w.h.p.\ $T_{i^\star}$ contains a vertex $u$, incident to an arc $e_{i_a}=(u,v)$, for some index $i_a$, such that $(i_a, \bar{s}(i_a))$ is among the $0.5|B_j|$ longest intervals in $B_j$. Fix such a vertex $u_j$ and the corresponding index $i_{a_j}$ for every chain $B_j\subseteq A'_{i^\star}$.

    Let $q'$ be the number of chains in $A'_{i^\star}$. Let $P_{a_1},\ldots, P_{a_{q'}}$,  be the monotone paths that correspond to $(i_{a_j}, \bar{s}(i_{a_j}))$, for $j=1,\ldots, q'$, by Lemmas~\ref{lemma:-bounded-to-monotone} and~\ref{lemma:-bounded-to-monotone_v2}.
    Notice that since $(i_{a_j}, \bar{s}(i_{a_j}))$ is among the $0.5|B_j|$ longest intervals in $B_j$, it follows that $|P_{a_j}|\ge 0.5 |B_j|$.
    By Corollary~\ref{corollary:long-shortcuts}, $\bar{M}$ ($\bar{M}$ is defined in the statement of the lemma) dominates $P_{a_j}$, for $j=1,\ldots,q'$. Since each $B_j$ is a maximal chain, the intervals $(i_{a_j}, \bar{s}(i_{a_j}))$, $j=1,\ldots,q'$, are pairwise disjoint so it follows that $P_{a_1},\ldots, P_{a_{q'}}$ are also disjoint. Therefore, if we replace each $P_{a_j}$ by the corresponding shortcut in $G^{\bar{M}}$, we get a path $P'$ in $G^{\bar{M}}$ of length
    \begin{align*}
       |P'| &\le k - \sum_{j=1}^{q'} {|P_{a_i}|} \le k - \sum_{j=1}^{q'} {0.5 |B_j|} =
       k - 0.5 |A'_{i^\star}| \\
       &\numle{1} k -0.5 \frac{k}{20\log k} \numeq{2} 
       \left(1 - \Omega\left(\frac{1}{\log n} \right) \right) \cdot k = 
       \left(1 - \Omega\left(\frac{1}{\log n} \right) \right) \cdot |P|,
    \end{align*}
    where inequality $(1)$ follows from Equation~(\ref{eq:main_lemma}) and equality $(2)$ follows since $k = O(poly(n))$.

    We are left to prove that $P'$ is monotone with respect to some charge drop schedule. If $P$ is ascending then it is clear. Assume $P$ is descending with respect to $C$ and denote $P' = v_1\ldots v_k$. We 
    claim that there is a charge drop schedule $C'$ such that $g^{P',C'}_{v_i} = g^{P,C}_{v_i}$, for every $i=1,\ldots,k$.\footnote{There is vagueness when writing  $g^{P,C}_{v_i}$ since $P$ is not necessarily simple. We refer to the appropriate copy of $v_i$ according to the shortcutting performed on $P$} This claim holds since $\bar{M}\ge M$ coordinate-wise and since $\bar{M}$ dominates all monotone paths $P_{a_j}$, for $j=1,\ldots, q'$. 
\end{proof}

We are ready to prove Lemma~\ref{lemma:paths-shrink}.

\begin{proof}[Proof of Lemma~\ref{lemma:paths-shrink}]
    Let $r= n^\alpha$ and let $M_0(=M),M_1,\ldots, M_r$ be the shortcuts tables throughout the $r$ iterations of $\UpdateS$. Let $(P_0,C_0)(=(P,C)),(P_1,C_1),\ldots, (P_r,C_r)$ be a series of monotone paths, where $P_i$ is the shortest path in $G^{M_i}$ from $v_1$ to $v_k$ that has no smaller gain (with respect to $G^{M_i}$ and $C_i$) than $P_{i-1}$ (with respect to $G^{M_{i-1}}$ and $C_{i-1}$). These paths are guaranteed to exist by the definition of the algorithm. We split the proof into cases.
    
     \textbf{Case $|P|\le r$:} Since we make $r$ rounds of $\Simple$, we get by Lemma~\ref{lemma:mono-has-shortcut} that, for every $1\le i < r$, if $|P_i|>1$ then $|P_{i+1}| < |P_i|$. Thus, $|P_r|=1$ and the lemma follows.

    \textbf{Case $|P| > r$:} 
    If $P_r \le |P|/2$, then we are done. Otherwise $P_r > |P|/2$ and therefore for at least $r/2$ indices $0\le i< r$, it holds that $|P_i|-|P_{i+1}| \le |P|/r$. This mean that, for each such index $i$, $P_i$ has at most $|P|/r$ disjoint short shortcuts as subpaths. 
    Since at the end of a maximal funnel there is a short shortcut, it follows that $P_i$ has $O(|P|/r)$ maximal funnels in its funnel decomposition. Therefore, w.h.p.\, we run $\LS(M_i)$ at an iteration $i$ such that $|P_i|-|P_{i+1}| \le |P|/r$ and  $P_i$ has $O(|P_i|/r) = O(|P|/n^\alpha)$ funnels in its funnel decomposition. Hence, the conditions of Lemma~\ref{lemma:long-shortcuts-shrinks} are satisfied and we are done.
    
\end{proof}

\subsection{Running Time}

\begin{lemma}\label{running-time-funnels}
    Procedure $\ComputeF(M)$ terminates in expected $\tilde{\Theta}(n^{10/3})$ time.
\end{lemma}

\begin{proof}
     Denote by $T_{Funnel}, T_{BFS}$ the expected running times of $\ComputeF(M)$ and $\BFS(M,D)$, respectively. 
     Let $T_{Concat}(u,w,x)$ be the running time of $\Concat(M,D,U,W,X)$, where $|U|=u, |W|=w, |X|=x$.
    
    Clearly $T_{BFS} = \tilde{\Theta}(n^{3})$ and $T_{Concat}(u,w,x) = \tilde{\Theta}(n^3 + uwx\cdot n)$. Therefore, 
    \begin{align*}
        T_{Funnel} = n^{1-\beta}\cdot T_{BFS} + \tilde{\Theta}\left(T_{Concat}\left( n^{\beta}, n^{\beta}, n \right)\right) = \tilde{\Theta}(n^{4-\beta}) + \tilde{\Theta}(n^3 + n^{2+2\beta}).
    \end{align*}
    Therefore, by setting $\beta = 2/3$, we get $T_{Funnel} = \tilde{\Theta}(n^{10/3})$.
\end{proof}

\begin{lemma}
    Procedure $\ComputeS(G)$ terminates in expected $\tilde{\Theta}(n^{3.5})$ time.
\end{lemma}

\begin{proof}
    Denote by $T_{Short}, T_{Long},T_{Funnel}$ the expected running times of $\Simple(M)$, $\LS(M)$, $\ComputeF(M)$, respectively. 
    
    Let $T_{Concat}(u,w,x) = \tilde{\Theta}(n^3 + uwx\cdot n)$ and note that this is the running time of $\Concat(M,D,U,W,X)$ and $\CO(M,D,U,W,X)$, where $|U|=u, |W|=w, |X|=x$. Let $T_{Bounded}(t) = \Theta(t \cdot n^3)$ be the running time of $\Extend(M,D,T)$, where $|T|=t$.
    
    Clearly $T_{Short} = \tilde{\Theta}(n^{3})$. By Lemma~\ref{running-time-funnels}, it holds that $T_{Funnel} = \tilde{\Theta}(n^{10/3})$.

    We now analyze the expected running time of $\LS(M)$. 
    Consider the \emph{For} loop in $\LS(M)$. For every $i=1,\ldots ,O(\log n)$, the expected size of $T_i$ is $\tilde{\Theta}(\kappa/2^i)$, where $\kappa=n^{1-\alpha}$. Therefore, the expected size of $T$ (the union of all the sets $T_i$ throughout the iterations) is $\tilde{\Theta}(\kappa)$. We get that
    \begin{align*}
        T_{Long} = T_{Funnel} + \sum_{i=1}^{\log n} {2^i \cdot T_{Concat}\left(\frac{\kappa}{2^i},n,n\right)} + T_{Bounded}(\kappa) = \tilde{\Theta}(n^{10/3}) +  \tilde{\Theta}(\kappa n^{3})  + \tilde{\Theta}(\kappa n^{3}) = \tilde{\Theta}(n^{10/3}+n^{4-\alpha}).
    \end{align*}

    Finally, the expected running time of $\ComputeS(G)$ is $n^\alpha \cdot T_{Short} + \tilde{\Theta}(1)\cdot T_{Long} = \tilde{\Theta}(n^{3+\alpha}) + \tilde{\Theta}(n^{10/3}+n^{4-\alpha})$. Therefore, by setting $\alpha = 0.5$, we get that the expected running time of $\ComputeS(G)$ is $\tilde{\Theta}(n^{3.5})$.
\end{proof}

\section{Relating \texorpdfstring{$M$ and $D$}{M and D} to \texorpdfstring{$G$}{G}}\label{S-relating}
In Theorem~\ref{theorem:shortcut} we have seen that every monotone simple path in $G$ is dominated w.h.p.\ by the final shortcuts table $M$ returned by $\ComputeS$. Moreover, by Invariant~\ref{invariant} we know that every value in $D$ is realizable by a traversable path in $G^M$.

The following lemma gives the relation between $G^M$ and $G$. 
The lemma states that any traversable path in $G^M$ can be ``unwrapped" to a traversable path in $G$ that has ``better" $\alpha$ (maximum final charge) values.

\begin{lemma}\label{lem:GM_to_G}
    Let $M$ be the shortcut table return by $\ComputeS$. Let $P= v_1\ldots v_k$ be a traversable path in $G^M$ and let $C$ be a charge drop schedule for $P$. There exists a traversable path $P' = P^{v_1 v_2} \mid P^{v_2 v_3}\mid \ldots \mid P^{v_{k-1}v_k}$ in $G$ and a charge drop schedule $C' = C^{v_1 v_2} \mid C^{v_2 v_3}\mid \ldots \mid C^{v_{k-1}v_k}$ such that

    \begin{enumerate}[label=(\alph*)]
        \item $P^{v_i v_{i+1}}$ is a monotone path from $v_i$ to $v_{i+1}$ in $G$ with respect to the charge drop schedule  $C^{v_i v_{i+1}}$, for every $1\le i<k$. In particular, if $P$ is of length $1$ then $P'$ is monotone with respect to $C'$.
        \item $g^{P,C}_{v_i} = g^{P',C'}_{v_i}$, for every $1\le i\le k$.
        \item  $\alpha^{G}_b(P') \ge \alpha^{G^M}_b(P)$ for every $b\in [0,B]$.
    \end{enumerate}
\end{lemma}

\begin{proof}
    Let $M_1$ be the adjacency matrix of $G$. Let $M_i$ for $i\ge 2$ be the shortcuts table computed by the $i-1$'th iteration of $\ComputeS$ and let $M_t = M$, where $t$ is the number of iterations of $\ComputeS$. For every $i=1,\ldots,t-1$, let  $D_i$ be the data structures that we used to generate $M_{i+1}$. We prove by induction on $i$ that the lemma holds in $G^{M_i}$ for every $i=1,\ldots ,t$. The base case $i=1$ follows since $G^{M_1} = G$. Let $i>1$ and let $P= v_1\ldots v_k$ be a traversable path in $G^{M_i}$. By definition, for every $s,t\in V$ it holds that $M_i[s][t] = D_{i-1}[s][t]$. Moreover, by invariant~\ref{invariant}\ref{1c}, there is a monotone path $P^{st}$ in $G^{M_{i-1}}$ with respect to a charge drop schedule $C^{st}$ such that $g^{C^{st}}(P^{st}) = D_{i-1}[s][t] =  M_i[s][t]$. Let $P' = P^{v_1 v_2} \mid P^{v_2 v_3}\mid \ldots \mid P^{v_{k-1}v_k}$ and let $C' = C^{v_1 v_2} \mid C^{v_2 v_3}\mid \ldots \mid C^{v_{k-1}v_k}$. It follows that $g^{P,C}_{v_j} = g^{P',C'}_{v_j}$, for every $1\le j\le k$.
    Since $P$ is traversable and by Lemma~\ref{lemma:alpha-of-monotone}, it follows that $P'$ is traversable (in $G^{M_{i-1}}$) and satisfies $\alpha^{G^{M_{i-1}}}_b(P') \ge \alpha^{G^{M_i}}_b(P)$ for every $b\in [0,B]$. The inductive step follows by applying the inductive assumption to $P'$ and $C' = C^{v_1 v_2} \mid C^{v_2 v_3}\mid \ldots \mid C^{v_{k-1}v_k}$.
    \end{proof}

We get as a corollary the following structural lemma about paths realizing the values in $D$.

\begin{corollary}\label{cor:M-G-connection}
    Let $M$ be a shortcuts table and let $D$ be a data structure that maintains Invariant~\ref{invariant} with respect to $G^M$. The following holds for every $x,y,z\in V$. 
    \begin{enumerate}
        \item \label{cor:1}
        Assume $D[xy][z]\neq -\infty$. Then there exists a traversable path $P = P^{xy}\mid P^{yz}$ in $G$ and a charge drop schedule $C = C^{xy}\mid C^{yz}$ such that\footnote{The paths $P^{xy}$ and $P^{yz}$ are paths from $x$ to $y$ and from $y$ to $z$, respectively. We use the same convention also for claims $2$ and $3$.}
        \begin{enumerate}[label=(\alph*)]
            \item \label{cor:1a}  $g^C(P) = D[xy][z]$,
            \item \label{cor:1b} $P^{xy}$ is monotone with respect to $C^{xy}$ and $M[x][y] = g^{C^{xy}}(P^{xy})$,
            \item \label{cor:1c} The gains of the first and last vertices of $P^{xy}$ (i.e. $x$ and $y$) bound the gains of all other vertices in $P$. All gains are with respect to $C$.
        \end{enumerate}
        \item \label{cor:2}
        Assume $D[x][yz]\neq -\infty$. Then there exists a traversable path $P = P^{xy}\mid P^{yz}$ in $G$ and a charge drop schedule $C = C^{xy}\mid C^{yz}$ such that
        \begin{enumerate}[label=(\alph*)]
            \item \label{cor:2a} $g^C(P) = D[x][yz]$,
            \item \label{cor:2b} $P^{yz}$ is monotone with respect to $C^{yz}$ and $M[y][z] = g^{C^{yz}}(P^{yz})$,
            \item \label{cor:2c} The gains of the first and last vertices of $P^{yz}$ (i.e. $y$ and $z$) bound the gains of all other vertices in $P$. All gains are with respect to $C$.
        \end{enumerate}
        \item \label{cor:3}
         Assume $D[x][y]\neq -\infty$. Then there exists a traversable path $P $ in $G$ and a charge drop schedule $C$ such that
        \begin{enumerate}[label=(\alph*)]
            \item \label{cor:3a} $g^C(P) = D[x][y]$,
            \item \label{cor:3b} $P$ is monotone with respect to $C$.
        \end{enumerate}
    \end{enumerate}
\end{corollary}

\begin{proof}
    We prove only the first claim, the other claims are similar. Assume $D[xy][z]\neq -\infty$ and assume w.l.o.g.\ $M[x][y] >0$. Let $P=v_1 \ldots v_k$ and $C$ be the path in $G^M$ and charge drop schedule that realize $D[xy][z]$ by Invariant~\ref{invariant}\ref{1a}. Thus, $g^C(P)=D[xy][z]$. 
    Let $P' = P^{v_1 v_2} \mid P^{v_2 v_3}\mid \ldots \mid P^{v_{k-1}v_k}$ and $C' = C^{v_1 v_2} \mid C^{v_2 v_3}\mid \ldots \mid C^{v_{k-1}v_k}$ be the path in $G$ and charge drop schedule realizing $P$  by Lemma~\ref{lem:GM_to_G}.
    Thus, $g^{C'}(P')=g^C(P)=D[xy][z]$, proving claim~\ref{cor:1}\ref{cor:1a}.
    By Lemma~\ref{lem:GM_to_G}, $P^{v_1 v_2} = P^{xy}$ is monotone with respect to $C^{v_1 v_2} = C^{xy}$ and $g^{C^{xy}}(P^{xy}) = M[x][y]$, proving claim~\ref{cor:1}\ref{cor:1b}.
    By Lemma~\ref{lem:GM_to_G}, we get that $g^{P,C}_{v_i} = g^{P',C'}_{v_i}$ for every $1\le i\le k$. Since $P$ is first-arc bounded with respect to $C$, we get that $g^{P',C'}_{v_1} \le g^{P',C'}_{v_i} \le g^{P',C'}_{v_2}$ for every $1\le i\le k$. 

    We now prove claim~\ref{cor:1}\ref{cor:1c}. Let $v\in P'$ and $1\le i<k$ be such that $v\in P^{v_i v_{i+1}}$. Since $P^{v_i v_{i+1}}$ is monotone with respect to $C^{v_i v_{i+1}}$, we get that 
    \begin{align*}
        g^{P',C'}_{v_1} \le g^{P',C'}_{v_i} \le 
        g^{P',C'}_{v} \le
        g^{P',C'}_{v_{i+1}} \le
        g^{P',C'}_{v_2}.
    \end{align*}
\end{proof}

\section{Stage II - Computing the \texorpdfstring{$\alpha$}{alpha} values}\label{S-algorithm}
Let $M$ be the shortcuts table we receive from Stage I and let $D = \ComputeF(M)$.

In this appendix, using $M$ and the data structure $D$, we show how to compute $\alpha_B(s,t)$ for every $s,t\in V$. Recall that $\alpha_B(s,t)$ is the maximum final charge at $t$ when the car starts at $s$ with a full battery. The algorithms proceeds in two steps.

In the first step we build a graph $H = (V^0 \cup V^B, E(H))$, where $V^b = \{v^b \mid b\in \{0,B\} \}$ represents that we are at $v$ with at least $b$ charge.  An arc $u^{b_1} v^{b_2}\in E(H)$ represents that $\alpha_{b_1}(u,v) \ge b_2$.\footnote{Note that the other direction does not necessarily hold: It is possible that $\alpha_{b_1}(u,v) \ge b_2$ but $u^{b_1} v^{b_2}\notin E(H)$.} We create the arcs $E(H) \subseteq \{u^{b_1} v^{b_2} \mid  \alpha_{b_1}(u,v) \ge b_2\}$ by observing simple properties of the values in $D$. Finally we 
compute the 
transitive closure $H^\star$ of $H$. We claim in Theorem~\ref{theorem:B-B-strong-full} that w.h.p., for every $s,t\in V$, $\alpha_B(s,t)=B$ if and only if $s^B t^B \in E(H^\star)$.

The second (and final) step is based on combining the following observations. Let $s,t \in V$ and let $P= v_1 (=s) \ldots v_k(=t)$ be an optimal path  from $s$ to $t$ (i.e., $\alpha_B(s,t)=\alpha_B(P)$). If $g_{v_i} < 0$ for every $i\le k$ then we can assume that $P$ is simple (otherwise it contains a positive gain cycle and we can repeat this cycle to improve final charge) and we show in Lemma~\ref{lemma:alpha-simple-B-0} that $\alpha_B(s,t)$ is realized by a funnel in $G^M$. Otherwise, some vertices in $P$ are visited with full charge. Using $H^\star$ from the first step (Theorem~\ref{theorem:B-B-strong-full}), we can find the last vertex $y\in P$ that is reached with full charge and compute $\alpha_B(y,t)$. We claim that the suffix $P^{yt}$ of $P$ from $y$ to $t$ is simple, so (by Lemma~\ref{lemma:alpha-simple-B-0})  $\alpha_B(y,t)$ is realized by a funnel in $G^M$. Thus, the second step amounts to finding pairs $(y,t)$ such that $s^B y^B \in H^\star$ and $\alpha_B(y,t)$ can be realized by a funnel. We use the best such pairs in order to compute $\alpha_B(s,t)$ for every $s,t\in V$.   

The rest of this section is organized as follows. In Appendix~\ref{section:transitive-graph} we build the transitive closure graph~$H^\star$ and prove basic properties of $H^\star$. In Appendix~\ref{section:transitive-correctness} we prove that $H^\star$ indeed finds all $s,t\in V$ such that $\alpha_B(s,t)=B$. Finally, in Appendix~\ref{section:compute-alpha} we complete the computation of $\alpha_B(\cdot,\cdot)$ and prove its correctness as described above. 

\subsection{The transitive closure graph}\label{section:transitive-graph}

\begin{figure}
    \centering
\includegraphics[width=1.00\textwidth]{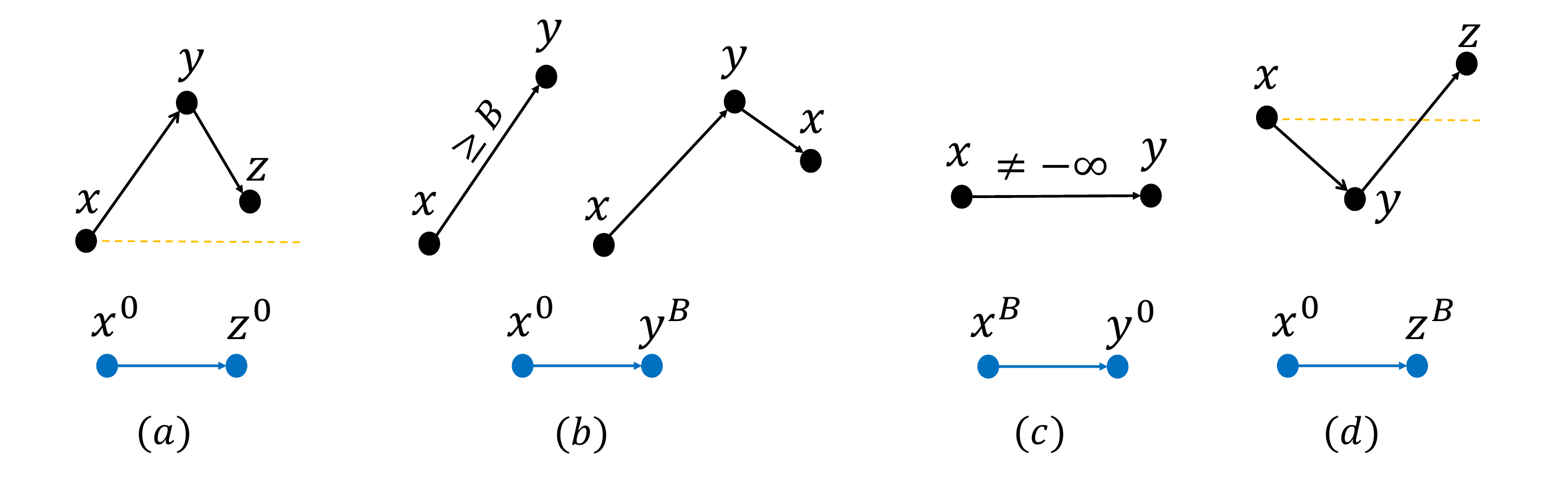}
\caption{The $4$ types of edges we include in $H$.}
\label{fig:H}
\end{figure}

After performing $\ComputeS(G)$ we received a table $M$ of shortcuts and computed the data structure $D=\ComputeF(M)$. Using $M$ and $D$, we construct the graph $H$. 
In the following sections we define the arcs of $H$, see Figure~\ref{fig:H}. After building $H$, we compute its transitive closure graph $H^\star$.

\subsubsection{\texorpdfstring{$0$-$0$}{0-0} arcs}
For every $x,y,z \in V$, add an arc $x^0 z^0$ to $E(H)$ if $M[x][y]+M[y][z]\ge 0$ and $M[x][y]\ge 0$. See Figure~\ref{fig:H}$(a)$.

\begin{lemma}\label{lemma:0-0}
    Let $x^0 z^0 \in E(H)$ , then $\alpha_0(x,z) \ge 0$. 
\end{lemma}

\begin{proof}
    Let $y\in V$ be such that $M[x][y]+M[y][z]\ge 0$ and $M[x][y]\ge 0$. The path $xyz$ in $G^M$ is strongly traversable. 
    By Lemma~\ref{lem:GM_to_G}, it follows that there is a strongly traversable path from~$x$ to~$z$ in~$G$.
\end{proof}

\subsubsection{\texorpdfstring{$0$-$B$}{0-B} arcs}
For every $x,y\in V$, add an arc $x^0 y^B$ to $E(H)$ if either $M[x][y]\ge B$, or $M[x][y]+M[y][x]>0$ and $M[x][y]>0$. See Figure~\ref{fig:H}$(b)$.

\begin{lemma}\label{lemma:0-B}
    Let $x^0 y^B\in E(H)$, then $\alpha_0(x,y) = B$.
\end{lemma}

\begin{proof}
    If $M[x][y] \ge B$ then by Lemma~\ref{lem:GM_to_G} there is a path from $x$ to $y$ in $G$ that satisfies $\alpha^G_0(P) \ge \alpha^{G^M}_0(xy)=B$.

    Assume $M[x][y] + M[y][x] >0$ and $M[x][y] > 0$. Note that the path $xyxy$ in $G^M$ is a strongly traversable ascending path from $x$ to $y$ of gain at strictly larger than $M[x][y]$. By extending this argument, it follows that for every $j>0$, $P=x(yx)^j y$ is strongly traversable ascending path from $x$ to $y$. Thus, there exists a $j>0$ such that $g^{G^M}(P)\ge B$. By Lemma~\ref{lemma:alpha-of-monotone}, we get that $\alpha^{G^M}_0(P)= B$. By Lemma~\ref{lem:GM_to_G}, there is also a path $P'$ from $x$ to $y$ in $G$ that satisfies $\alpha^G_0(P)=B$.
\end{proof}

\subsubsection{\texorpdfstring{$B$-$0$}{B-0} arcs}
For every $x,y\in V$, add an arc $x^B y^0$ to $E(H)$ if $M[x][y] \neq -\infty$. See Figure~\ref{fig:H}$(c)$.

\begin{lemma}\label{lemma:B-0}
    Let $x^B y^0 \in E(H)$, then $\alpha_B(x,y)\ge 0$. 
\end{lemma}

\begin{proof}
    By the design of $\ComputeS$, we have $M[x][y] \ge -B$. The proof follows by applying Lemma~\ref{lem:GM_to_G} on the traversable path $xy$ in $G^M$.
\end{proof}

\subsubsection{\texorpdfstring{$B$-$B$}{B-B} arcs}
For every $x,y,z\in V$, add an arc $x^B z^B$ to $E(H)$ if $M[x][y]+M[y][z]\ge 0$ and $M[y][z]\ge 0$. See Figure~\ref{fig:H}$(d)$. 

\begin{lemma}\label{lemma:B-B}
    Let $x^B z^B \in E(H)$, then  $\alpha_B(x,z) = B$. 
\end{lemma}

\begin{proof}
    By the definition of $\ComputeS$,  $M[x][y]\ge -B$ and therefore $\alpha^{G^M}_B(xyz)=B$. Therefore, by Lemma~\ref{lem:GM_to_G}, there is a path $P$ from $x$ to $z$ in $G$ that satisfies $\alpha^G_B(P)=B$.
\end{proof}

The following theorem is an immediate consequence of Lemmas~\ref{lemma:0-0}, \ref{lemma:0-B}, \ref{lemma:B-0} and \ref{lemma:B-B}.

\begin{theorem}\label{theorem:trasitive-graph}
    Let $x^{b_1} y^{b_2} \in E(H^\star)$,
    then  $\alpha_{b_1}(x,y)\ge b_2$.
\end{theorem}

\subsection{Transitive closure graph - correctness}\label{section:transitive-correctness}
In this appendix we show that for every $s,t\in V$ it holds that $\alpha_B(s,t)=B$ if and only if $s^B t^B \in E(H^\star)$, see Theorem~\ref{theorem:B-B-strong}. We begin by addressing entry-exit pairs on positive gain cycles, see Definition~\ref{D-entry-exit}.

\begin{lemma}\label{lemma:special-entry-exit}
    Let $C$ be a positive gain simple cycle in $G$. There exists an entry-exit pair $(x,y)$ in $C$ such that w.h.p.\ $x^0 y^B \in E(H)$.
\end{lemma}

\begin{proof}
    Let $(x',y')$ be an entry-exit of $C$. If $C$ is not strongly traversable from $x'$  then for every $y\in C$ such that $(x',y)$ is an entry-exit pair, it follows from the definition of an entry-exit pair that the simple path $P^{x'y}$ from $x'$ to $y$ through $C$ satisfies $\alpha_0(P^{x'y})=B$ and therefore, by Lemma~\ref{lemma:observation-monotone}, $P^{x'y}$ is ascending and $g(P^{x'y})\ge B$. Therefore, by Theorem~\ref{theorem:shortcut} it holds that w.h.p.\ $M[x'][y] \ge B$, so by definition, $x'^0 y^B \in E(H)$.
    
    Assume that $C$ is strongly traversable from $x'$ and consider the path $P$ from $x'$ to itself through $C$. Let $y\in P$ be the vertex of maximum gain on $P$. Observe that the path from $x'$ to $y$ on $C$ is ascending. Indeed the charge level cannot go below the initial charge at~$x'$ (which is zero) and the charge level at $y$ is maximum. Thus, by Theorem~\ref{theorem:shortcut} it holds w.h.p.\ that $M[x'][y] > 0$. If $y=x'$ then $M[x'][y]+M[y][x'] > 0 $ and therefore, by the definition of $E(H)$, $x'^0 y^B \in E(H)$. By Theorem~\ref{theorem:trasitive-graph} this means that $(x',y)=(x',x')$ is an entry-exit pair and we are done. 
    
    Otherwise, consider $P^{yx'}$, the simple path from $y$ to $x'$ through $C$, and let $x$ be the vertex of minimum gain in $P^{yx'}$, see Figure~\ref{fig:entry-exit-cycle}. By the choice of $x$, $P^{yx}$, the path from $y$ to $x$ through $C$, is descending. We now show that $P^{xy}=P^{xx'}|P^{x'y}$ is ascending. Since $x$ is of minimum gain in $P^{yx'}$, it follows that the gains of the vertices on $P^{xx'}$ are nonnegative. Moreover, since $P^{x'y}$ is ascending it follows that all gains on $P^{xy}$ are nonnegative. We are left to show that $y$ has maximum gain in $P^{xy}$. Since $P^{x'y}$ is ascending, it is enough to show that $(g^{P^{xx'}}_v = ) g^{P^{xy}}_v \le g^{P^{xy}}_{y}$ for every $v \in P^{xx'}$. Let $b = g^{P^{x'y}}_{y}$, it follows that $g^{P^{xy}}_{y} = g^{P^{xx'}}_{x'}+g^{P^{x'y}}_{y} \ge b$. We prove that $g^{P^{xx'}}_v \le b$ for every $v \in P^{xx'}$.
    By contradiction, assume there is $v\in P^{xx'}$ such that $g^{P^{xx'}}_v > b$. Since all gains of vertices in $P$ are nonnegative we get  that $g^{P}_v = g^{P}_{x} + g^{P^{xx'}}_v > b = g^{P}_{y}$, a contradiction to the definition of $y$.  
    
    By Theorem~\ref{theorem:shortcut}, w.h.p.\ $M[x][y]\ge g^{P^{xy}}_{y}$ and $M[y][x]\ge g^{P^{yx}}_{x}$. Thus,  $M[x][y]+M[y][x]\ge g^{P^{xy}}_{y}  + g^{P^{yx}}_{x} = g(C) > 0$, so by the definition of $H$, we get that ${x}^0 {y}^B \in E(H)$, so by Lemma~\ref{lemma:0-B}, $(x,y)$ is an entry-exit pair of $C$. 
\end{proof}

\begin{figure}
    \centering
\includegraphics[width=1.00\textwidth]{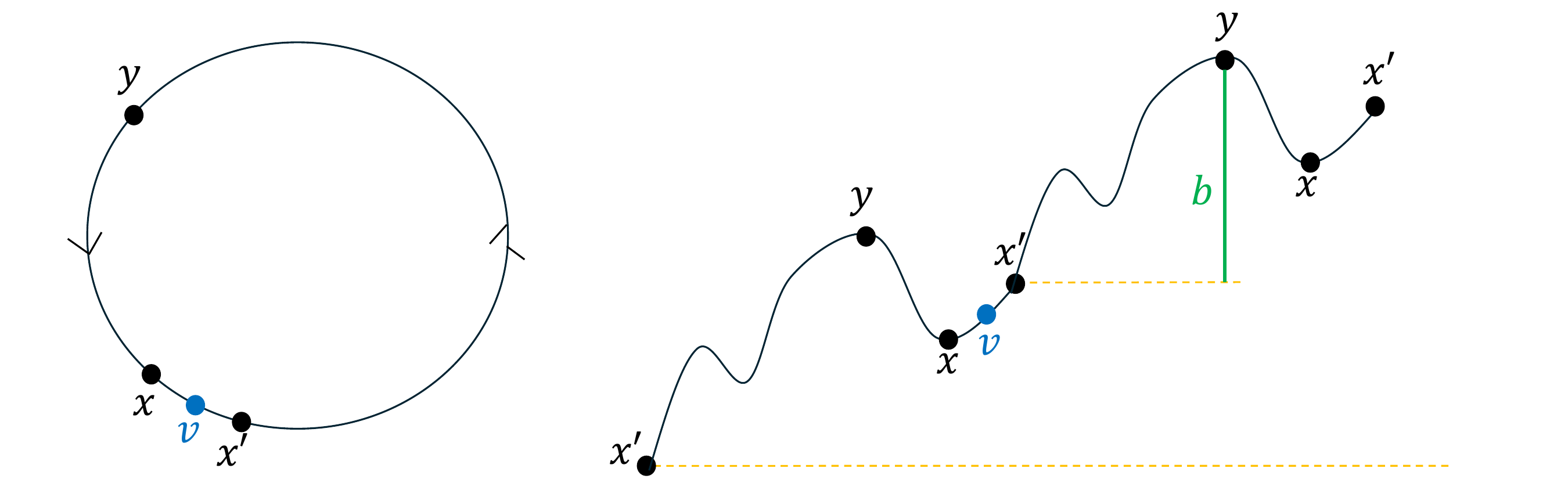}
\caption{Illustration of Lemma~\ref{lemma:special-entry-exit}. Note that $y$ is of maximum gain in the path from $x'$ to itself (through the cycle) and that $x$ is of minimum gain on the subpath from $y$ to $x'$. 
As shown in the proof of Lemma~\ref{lemma:special-entry-exit}, the path from $y$ to $x$ is descending and 
the path from $x$ to $y$ is ascending.}
\label{fig:entry-exit-cycle}
\end{figure}

\begin{lemma}\label{lemma:0-0-strong}
    Let $P$ be a strongly traversable simple path in $G$ from $x$ to $y$, then w.h.p.\ $x^0 y^0 \in E(H^\star)$.
\end{lemma}

\begin{proof}
    Denote $P= v_1\ldots v_k$ where $v_1=x$ and $v_2 = y$. Since $P$ is strongly traversable, $v_1$ has minimum gain in $P$. We decompose $P$ into monotone segments as follows, see Figure~\ref{fig:path-zig-zag-decomposition}. Let $i_1=1$ and let $i_1 < i_2 \le k$ be such that $v_{i_2}$ has maximum gain in $v_{i_1}\ldots v_k$. In particular, $v_{i_1} \ldots v_{i_2}$ is ascending.
    Let $i_2 < i_3 \le k$ be such that $v_{i_3}$ has the minimum gain in 
    $v_{i_2} \ldots v_k$. In particular, 
    $v_{i_2} \ldots v_{i_3}$ is descending. 
    In general, let $  i_{j-1} < i_j \le k$ be such that  $v_{i_{j-1}} \ldots v_{i_j}$ is ascending if $j$ is even and descending otherwise. Let $1=i_1,\ldots i_t=k$ be the indices we defined.
    
    We prove that 
    $g(v_{i_{2j-1}}\ldots v_{i_{2j+1}}) \ge 0$ for every $1\le j < t/2$. Indeed, if $j=1$, then since $P$ is strongly traversable, we get that $g(v_{i_{1}}\ldots v_{i_{3}}) \ge 0$. Let $1< j < t/2$. By the definition of $v_{i_{2j-1}}$, we get that $g(v_{i_{2j-2}}\ldots v_{i_{2j-1}}) \le g(v_{i_{2j-2}}\ldots v_{i_{2j+1}}) $ and therefore $g(v_{i_{2j-1}}\ldots v_{i_{2j+1}}) \ge 0$.
    
    By Theorem~\ref{theorem:shortcut}, for every $1\le j < t/2$, w.h.p.\ it holds that 
    $$M[v_{i_{2j-1}}][v_{i_{2j}}] + M[v_{i_{2j}}][v_{i_{2j+1}}] \ge
    g(v_{i_{2j-1}}\ldots v_{i_{2j}}) + g(v_{i_{2j}}\ldots v_{i_{2j+1}})
    =
    g(v_{i_{2j-1}}\ldots v_{i_{2j+1}}) \ge 0.$$
    Thus, by the definition of $E(H)$, for every $1\le j<t/2$ it holds that $v_{i_{2j-1}}^0 v_{i_{2j+1}}^0 \in  E(H)$. As for the last piece of $P$, if $t$ is even then $M[v_{i_{t-1}}][v_{i_t}] \ge 0$
    and $M[v_{i_{t-1}}][v_{i_t}] + M[v_{i_{t}}][v_{i_t}] \ge 0$ and therefore by the definition of $E(H)$, $v_{i_{t-1}}^0 v_{i_t}^0 \in E(H)$. 
    
    We conclude that since $H^\star$ is transitively closed, $x^0 y^0 = v^0_{i_1} v^0_{i_t} \in E(H^\star)$.
\end{proof}

\begin{figure}
    \centering
\includegraphics[width=1\textwidth]{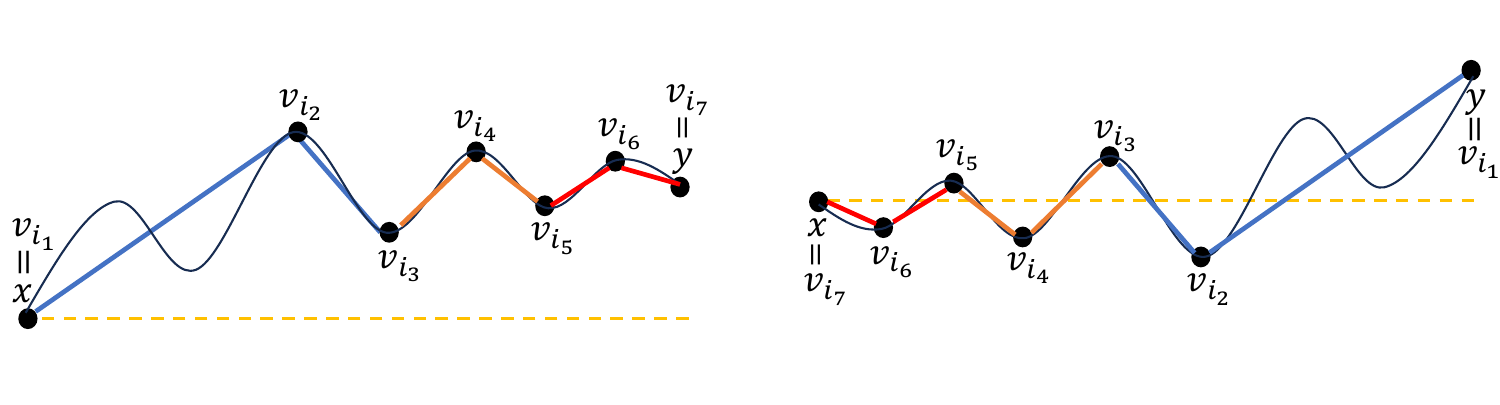}
\caption{Right: The path decomposition in Lemma~\ref{lemma:0-0-strong}. Note that we can pair ascending and descending paths and have overall nonnegative gain. Left: The path decomposition in Lemma~\ref{lemma:B-B-strong}. We pair descending paths with ascending paths such that the descending path comes first.}
\label{fig:path-zig-zag-decomposition}
\end{figure}

\begin{lemma}\label{lemma:B-B-strong}
Let $P$ be a simple path from $x$ to $y$ such that $\alpha_B(P)= B$, then w.h.p.\ $x^B y^B \in E(H^\star)$.    
\end{lemma}

\begin{proof}
    Denote $P = v_1 \ldots v_k$ where $v_1 = x$ and $v_k = y$.
    Note that $v_k$ has the largest gain in $P$ (since otherwise it cannot be reached with full charge). Similarly to Lemma~\ref{lemma:0-0-strong}, we decompose $P$ to monotone segments but this time we start the decomposition from $v_k$, see Figure~\ref{fig:path-zig-zag-decomposition}. Let $i_1=k$ and let $1\le i_2<i_1$ be such that $v_{i_2}$ has the minimum gain in $v_1\ldots v_{i_1}$. In particular, $v_{i_2} \ldots v_{i_1}$ is ascending. Let $1\le i_3<i_2$ be such that $v_{i_3}$ has the maximum gain in $v_1\ldots v_{i_2}$. In particular, $v_{i_3} \ldots v_{i_2}$ is descending.
    In general, let $1\le i_{j} < i_{j-1}$ be such that  $v_{i_{j}} \ldots v_{i_{j-1}}$ is ascending if $j$ is even and descending otherwise. Let $1=i_t,\ldots i_1=k$ be the indices we constructed.

    Similarly to Lemma~\ref{lemma:0-0-strong}, we prove that 
    $g(v_{i_{2j+1}}\ldots v_{i_{2j-1}}) \ge 0$  for every $1\le j < t/2$. Indeed, if $j=1$, then since $v_{i_1}=v_k$ has the largest gain in $P$, we get that $g(v_{i_{3}}\ldots v_{i_{1}}) \ge 0$. Let $1< j < t/2$. By the definition of $v_{i_{2j-1}}$, we get that $g(v_1\ldots v_{i_{2j-1}}) \ge g(v_1\ldots v_{i_{2j+1}}) $ and therefore $g(v_{i_{2j+1}}\ldots v_{i_{2j-1}}) \ge 0$.
    By Theorem~\ref{theorem:shortcut}, w.h.p.\ we get that 
    $$M[v_{i_{2j+1}}][v_{i_{2j}}] + M[v_{i_{2j}}][v_{i_{2j-1}}] 
    \ge g(v_{i_{2j+1}}\ldots v_{i_{2j}}) 
    +
    g(v_{i_{2j}}\ldots v_{i_{2j-1}})
    = g(v_{i_{2j+1}}\ldots v_{i_{2j-1}}) \ge 0.$$ 
    Thus, by the definition of $E(H)$, for every $1 \le j < t/2$, we get that $v_{i_{2j+1}}^B v_{i_{2j-1}}^B \in E(H)$. As for the last piece, note that if $t$ is even then $M[v_{i_t}][v_{i_{t-1}}] \ge 0$, so $M[v_{i_t}][v_{i_t}] + M[v_{i_t}][v_{i_{t-1}}] \ge 0$  and therefore by the definition of $E(H)$, $v_{i_t}^B v_{i_{t-1}}^B \in E(H)$. 
    
    Since $H^\star$ is transitively closed, $x^B y^B = v_{i_t}^B v_{i_1}^B \in E(H^\star)$. 
\end{proof}

The following theorem states that we have indeed found all entry-exit pairs.
\begin{lemma}\label{lemma:find-entry-exit}
    Let $(x,y)$ be an entry-exit pair in a positive gain cycle $C$, then w.h.p.\ $x^0 y^B$ is an arc in~$H^{\star}$.
\end{lemma}

\begin{proof}
    Let $P^{xy}$ be the simple path from $x$ to $y$ through $C$. We split into cases.

    \textbf{Case $1$: $\alpha_0(P^{xy}) = B$:}
    Therefore, by Lemma~\ref{lemma:observation-monotone}, $P^{xy}$ is ascending with gain at least $B$. Therefore, by Theorem~\ref{theorem:shortcut}, w.h.p., $M[x][y]\ge B$ and therefore w.h.p.\ $x^0 y^B \in E(H) \subseteq E(H^\star)$.
    
    \textbf{Case $2$: $\alpha_0(P^{xy}) < B$:}
    Thus, in order to start at $x$ with no charge and reach $y$ (through $C$) with full-charge the car must traverse $C$ at least once. Let $P'$ be such a path from $x$ to $y$ through $C$ such that $\alpha_0(P')=B$ (note that $P'$ must cycle $C$ at least once).
    By Lemma~\ref{lemma:special-entry-exit}, there exists an entry-exit pair $(x',y')$ on $C$ that satisfies $x'^0 y'^B \in E(H)$ and theretofore $x'^0 y'^B \in E(H^\star)$. Since $P'$ is strongly traversable and it cycles around $C$ at least once, it follows that the simple path from $x$ to~$x'$ (which is a prefix of $P'$) through $C$ is strongly traversable. So, by Lemma~\ref{lemma:0-0-strong}, it holds that $x^0 x'^0 \in E(H^\star)$. 
    Finally, since $y$ is an exit of $C$ and $y'$ lies on the same cycle $C$, it follows that $P^{y'y}$, the simple path from $y'$ to $y$  through $C$ satisfies $\alpha_B(P^{y'y}) = B$ (since otherwise, the exit $y$ cannot be reached with full charge from $y'$ and in particular $\alpha_0(P')<B$, a contradiction). Therefore, by Lemma~\ref{lemma:B-B-strong}, $y'^B y^B \in E(H^\star)$. Since $H^\star$ is transitively closed, we get $x^0 y^B \in E(H^\star)$.
\end{proof}

We are now ready to prove the main claim.

\begin{theorem}\label{theorem:0-B-strong}
    Let $s,t\in V$. If $\alpha_0(s,t) = B$ then w.h.p.\ $s^0 t^B \in E(H^\star)$.
\end{theorem}

\begin{proof}
    Let $P$ be a path from $s$ to $t$ of the form of Lemma~\ref{lemma:optimal-structure} and let $C_1, \ldots C_k$ and $(x_1,y_1),\ldots(x_k,y_k)$ as in Lemma~\ref{lemma:optimal-structure}. By Lemma~\ref{lemma:observation-monotone}, $P$ is ascending.

    If $P$ is simple (i.e., $k=0$) then by Theorem~\ref{theorem:shortcut} it holds w.h.p.\ that $M[s][t] \ge g(P) \ge B$ and therefore $s^0 t^B \in E(H) \subseteq E(H^\star)$.

    Otherwise, Since $P$ starts with a simple path from $s$ to $x_1$ then by Lemma~\ref{lemma:0-0-strong}, $s^0 x_1^0 \in E(H^\star)$.  By Lemma~\ref{lemma:find-entry-exit}, $x_i^0 y_i^B\in E(H^\star)$, for every $i\le k$. Let $i < k$, and consider $Q=u_1 \ldots u_t$, the simple subpath of $P$ from $y_i$ to $x_{i+1}$. Let $j$ be maximal such that $u_1 \ldots u_j$ is descending (and also traversable as a subpath of $P$). Since $Q$ is simple and traversable, we get by Theorem~\ref{theorem:shortcut} that $M[u_1][u_k] \neq -\infty$. By the definition of $E(H)$, we get that $y_i^B u_j^0 = u_1^B u_j^0 \in E(H)$. By the minimality of $u_j$, we get that $u_j \ldots u_t$ is strongly traversable and therefore by Lemma~\ref{lemma:0-0-strong} we get that w.h.p.\ $u_j^0 x^0_{i+1}=u^0_j u_t^0 \in E(H^\star)$. 
    Let $P^{y_k t}$ be the (simple) subpath of $P$ from $y_k$ to $t$. It holds that $\alpha_B(P^{y_k t}) = B$. Therefore, 
    by Lemma~\ref{lemma:B-B-strong}, w.h.p., $y_k^B t^B \in E(H^\star)$. Since $H^\star$ is transitively closed,  we get $s^0 t^B \in E(H^\star)$.
\end{proof}

\begin{theorem}\label{theorem:B-B-strong}
    Let $s,t\in V$. If $\alpha_B(s,t) = B$ then w.h.p.\ $s^B t^B \in E(H^\star)$.
\end{theorem}

\begin{proof}
    Let $P$ be a path from $s$ to $t$ of the form of Lemma~\ref{lemma:optimal-structure} and let $C_1, \ldots C_{\ell}$ and $(x_1,y_1),\ldots(x_{\ell},y_{\ell})$ as in Lemma~\ref{lemma:optimal-structure}.
    If $P$ is simple (i.e., $k=0$) then we are done by Lemma~\ref{lemma:B-B-strong}. Assume otherwise, and let $P^{s x_1}$ be the simple subpath from $s$ to $x_1$. 
    By Theorem~\ref{theorem:0-B-strong}, we get that $x_1^0 t^B \in E(H^\star)$. Since $H^\star$ is transitively closed, it is enough to prove that $s^B x_1^0 \in E(H^\star)$. 

    Denote $P^{s x_1} = v_1\ldots v_k$ and let $i$ be maximal such that $\alpha_B(s,v_i)= B$. By Lemma~\ref{lemma:B-B-strong}, it holds that $s v_i^B \in E(H^\star)$. Denote $P^{v_i x_1} = v_i\ldots v_k$ and let  $v_j$ be the vertex of smallest gain in $P^{v_i x_1}$. By the definition of $v_i$, we get that $v_i$ has the largest gain in $P^{s x_1}$. In particular, $v_i$ has the largest gain in $P^{v_i v_j} = v_i \ldots v_j$, so $P^{v_i v_j}$ is descending. By Theorem~\ref{theorem:shortcut}, we get that w.h.p.\ $M[v_i][v_j]\ge g(P^{v_i v_j})(\ge -B)$ and therefore $v_i^B v_j^0 \in E(H)$. Since $P^{v_i x_1}$ is traversable and $v_j$ has the minimum gain in $P^{v_i v_j}$, we get by Lemma~\ref{lem:alpha_general} that $P^{v_i v_j}$ is strongly traversable. Therefore, by Lemma~\ref{lemma:0-0-strong}, we get that $v_j^0 x_1^0 \in E(H^\star)$. Since $H^\star$ is transitively closed, we get that $s^B x_1^B \in E(H^\star)$.
\end{proof}

By combining Theorem~\ref{theorem:B-B-strong} with Theorem~\ref{theorem:trasitive-graph} we get the following theorem.

\begin{theorem}\label{theorem:B-B-strong-full}
    For every $s,t\in V$, w.h.p., $\alpha_B(s,t)=B$ if and only if $s^B t^B \in E(H^\star)$.
\end{theorem}

\subsection{Computing the \texorpdfstring{$\alpha_B(\cdot, \cdot)$}{alphaB function} values}\label{section:compute-alpha}

The algorithm $MFC(M)$ for deriving of the $\alpha_B(\cdot,\cdot)$ values is given in Figure~\ref{alg:alpha_B}.
The algorithm computes a table $\alpha_B[\cdot][\cdot]$ and we prove in Theorem~\ref{theorem:alpha-correctness} that  $\alpha_B[s][t] = \alpha_B(s,t)$, for every $s,t\in V$. Algorithm $MFC(M)$ starts by computing $H^\star$ as explained in Appendix~\ref{section:transitive-graph}.

The algorithm is based on the following idea. For $s,t \in V$, let~$P=v_1(=s)\ldots v_k(=t)$ be an optimal path from $s$ to $t$ (i.e., $\alpha_B(P)=\alpha_B(s,t)$) that follows the structure of Lemma~\ref{lemma:optimal-structure} and let $C_1, \ldots C_{\ell}$ and $(x_1,y_1),\ldots(x_{\ell},y_{\ell})$ as in Lemma~\ref{lemma:optimal-structure}.
Let $1\le i \le k$ be the maximum index that satisfies $\alpha_B(v_1 \ldots v_i) = B$. By Theorem~\ref{theorem:B-B-strong-full}, w.h.p.\ $s^B v_i^B \in E(H^\star)$. 
By Lemma~\ref{lemma:optimal-structure} it holds that $\alpha_B(v_1,y_{\ell})=B$ and therefore $v_i \in P^{y_{\ell}v_k}$, where $P^{y_{\ell}v_k}$ is the (simple) subpath of $P$ from $y_{\ell}$ to $t$. Hence $v_i \ldots v_k$ is a simple path.

We prove in Lemma~\ref{lemma:alpha-simple-B-0} that there exists $x\in V$ that satisfies $M[v_i][x]\in [-B,0]$ and $B+D[v_i x][v_k] = \alpha_B(v_i,v_k)(=\alpha_B(v_1,v_k))$, see Figure~\ref{fig:alpha-simple-B-0} where $y=v_i, x=v_{i_2},v_k=t$.

Based on the above, the algorithm proceeds as follows. For every $y,t\in V$, we upper bound the largest final charge we can get if we use a simple path $P$ that starts at $y$ with full charge and ends at $t$ such that~$y$ has the maximum gain in $P$. We store these values in a table $A_B$ whose computation is done by assigning $A_B[y][t] \gets \max \{B+D[yx][t] \mid x\in V, \; M[y][x] \le 0 \} $, for every $y,t\in V$. Finally, the computation of  $\alpha_B[s][t]$ is done by assigning $\alpha_B[s][t] \gets \max \{ A_B[y][t] \mid s^B y^B \in E(H^\star)\}$.

The following lemma states that the $A_B[\cdot][\cdot]$ values $MFC(M)$ computes lower bound the actual $\alpha_B(\cdot,\cdot)$ values.

\begin{lemma}\label{lem:alphaB-bounds-D}
    Let $M$ be the shortcut table returned by $\ComputeS$. Let $D=\ComputeF(M)$ and let $\alpha_B[\cdot][\cdot]$ be the result of $MFC(M)$. Then $\alpha_B(y,t) \ge B +  D[yx][t]$ for every $y,x,t\in V$ that satisfy $-B \le M[y][x] \le 0$. In particular, $A_B[y][t] \le \alpha_B(y,t)$ for every $y,t\in V$.
\end{lemma}

\begin{proof}
    Let $y,x,t\in V$ be as in the statement of the lemma. Let $P=P^{yx}\mid P^{xt}$
    (a traversable path in $G$) and $C= C^{yx}\mid C^{xt}$ be as in Corollary~\ref{cor:M-G-connection}~\ref{cor:1}\ref{cor:1a}-\ref{cor:1c} when applied on $D[yx][t]$. Denote $P = v_1\ldots v_k$. Since $-B \le M[y][x]\le 0$, it follows by Corollary~\ref{cor:M-G-connection}~\ref{cor:1}\ref{cor:1c} that   $-B\le g^{P,C}_x \le g^{P,C}_{v_i}\le g^{P,C}_y= 0$ for every $1\le i \le k$. We prove by induction on $i=1,\ldots, k$ that $\alpha_B(v_1 \ldots v_i) \ge B +  g^{P,C}_{v_i}$ and therefore $\alpha_B(P) \ge B +  g^C(P)=B + D[yx][t]$. 

    The base of induction holds since 
    $\alpha_B(v_1)=B= B+ g^{P,C}_{v_1}$. Let $i>1$, since $P$ is traversable it holds that $\alpha_B(v_1\ldots v_{i+1})\ge 0$ and therefore 
    \begin{align*}
    \alpha_B(v_1\ldots v_{i+1}) 
    &= \min \{B , \alpha_B(v_1 \ldots v_i) + g(v_i v_{i+1}) \}
     \numge{1} 
    \min \{B , B + g^{P,C}_{v_i} + g(v_i v_{i+1}) \} \\
    &\ge  \min \{B , B + g^{P,C}_{v_{i+1}}  \} 
    \numeq{2} B +  g^{P,C}_{v_{i+1}},
    \end{align*}
     where Inequality~$(1)$ holds by the inductive hypothesis and Equality~$(2)$ holds since we showed that $g^{P,C}_{v_i}\le 0$ for every $1\le i \le k$.
\end{proof}

As a corollary, we get that the $\alpha_B[\cdot][\cdot]$ values that $MFC(M)$ computes, lower bound the actual $\alpha_B(\cdot,\cdot)$ values.

\begin{corollary}\label{cor:A_lowerbounds_alpha}
    For every $s,t\in V$ it holds that $\alpha_B[s][t] \le \alpha_B(s,t)$.
\end{corollary}

\begin{proof}
    Let $s,t\in V$. If $\alpha_B[s][t] = -\infty$ or $\alpha_B(s,t)=B$ then we are done. Assume otherwise. Since $\alpha_B(s,t) < B$, then by Theorem~\ref{theorem:trasitive-graph} 
    $s^B t^B \notin E(H^\star)$. Moreover, since $\alpha_B[s][t] \neq -\infty$, there is $(t\neq )y\in V$ such that $s^B y^B \in E(H^\star)$ and $\alpha_B[s][t] = A_B[y][t]$. 
    Let $x\in V$ be such that $-B\le M[y][x]\le 0$ and $A_B[y][t] = B + D[yx][t]$. We conclude that $\alpha_B(s,t) \numge{1} \alpha_B(y,t) \numge{2} A_B[y][t] = \alpha_B[s][t]$, where Inequality~$(1)$ holds by Theorem~\ref{theorem:trasitive-graph} (recall that $s^B y^B \in E(H^\star)$) and Inequality~$(2)$ holds by Lemma~\ref{lem:alphaB-bounds-D}  
\end{proof}

\begin{figure}[t]
\begin{algorithm}[H]
\Fn{$MFC(M)$}{ 
    \BlankLine
    $H \gets Build\_H(M)$ \tcp*{As explained in Appendix~\ref{section:transitive-graph}}
    $H^\star \gets Transitive\_closure(H)$\;
    $D \gets \ComputeF(M)$\;
    \BlankLine
    $A_{B}\gets matrix(n,n, -\infty)$ \tcp*{$\alpha_B(\cdot,\cdot)$ of simple bounded paths starting with $B$ charge}
    \For{$y,t\in V$}{ 
        \For{$x\in V$}{
            \If(\tcp*[h]{$\overunderline{yxt}{1-1}{2-2}$ paths}){$M[y][x] \le 0$}{ 
                $A_{B}[y][t] \gets \max \{A_B[y][t], B + D[yx][t] \}$
            }
        }
    }
    \BlankLine
    $\alpha_B \gets matrix(n,n,-\infty)$\;
    \For{$s,t\in V$}{
        \If{$s^B t^B\in E(H^\star)$}{
            $\alpha_B[s][t] \gets B$    
        }
        \For{$y\in V$}{
            \If{$s^B y^B\in E(H^\star)$}{
                $\alpha_B[s][t] \gets \max \{\alpha_B[s][t] , A_{B}[y][t]\}$
            }
        }
    }
    \BlankLine
    \Return $\alpha$
}
\end{algorithm}
\caption{Computing the \emph{maximum final charges} $\alpha_B(s,t)$ for every $s,t\in V$.}\label{alg:alpha_B}
\end{figure}

\begin{figure}
    \centering
\includegraphics[width=0.5\textwidth]{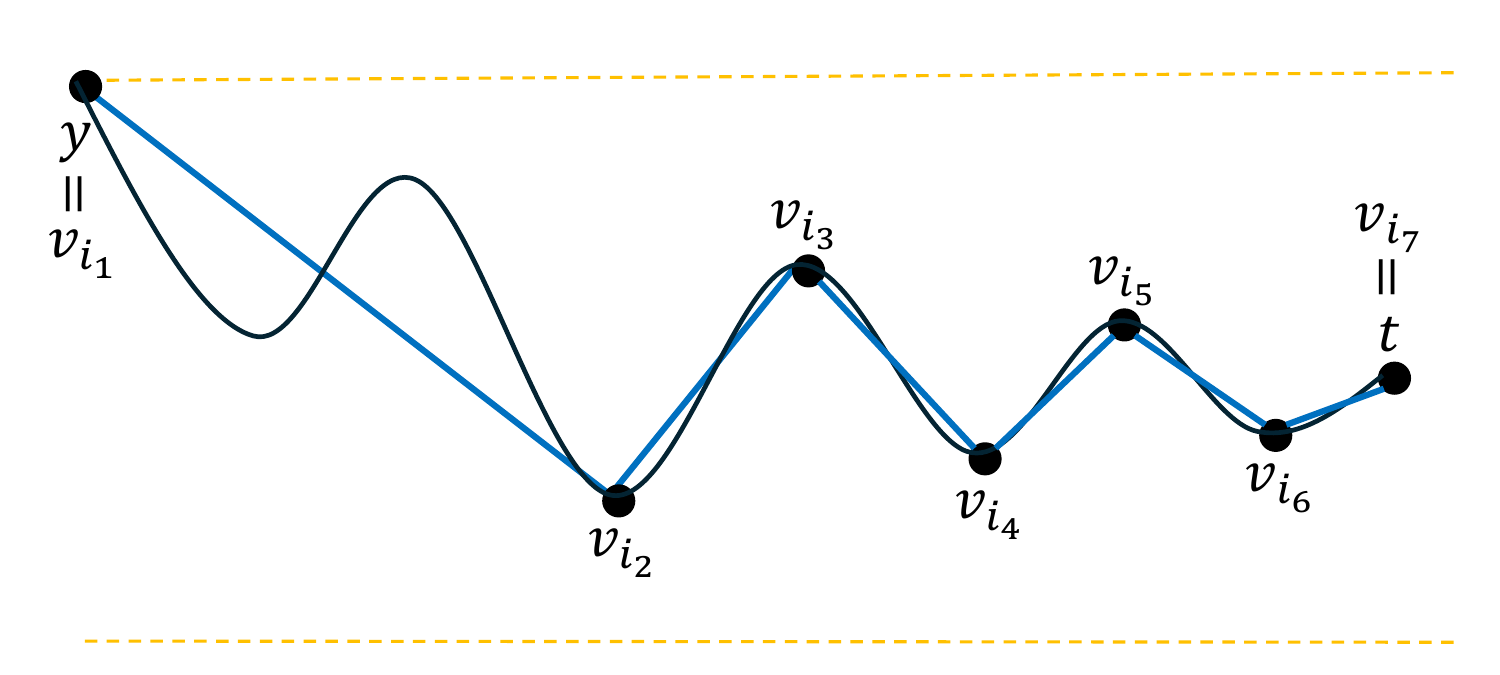}
\caption{The path decomposition in Lemma~\ref{lemma:alpha-simple-B-0}. The blue arcs correspond to arcs in $G^M$ of \textbf{the same} gain as the subpaths.}
\label{fig:alpha-simple-B-0}
\end{figure}

\begin{lemma}\label{lemma:alpha-simple-B-0}
    Let $y,t\in V$. If there is a simple traversable path $P$ from $y$ to $t$ such that $\alpha_B(P) = \alpha_B(y,t)$ and $g_v < g_y = 0$ for every $(y \neq) v \in P$, then  $\alpha_B(y,t) = A_B[y][t]$. 
\end{lemma}

\begin{proof}
    By Lemma~\ref{lem:alpha_general} and by the assumption, we get that $\alpha_B(P) = B+g(P)$.
    
    Denote $P= v_1 (=y)\ldots v_k (=t)$. We decompose $P$ into monotone subpaths as follows (see Figure~\ref{fig:alpha-simple-B-0}).
    Let $i_1=1$ and let $v_{i_2}$, where $i_1 < i_2 \le k$, be the last vertex of minimum gain in $v_{i_1}\ldots v_k$. Since  $g_v \le g_y$ for every $ v \in P$, we get that $v_{i_1} \ldots v_{i_2}$ is descending. 
    Let $v_{i_3}$, where $i_2<i_3 \le k$, be the last vertex of maximum gain in $v_{i_2} \ldots v_{k}$. In particular, $v_{i_2} \ldots v_{i_3}$ is ascending.
    In general, let $v_{i_j}$, where $i_{j-1} < i_j\le k$ be the last vertex of maximum gain in $v_{i_{j-1}} \ldots v_{k}$ if $j$ is odd and and the last vertex of minimum gain in $v_{i_{j-1}} \ldots v_{k}$ if $j$ is even. We get that if $j$ is even then $v_{i_{j-1}} \ldots v_{i_j}$ is descending and otherwise ascending. Let $i_1(=1),\ldots i_t=k$ be the indices we constructed. Let $P_j = v_{i_{j-1}} \ldots v_{i_j}$ for $j=2,\ldots,t$. It follows by the construction and from the assumption that $g_v < g_y=0$ for every $(y\neq) v\in P$, that
    \begin{enumerate}
        \item $|g(P_2)| > |g(P_3)| > \ldots >|g(P_t)|$.
        \item $sign(g(P_{i-1})) = - sign(g(P_{i}))$ for $i=2,\ldots t$.
    \end{enumerate}

    Since $P$ is simple, it follows by Theorem~\ref{theorem:shortcut} that w.h.p.\ $M[v_{i_{j-1}}][v_{i_{j}}]\ge g(v_{i_{j-1}} \ldots v_{i_{j}})$ for every $2\le j\le t$. Note that actually $M[v_{i_{j-1}}][v_{i_{j}}] =  g(v_{i_{j-1}} \ldots v_{i_{j}})$. Otherwise, since $g_v < 0$ for every $(y\neq )v\in P$, we can improve $P$ by constructing a path $P'$ from $P$ by replacing a subpath $v_{i_{j-1}} \ldots v_{i_{j}}$ by a better subpath that corresponds to $M[v_{i_{j-1}}][v_{i_{j}}]$ by Lemma~\ref{lem:GM_to_G}. This yields $\alpha_B(P') > \alpha_B(P)$, a contradiction to the optimality of $P$: $\alpha_B(P) = \alpha_B(y,t)$. Therefore, by Lemma~\ref{lemma:funnel-zigzag-structure}, $v_{i_1} v_{i_2},\ldots v_{i_t}$ is a funnel in $G^M$, so by Lemma~\ref{lemma:funnels-optimality} we get w.h.p.\ that $D[y v_{i_1}][t] = D[v_{i_1} v_{i_1}][v_{i_t}] \ge g(P) = \alpha_B(y,t) - B$. Therefore, by the definition of $A_B$ in $MFC(M)$.
    \[\alpha_B(y,t) \le 
    B + D[y v_{i_2}][t] \le A_B[y][t].
    \]
    Thus, by Lemma~\ref{lem:alphaB-bounds-D} it follows that
     $A_B[y][t] =  \alpha_B(y,t)$.
\end{proof}

\begin{theorem}\label{theorem:alpha-correctness}
    Algorithm $MFC(M)$ computes $\alpha_B(s,t)$ for every $s,t\in V$.
\end{theorem}

\begin{proof}
    Let $s,t\in V$. We prove that $\alpha_B[s][t] =\alpha_B(s,t)$, where $\alpha_B[\cdot][\cdot]$ is the table used in $MFC(M)$, see Figure~\ref{alg:alpha_B}.
    Let $P$ be a traversable path from $s$ to $t$ of the form of Lemma~\ref{lemma:optimal-structure} and let $C_1, \ldots C_k$ and $(x_1,y_1),\ldots(x_{\ell},y_{\ell})$ be as in Lemma~\ref{lemma:optimal-structure} and let $P^{y_{\ell} t}$ be the simple subpath from $y_{\ell}$ to $t$.

    Let $y\in P$ be the last vertex in $P$ that satisfies $\alpha_B(s,y) = B$. By the definition of the decomposition, $\alpha_B(s,y_{\ell}) = B$ and therefore $y$ 
    is on the simple path $P^{y_\ell t}$.
    By Theorem~\ref{theorem:B-B-strong-full}, we get w.h.p.\ that $s^B y^B \in E(H^\star)$.

    Let $P^{y t}$ be the simple subpath of $P^{y_{\ell} t}$ from $y$ to $t$. Since $\alpha_B(P)= \alpha_B(s,t)$ and $\alpha_B(s,y)= B$  it follows that $ \alpha_B(s,t) = \alpha_B(P^{y t}) = \alpha_B(y,t)$.
    By the definition of $y$, it holds that $g^{P^{y t}}_v < g^{P^{y t}}_y = 0$ for every $(y\neq) v \in P^{y t}$. Therefore, by Lemma~\ref{lemma:alpha-simple-B-0}, we get that $\alpha_B(y,t) = A_B[y][t]$. 
    
    By the definition of algorithm $MFC(M)$, since $s^B y^B \in E(H^\star)$, we get that $\alpha_B[s][t] \ge  A_B[y][t] = \alpha_B(s,t)$. 
    On the other hand, by Corollary~\ref{cor:A_lowerbounds_alpha} we have that $\alpha_B[s][t] \le \alpha_B(s,t)$, so we are done.
\end{proof}

The following theorem summarise the main result of this paper.

\begin{theorem}
    Let $G = (V,A,g)$ be a road network that may contain positive gain cycles and let $B \in \RR^{+}$. There is a randomized algorithm that in expected $\tilde{O}(n^{3.5})$ time computes a table $\alpha_B[\cdot][\cdot]$ such w.h.p.\ $\alpha_B[s][t] = \alpha_B(s,t)$ for every $s,t\in V$.
\end{theorem}

\end{document}